\documentclass[11pt]{article}
\usepackage[letterpaper,margin=1in]{geometry}
\usepackage{amssymb}

\usepackage[utf8]{inputenc} 
\usepackage[T1]{fontenc}    
\usepackage{url}            
\usepackage{booktabs}       
\usepackage{amsfonts}       
\usepackage{nicefrac}       
\usepackage{microtype}      
\usepackage{xcolor}         

\usepackage[right, outer, inline, final]{showlabels} 
\usepackage{nicematrix}

\usepackage{tikz}
\usepackage{float}
\usetikzlibrary{shapes,backgrounds,patterns,calc,arrows,arrows.meta}
\usepackage{chemfig}

\catcode`\_=11
\definearrow5{-x>}{%
    \CF_arrowshiftnodes{#3}%
    \CF_expafter{\draw[}\CF_arrowcurrentstyle](\CF_arrowstartnode)--(\CF_arrowendnode)%
        coordinate[midway,shift=(\CF_ifempty{#4}{225}{#4}+\CF_arrowcurrentangle:\CF_ifempty{#5}{5pt}{#5})](line@start)%
        coordinate[midway,shift=(\CF_ifempty{#4}{45}{#4}+\CF_arrowcurrentangle:\CF_ifempty{#5}{5pt}{#5})](line@end)%
        coordinate[midway,shift=(\CF_ifempty{#4}{135}{#4}+\CF_arrowcurrentangle:\CF_ifempty{#5}{5pt}{#5})](line@start@i)%
        coordinate[midway,shift=(\CF_ifempty{#4}{315}{#4}+\CF_arrowcurrentangle:\CF_ifempty{#5}{5pt}{#5})](line@end@i);
    \draw(line@start)--(line@end);%
    \draw(line@start@i)--(line@end@i);%
    \CF_arrowdisplaylabel{#1}{0.5}+\CF_arrowstartnode{#2}{0.5}-\CF_arrowendnode
}
\catcode`\_=8

\usepackage{pict2e,picture,graphicx}
\usepackage[overload]{empheq}
\usepackage{microtype}
\usepackage{graphicx}
\usepackage{subcaption}
\usepackage{booktabs}
\usepackage{amsmath}
\usepackage{cases}
\usepackage{mathtools}
\usepackage{amsthm}
\usepackage{thm-restate}
\usepackage{enumitem}
\allowdisplaybreaks
\usepackage{xcolor}
\usepackage{nicefrac, xfrac}
\usepackage{xparse}
\usepackage{wrapfig}
\usepackage{bbm}
\usepackage{nicefrac} 
\usepackage{multirow}
\usepackage{bm}
\usepackage{bbm}
\usepackage{subcaption}

\usepackage{hyperref} 

\hypersetup{
	colorlinks = true,
	linkcolor = royalBlue,
	citecolor = darkGreen,
	linktocpage = true,
	urlcolor = mgreen
}

\usepackage{comment}

\usepackage{colortbl}
\usepackage{cellspace}

\usepackage[capitalize, noabbrev]{cleveref}

\colorlet{darkgreen}{green!70!black}

\newcommand{\cA}{\mathcal{A}}

\newcommand{\poly}{\textnormal{poly}}

\usepackage{color}
\definecolor{mygreen}{rgb}{0.0, 0.5, 0.0}
\definecolor{myorange}{rgb}{0.55, 0.62, 1}
\newcommand{\nb}[3]{{\colorbox{#2}{\bfseries\sffamily\scriptsize\textcolor{white}{#1}}}{\textcolor{#2}{\sf\small\textsf{#3}}}}

\theoremstyle{plain}

\newcommand{\Fcal}{\mathcal{F}}
\newcommand{\Acal}{\mathcal{A}}
\theoremstyle{plain}
\newtheorem{theorem}{Theorem}[section]

\newtheorem{lemma}[theorem]{Lemma}

\newtheorem{corollary}[theorem]{Corollary}
\theoremstyle{definition}
\newtheorem{definition}[theorem]{Definition}

\theoremstyle{remark}

\definecolor{niceRed}{RGB}{190,38,38}
\definecolor{Red2}{RGB}{219, 50, 54}
\definecolor{mgreen}{RGB}{160, 200, 140}
\definecolor{blueGrotto}{RGB}{5,157,192}
\definecolor{limeGreen}{HTML}{81B622}
\definecolor{myellow}{rgb}{0.88,0.61,0.14}
\definecolor{darkGreen}{HTML}{2E8B57}
\definecolor{navyBlueP}{HTML}{03468F}
\definecolor{Sepia}{HTML}{7F462C}
\definecolor{red2}{HTML}{1F462C}
\definecolor{orange2}{HTML}{FF8000}
\definecolor{mgray}{HTML}{ABB3B8}
\definecolor{lgray}{HTML}{E5E8E9}
\definecolor{myPurple}{RGB}{175,0,124}
\definecolor{mypurple2}{rgb}{0.8,0.62,1}
\definecolor{royalBlue}{HTML}{057DCD}
\definecolor{mpink}{HTML}{FC6C85}
\definecolor{lblue}{RGB}{74,144,226}
\definecolor{peagreen}{RGB}{152,193,39}
\definecolor{typ_navy}{HTML}{001f3f}
\definecolor{typ_blue}{HTML}{0074d9}
\definecolor{typ_aqua}{HTML}{7fdbff}
\definecolor{typ_teal}{HTML}{39cccc}
\definecolor{typ_eastern}{HTML}{239dad}
\definecolor{typ_purple}{HTML}{b10dc9}
\definecolor{typ_fuchsia}{HTML}{f012be}
\definecolor{typ_maroon}{HTML}{85144b}
\definecolor{typ_red}{HTML}{ff4136}
\definecolor{typ_orange}{HTML}{ff851b}
\definecolor{typ_yellow}{HTML}{ffdc00}
\definecolor{typ_olive}{HTML}{3d9970}
\definecolor{typ_green}{HTML}{2ecc40}
\definecolor{typ_lime}{HTML}{01ff70}
\definecolor{newgreen}{HTML}{83c702}
\definecolor{newpurp}{RGB}{97,96,121}
\usepackage[style=alphabetic,natbib=true,maxcitenames=2, maxbibnames=10,backend=bibtex]{biblatex}
\usepackage{xspace}
\usepackage{algpseudocode}
\usepackage{algorithm}

\usepackage{hyperref}

\usepackage{xcolor}

\newcommand{\anna}[1]{\nb{Anna}{darkgreen}{#1}}

\newcommand{\mat}[1]{\nb{Matteo}{blue}{#1}}

\newcommand{\BigOL}[1]{\tilde{\mathcal{O}}\left(#1\right)}

\newcommand{\rev}{\textsf{Rev}\xspace}
\newcommand{\gft}{\textsf{GFT}\xspace}

\newcommand{\cW}{\mathcal{W}}

\newcommand{\defeq}{\coloneqq}
\usepackage{comment}
\usetikzlibrary{external}

\allowdisplaybreaks

\usepackage{environ}
\definecolor{lightblue}{rgb}{0.8, 0.9, 1}
\definecolor{lightred}{rgb}{1, 0.8, 0.8}
\definecolor{lightgreen}{rgb}{0.8, 1, 0.8}

\usetikzlibrary{tikzmark}
\usetikzlibrary{calc}

\title{Better Regret Rates in Bilateral Trade via Sublinear Budget Violation} 

\author{
    Anna Lunghi \quad 
    Matteo Castiglioni \quad
    Alberto Marchesi \vspace{6mm}\\
    Politecnico di Milano, \vspace{2mm}\ \\
    \texttt{\{name.surname\}}@polimi.it
}

\date{}

\begin{document}

\maketitle
\begin{abstract}
    Bilateral trade is a central problem in algorithmic economics, and recent work has explored how to design trading mechanisms using no-regret learning algorithms. However, no-regret learning is impossible when budget balance has to be enforced at each time step. \citet{bernasconi2024no} show how this impossibility can be circumvented by relaxing the budget balance constraint to hold only globally over all time steps. In particular, they design an algorithm achieving regret of the order of $\tilde O(T^{3/4})$ and provide a lower bound of $\Omega(T^{5/7})$.

    In this work, we interpolate between these two extremes by studying how the optimal regret rate varies with the allowed violation of the global budget balance constraint. Specifically, we design an algorithm that, by violating the constraint by at most $T^{\beta}$ for any given $\beta \in [\nicefrac 3 4, \nicefrac 6 7]$, attains regret $\tilde O(T^{1 - \beta/3})$. We complement this result with a matching lower bound, thus fully characterizing the trade-off between regret and budget violation.
    Our results show that both the $\tilde O(T^{3/4})$ upper bound in the global budget balance case and the $\Omega(T^{5/7})$ lower bound under unconstrained budget balance violation obtained by \citet{bernasconi2024no} are tight.
\end{abstract}

\pagenumbering{gobble} 

\clearpage
\tableofcontents
\renewcommand{\thmtformatoptarg}[1]{~#1}
\newpage
\pagenumbering{arabic}
\section{Introduction}

Bilateral trade is a classic economic problem where a broker faces two rational agents---a \emph{seller} and a \emph{buyer}---who wish to trade an object. Each agent has their own private valuation for the object, which is unknown to the broker. The goal of the broker is to design a mechanism that intermediates between the seller and the buyer, in order to make a trade happen.
%
%
Ideally, a mechanism should be ex-post \emph{efficient}, namely, it should maximize social welfare. In bilateral trade, this amounts to trading every time the buyer's valuation for the object is higher than the seller's one. Unfortunately, a well-known impossibility result by~\citet{myerson1983efficient} establishes that \emph{no economically viable} mechanism can ensure ex-post efficiency. Specifically, economic viability means that: (i) no agent has an incentive to behave untruthfully (\emph{incentive compatibility}), (ii) every agent gets nonnegative utility by participating in the mechanism (\emph{individual rationality}), and (iii) the mechanism does \emph{not} subsidize the market (\emph{budget balance}). The latter requirement is especially relevant in bilateral trade problems, as it ensures that the broker does \emph{not} incur in financial losses. 

A recent line of research (see, \emph{e.g.},~\citep{cesa2021regret,cesa2024bilateral,azar2022alpha}) aims at circumventing the impossibility result by~\citet{myerson1983efficient} by addressing \emph{repeated} bilateral trade problems, where the broker faces a sequence of sellers and buyers over a finite time horizon $T$. In a repeated bilateral trade problem, at each time $t$, a new seller and a new buyer arrive, with private valuations $s_t$ and $b_t$, respectively, for the object. The broker implements a mechanism that defines a pair of (possibly randomized) prices $(p_t,q_t)$: the price $p_t$ is proposed to the seller, while $q_t$ to the buyer.\footnote{Notice that it can be shown that \emph{fixed-price} mechanisms that posts a (possibly randomized) pair of prices $(p,q)$ are the only one that satisfy: ex-post individual rationality, dominant-strategy incentive compatibility, and budget balance~\citep{hagerty1987robust}. Thus, it is common in the literature on repeated bilateral trade to assume that the broker implements a fixed-price mechanism $(p_t,q_t)$ at every time $t$.} Then, a trade happens only if both agents are willing to trade at the proposed prices given their private valuations, namely, whenever $s_t \leq p_t$ and $b_t \geq q_t$. In the case in which a trade occurs, the broker charges the price $q_t$ to the buyer and pays $p_t$ to the seller.

Different repeated bilateral trade settings can be considered, depending on how the agents' valuations are selected and on which type of feedback the broker gets at the end of each time step. As for agents' valuations, two cases are usually addressed: the \emph{adversarial} one in which the valuations are arbitrarily chosen by an adversary, and the \emph{stochastic} case, where the valuations are i.i.d.~samples from some probability distributions unknown to the broker. As for the feedback received by the broker, three different types are usually considered in the literature:
\begin{itemize}
    \item \emph{full feedback}, which reveals the valuations $s_t$ and $b_t$ to the broker at the end of very time step $t$;
    \item \emph{two-bit feedback}, under which the broker gets to know agents' decisions separately, \emph{i.e.}, the broker observes the values of $\mathbb{I} \{ s_t \le p_t \}$ and $\mathbb{I} \{ b_t \ge q_t \}$ at the end of every time step $t$;
    \item \emph{one-bit feedback}, which  only reveals whether a trade occurred or not, \emph{i.e.}, $\mathbb{I} \{ s_t \leq p, q \leq b_t \}$.
\end{itemize}
The last two types of feedback above have been the main focus of most of the works on repeated bilateral trade, as they constitute the most realistic feedback models.

The broker's performance is usually measured in terms of \emph{gain from trade} (GFT). At time $t$, the GFT for a pair of prices $(p,q)$ is defined as $\gft_t(p,q) \coloneqq (b_t - s_t) \mathbb{I} \{ s_t \leq p, q \leq b_t \}$, which intuitively encodes the net gain in social welfare in the case in which a trade occurs.\footnote{Notice that measuring the performance of a mechanism $(p,q)$ with respect to the gain from trade $\gft_t(p,q)$ or the social welfare $b_t \cdot \mathbb{I} \{ s_t \leq p, q \leq b_t \} + s_t \cdot (1-\mathbb{I} \{ s_t \leq p, q \leq b_t \})$ is indifferent, as the two quantities only differ by the term $s_t$, which is independent of the mechanism.}
%
%
Then, over the whole time horizon $T$, the broker is evaluated in terms of the (cumulative) \emph{regret} incurred by them. This is defined as $R_T \coloneqq \max_{p \in [0,1]} \sum_{t=1}^T \gft_t(p,p) - \sum_{t=1}^T \gft_t(p_t,q_t)$, which is the difference between what the broker could have obtained by means of the \emph{best-in-hindsight} mechanism and what they have actually obtained.
The goal of the broker is to guarantee that the regret grows sublinearly in the time horizon $T$, while ensuring some form of budget balance during learning. Indeed, in the repeated bilateral trade literature, three different budget balance notions have been considered:\footnote{Notice that the benchmark mechanism in the definition of the regret $R_T$ is a fixed-price mechanism that is SBB, \emph{i.e.}, such that $q=p$. Indeed, the benchmark is indifferent to the budget balance property that one considers, and, thus, the strongest SBB can be assumed w.l.o.g. This is \emph{not} the case for learning algorithms, whose performance is affected by the budget balance property that they ensure while learning.}
\begin{itemize}
    \item \emph{strong budget balance} (SBB) ensures that the broker never loses or gains money, which means that $p_t = q_t$ must hold at every time $t$;
    \item \emph{weak budget balance} (WBB) guarantees that the broker never loses money, allowing them to gain profits, namely, it must hold $p_t \leq q_t$ at every time $t$;
    \item \emph{global budget balance} (GBB) only ensures that the broker does \emph{not} subsidize the market over the whole time horizon $T$, allowing for some losses during learning, namely, this amounts to guaranteeing that $\sum_{t=1}^T \rev_t(p_t,q_t) \geq 0$, where we let $\rev_t(p_t,q_t) \coloneqq (q_t - p_t) \mathbb{I} \{ s_t \leq p_t, q_t \leq b_t \} $ be the revenue of the broker at time step $t$.  
\end{itemize}
While SBB and WBB have been the main focus of online bilateral trade research since its origin (see, \emph{e.g.},~\citep{cesa2021regret}), GBB has been recently introduced by~\citet{bernasconi2024no}, as it allows to circumvent strong impossibility results for SBB and WBB mechanisms.

When agents' valuations are stochastic, it is well known that SBB/WBB learning algorithms can attain $\mathcal{O}(\sqrt{T})$ regret under full feedback, while $\Omega(T)$ regret is unavoidable even with two-bit feedback~\citep{cesa2024bilateral}. Moreover, when agents' valuations are adversarial, no SBB/WBB learning algorithm can attain sublinear regret, even under full feedback~\citep{cesa2024bilateral,azar2022alpha}. \citet{bernasconi2024no} show that $\mathcal{O}(T^{3/4})$ regret is achievable by means of GBB learning algorithms, even when agents' valuations are adversarial and only one-bit feedback is available.
%
%
Moreover, they also prove that any GBB learning algorithm has to suffer a regret at least of the order of $\Omega(T^{5/7})$.
%
Very recently, \citet{chen2025tight} improved the lower bound by~\citet{bernasconi2024no} to $\Omega(T^{3/4})$, closing the gap between lower and upper regret bounds of GBB algorithms.

Previous results on GBB learning algorithms demonstrate that relaxing the budget balance requirement allows to attain better regret bounds. This raises the following natural question:
\begin{center}
    \emph{``Which regret rates can be attained if one allows learning algorithms to slightly violate GBB?''}
\end{center}
The goal of this paper is to answer this question. In particular, we look for learning algorithms that aim at jointly managing the regret and the (cumulative) \emph{violation} of GBB, which is defined as $V_T \coloneqq - \sum_{t=1}^T \rev_t(p_t,q_t)$. Ideally, we would like algorithms whose regret $R_T$ and violation $V_T$ both grow sublinearly in the time horizon $T$. As we will discuss shortly, learning algorithms that are allowed to incur in a sublinear violation of GBB can achieve better regret rates than those sticked to GBB. Our main result is a characterization of the optimal trade-off between regret and violation.

From a more practical perspective, allowing for a slight violation of GBB does \emph{not} substantially hinder the applicability of learning algorithms. Indeed, there are many real-world scenarios in which the broker is willing to incur in some losses, provided that the result is improving performance in terms of GFT. As long as the violation $V_T$ grows sublinearly in $T$, the financial commitment that the broker has to put in subsidizing the market is acceptable, as they are assured that, in the long run, the amount of money that they need to bring in at every time step is negligible.

\subsection{Main Contributions}

The main contribution of this paper is a complete characterization of the regret bounds achievable by means of learning algorithms that are allowed to violate GBB. In particular, we show that there exists a trade-off between violation of GBB and regret. Our first main result is an algorithm that is able to (optimally) manage such a trade-off, as informally stated in the following theorem.
\begin{theorem}[Informal]\label{thm:informal_upper}
    In repeated bilateral trade with adversarial agents' valuations and one-bit feedback, there exists an algorithm that, given as input a trade-off parameter $\beta \in [\nicefrac{3}{4},\nicefrac{6}{7}]$, guarantees that $V_T \leq \tilde{\mathcal{O}}(T^\beta)$ and $R_T \leq \tilde{\mathcal{O}}(T^{1-\nicefrac{\beta}{3}})$ with high probability.
\end{theorem}
Our second main result is a lower bound showing that the regret-violation trade-off attained by our learning algorithm is optimal. Indeed, this result is based on a construction using instances with stochastic agents' valuations and two-bit feedback, making our lower bound even stronger.
\begin{theorem}[Informal]\label{thm:informal_lower}
    In repeated bilateral trade with stochastic agents' valuations and two-bit feedback, any algorithm that guarantees $V_T \leq T^\beta$, for some parameter $\beta \in [\nicefrac{3}{4},\nicefrac{6}{7}]$ given as input to the algorithm, must suffer a regret $R_T \geq \Omega(T^{1-\nicefrac{\beta}{3}})$.
\end{theorem}
The trade-off between regret and violation is graphically depicted in Figure~\ref{fig: fig 1}. Intuitively, it shows that, as long as the broker is willing to incur in a violation of GBB of the order of $T^\beta$, a regret of the order of $T^{1-\nicefrac{\beta}{3}}$ can be attained. Notice that, for $\beta \in [\nicefrac{3}{4},\nicefrac{6}{7}]$, such a regret rate is substantially better than the $\Omega(T^{3/4})$ rate obtained by~\citet{bernasconi2024no} under GBB, while it matches the $\Omega(T^{5/7})$ lower bound by~\citet{bernasconi2024no} for $\beta = \nicefrac{6}{7}$.\footnote{While the statement in \citep{bernasconi2024no} is restricted to global budget balance mechanisms, it is easy to see that such constraints are never used in the proof.}
Notice that a matching $\Omega(T^{3/4})$ lower bound has been independently proven by the concurrent work of \citet{chen2025tight}.

An interesting corollary of our result is that both the $\tilde O(T^{3/4})$ upper bound in the GBB case and the $\Omega(T^{5/7})$ lower bound under unconstrained budget balance in \cite{bernasconi2024no} are tight.

Moreover, our results show that a trade-off between regret and violation of GBB exists only for specific orders of the violations. In particular, violations of order $O(T^{3/4})$ do \emph{} help in reducing the regret, while increasing violations above $T^{6/7}$ is ineffective.

\subsection{Challenges and Techniques}

The paper is split into three parts. The first one (Section~\ref{sec:stochastic}) focuses on designing a learning algorithm for the special case in which the agents' valuations are selected stochastically. As we will discuss shortly, this case already raises considerable challenges, and addressing them paves the way to the design of a learning algorithm dealing with adversarial valuations (Theorem~\ref{thm:informal_upper}), which is the focus of the second part of the paper (Section~\ref{sec:adv}). Finally, the third part of the paper presents the construction for our lower bound result (Theorem~\ref{thm:informal_lower}).

\subsubsection{Stochastic Agents' Valuations}

In the first part of the paper, we focus on a setting with stochastic valuations. While our results for this setting are implied by those for the adversarial one, starting from this simpler setting is a good warm up that allows to highlight some of the challenges faced in the more complex adversarial case.

One of the main challenges in designing learning algorithms for repeated bilateral trade problems is the optimization of a non-Lipchitz function (the GFT) over a continuous domain (the set of mechanisms defined by price pairs $(p,q)$). This challenge can be partially circumvented by focusing on GBB algorithms and exploiting the Lipschitzness of the GFT along some directions (see, \emph{e.g.},~\citet{bernasconi2024no}). In particular, by searching among mechanisms $(p,q)$ that lie slightly below the diagonal $p=q$ (\emph{i.e.}, the space of SBB mechanisms), it is possible to build a finite (fixed) grid of such mechanisms, which provide a good approximation of the GFT of mechanisms on the diagonal. The crucial observation made by~\citet{bernasconi2024no} is that the negative GFT resulting from being below the diagonal can be lower bounded by the distance from the diagonal. Thus, by moving below the diagonal in the bottom-right direction, the GFT does \emph{not} decrease too much. This approach allows to learn within a finite space of mechanisms and design GBB algorithms with $\tilde{\mathcal{O}}(T^{3/4})$ regret. However, it cannot be pushed further, even by relaxing GBB. 

The core idea that allows us to go beyond the approach by~\citet{bernasconi2024no} is a tighter lower bound on the negative GFT incurred by mechanisms below the diagonal. This allows us to work with a ``sparse'' set of mechanisms.  Our bound crucially accounts for the probability that a trade with negative GFT occurs. Intuitively, our approach uses mechanisms $(p,q)$ that lie below the diagonal but still guarantee that $(q-p) \mathbb{E}_{s,b} \left[ \mathbb{I} \{ s \leq p, q \leq b \} \right]$ is \emph{not} too negative. In order to get our results, we propose an approach based on building an adaptive gird, where adaptivity is intended with respect to the (stochastic) agents' valuations characterizing the problem instance at hand. At a high level, the grid is initialized with some (fixed) suitable set of mechanisms below the diagonal. Then, if the probability of a trade with negative GFT for a mechanism in the grid is discovered to be too large, the mechanism is replaced by two new mechanisms that lie nearer to the diagonal. This has the double effect of lowering the probability of a trade with negative GFT and reducing the difference between prices. Such a procedure is repeated recursively until a satisfactory grid is constructed.
Our main result is to show that the cardinality of our grid is much smaller than that of thw grid by~\citet{bernasconi2024no}, making learning an optimal mechanism on the grid easier.

After building an adaptive grid, our algorithm proceeds by estimating the GFT of each mechanism on the grid. Here, our main contribution is the design of an unbiased estimator of the GFT for non-SBB mechanisms. Notice that, since our grid has smaller cardinality than the one in \cite{bernasconi2024no}, the regret of our algorithm in this exploration phase is lower.

Finally, our algorithm for the case with stochastic agents' valuations  simply commits to the best mechanism on the grid (according to the estimates) for the remaining time steps.

\subsubsection{Adversarial Agents' Valuations}

Our algorithm for adversarial agents' valuations builds upon the tools introduced to deal with stochastic valuations. Notice that our algorithm for the stochastic case has no hope of working in adversarial settings, as it is based on an explore-then-commit approach.

We address this issue by using a block decomposition technique that allows to convert explore-then-commit algorithms designed for stochastic settings into ones capable of managing adversarial ones. This approach has already been applied in repeated bilateral trade problems by~\citet{bernasconi2024no}. The idea of block decomposition is to split the time horizon $T$ into $N$ blocks of consecutive time steps. In each round $j$, some random time steps are devoted to the exploratory task of the algorithm, making algorithm estimations robust to the adversarial nature of the environment.

The block decomposition alone is \emph{not} sufficient to address our adversarial repeated bilateral trade problem, as it raises two additional challenges. The first one is that a new adaptive grid $\mathcal{F}_j$ must be built at runtime in each block $j$. Intuitively, this is required since we have to replace the probability estimations of the stochastic setting with the empirical distributions. Clearly, the empirical distribution changes during rounds requiring the grid to change at runtime. The second challenge is that, since the grid is different in each block, the set of mechanisms available to the algorithms changes over time. 

Our algorithm and analysis for the adversarial case rely on the observation that the regret can be decomposed into two components:
\begin{itemize}
    \item The \emph{grid} regret, which measures the sub-optimality of the sequence of adaptive grids $\mathcal{F}_j$ built by the algorithm. This regret term is managed by means of a grid construction technique inspired by that employed in the stochastic setting, but working at runtime instead of upfront.
    \item The \emph{dynamic} regret, which measure the performance with respect to an optimal sequence of mechanisms on the grids $\mathcal{F}_j$. This second regret term is managed by adopting techniques inspired by~\citet{cesa2012mirror} into \emph{sleeping experts} problems. These are problems in which the set of actions of the learner changes over time, as in our setting. In particular, we derive novel regret bounds that depend only on the number of action changes, improving known guarantees for sleeping bandits in ``slightly non-stationary'' environments.
\end{itemize}

\subsubsection{Lower Bound Construction}

The instances constructed to prove our lower bound are inspired by the apple-tasting framework. Intuitively, each instance divides the space of mechanisms into disjoint regions for exploration and exploitation. This makes distinguishing among such instances a particularly challenging task for any learning algorithm.

Our construction builds on top of the one by~\citet{bernasconi2024no}. In such a construction GBB mechanisms are ``weakened'' by making their GFT small. However, their construction fails in characterizing the trade-off between the regret and GBB violations. Indeed, it provides an $\Omega(T^{5/7})$ lower bound that holds even when linear violation is allowed.

In our construction, we follow a different approach and replace the multi-apple tasting gadget in~\citet{bernasconi2024no} with two copies of it. This makes GBB mechanism worse than SBB ones whenever the learner is unable to identify the underlying instance. In doing so, we face the additional challenge of designing a family of hard instances that depends on the allowed budget violation.
One interesting implication of our construction is that the regret is not lower bounded thanks to the complexity of building an adaptive grid, but simply by the problem of determining an optimal mechanism in a known grid. Indeed, in all our hard instances an optimal mechanism belongs to a uniform grid (known to the learner).

\subsection{Concurrent Work by \citet{chen2025tight}}

Concurrently to our work, \citet{chen2025tight} derive a $\Omega(T^{3/4})$ lower bound for the special case of algorithms that satisfy GBB exactly. Their construction makes ineffective playing non-balanced mechanisms that are compensated by balanced one, through the introduction of additional trades that make global budget balance mechanisms worse than strongly budget balanced ones.

Our construction follows a different approach as we replace the multi-apple tasting gadget in \cite{bernasconi2024no} with two copies of the problem. Moreover, our result generalizes the lower bound of \citet{chen2025tight} to settings with violation of the global budget balance constraint.

\begin{figure}
    \centering
 \tikzset{every picture/.style={line width=0.75pt}} 

\begin{tikzpicture}[x=0.75pt,y=0.75pt,yscale=-1,xscale=1]

\draw  [fill={rgb, 255:red, 74; green, 144; blue, 226 }  ,fill opacity=0.44 ][dash pattern={on 0.84pt off 2.51pt}] (129,154) -- (429,154) -- (429,284) -- (129,284) -- cycle ;
\draw  [fill={rgb, 255:red, 208; green, 2; blue, 27 }  ,fill opacity=0.28 ][dash pattern={on 0.84pt off 2.51pt}] (129,84) -- (429,84) -- (429,194) -- (129,194) -- cycle ;
\draw  [fill={rgb, 255:red, 245; green, 166; blue, 35 }  ,fill opacity=0.74 ][dash pattern={on 0.84pt off 2.51pt}] (129,84) -- (149,84) -- (149,154) -- (129,154) -- cycle ;
\draw [color={rgb, 255:red, 41; green, 223; blue, 183 }  ,draw opacity=1 ][line width=3]    (129,154) -- (138.27,154) -- (309,154) -- (379,194) -- (429,194) ;
\draw [color={rgb, 255:red, 255; green, 0; blue, 0 }  ,draw opacity=1 ][line width=3]    (129,154) -- (149,154) ;
\draw    (129,284) -- (129,57) ;
\draw [shift={(129,54)}, rotate = 90] [fill={rgb, 255:red, 0; green, 0; blue, 0 }  ][line width=0.08]  [draw opacity=0] (8.93,-4.29) -- (0,0) -- (8.93,4.29) -- cycle    ;
\draw    (129,284) -- (309,284) -- (379,284) -- (466,284) ;
\draw [shift={(469,284)}, rotate = 180] [fill={rgb, 255:red, 0; green, 0; blue, 0 }  ][line width=0.08]  [draw opacity=0] (8.93,-4.29) -- (0,0) -- (8.93,4.29) -- cycle    ;
\draw   (490,80) -- (600,80) -- (600,220) -- (490,220) -- cycle ;
\draw  [draw opacity=0][fill={rgb, 255:red, 208; green, 2; blue, 27 }  ,fill opacity=0.28 ][dash pattern={on 0.84pt off 2.51pt}] (500,100) -- (510,100) -- (510,110) -- (500,110) -- cycle ;
\draw  [draw opacity=0][fill={rgb, 255:red, 74; green, 144; blue, 226 }  ,fill opacity=0.44 ][dash pattern={on 0.84pt off 2.51pt}] (500,130) -- (510,130) -- (510,140) -- (500,140) -- cycle ;
\draw  [draw opacity=0][fill={rgb, 255:red, 245; green, 166; blue, 35 }  ,fill opacity=0.74 ][dash pattern={on 0.84pt off 2.51pt}] (500,160) -- (510,160) -- (510,170) -- (500,170) -- cycle ;
\draw [color={rgb, 255:red, 41; green, 223; blue, 183 }  ,draw opacity=1 ][line width=3]    (500,195) -- (509.27,195) ;

\draw (131,287.4) node [anchor=north west][inner sep=0.75pt]    {$0$};
\draw (311,287.4) node [anchor=north west][inner sep=0.75pt]    {$\frac{3}{4}$};
\draw (381,287.4) node [anchor=north west][inner sep=0.75pt]    {$\frac{6}{7}$};
\draw (431,287.4) node [anchor=north west][inner sep=0.75pt]    {$1$};
\draw (471,287.4) node [anchor=north west][inner sep=0.75pt]    {$\log_{T}( V_{T})$};
\draw (127,50.6) node [anchor=south east] [inner sep=0.75pt]    {$\log_{T}( R_{T})$};
\draw (127,154) node [anchor=east] [inner sep=0.75pt]    {$3/4$};
\draw (127,194) node [anchor=east] [inner sep=0.75pt]    {$5/7$};
\draw (518,104.5) node [anchor=west] [inner sep=0.75pt]    { \footnotesize \cite{bernasconi2024no} (LB)};
\draw (518,134.5) node [anchor=west] [inner sep=0.75pt]    {\footnotesize \cite{bernasconi2024no} (UB)};
\draw (518,164.5) node [anchor=west] [inner sep=0.75pt]    {\footnotesize \cite{chen2025tight} (LB)};
\draw (518,201.6) node [anchor=south west] [inner sep=0.75pt]    {Our work};
\draw (345.67,148.58) node [anchor=north west][inner sep=0.75pt]  [rotate=-30]  {$1-\frac{1}{3} x$};

\end{tikzpicture}
\caption{(\emph{red area}) Possible regret-violation trade-offs according to the lower bound by~\citet{bernasconi2024no}---$\tilde \Omega(T^{5/7})$ with unbounded violations. (\emph{blue area}) Possible regret-violation trade-offs according to the upper bound by~\citet{bernasconi2024no}---$\tilde O(T^{3/4})$ with no violation. (\emph{orange area}) Possible regret-violation trade-offs according to the lower bound by~\citet{chen2025tight}---$\tilde \Omega(T^{3/4})$ with no violation. (\emph{green line}) Our tight (upper and lower bound) regret-violation trade-off. }
    \label{fig: fig 1}
\end{figure}
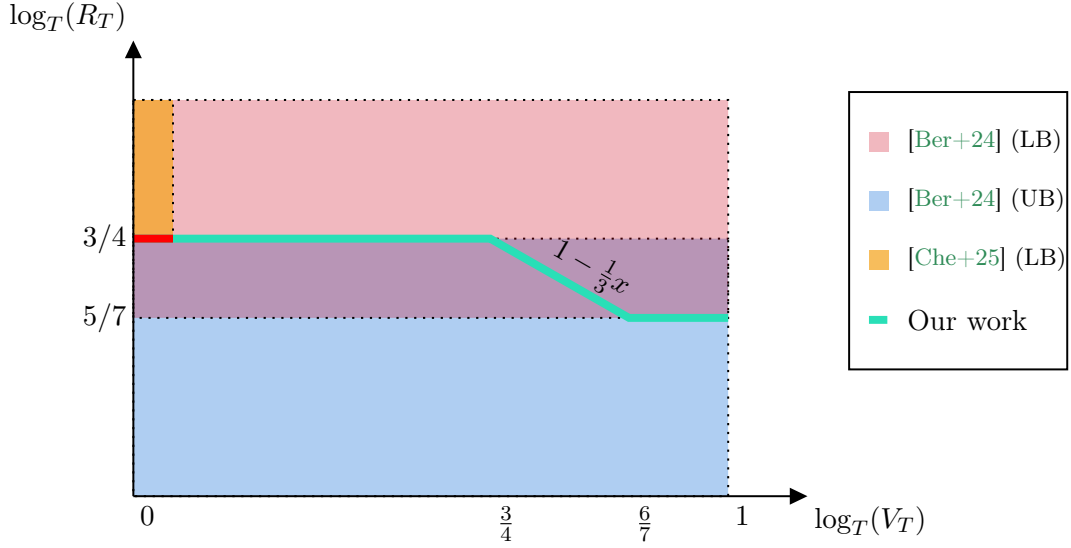
\section{Preliminaries}


We consider \emph{repeated bilateral trade} problems where a learner faces a sequence of buyer-seller pairs willing to trade an item.
At each time step $t\in [T]$, a new seller and a new buyer arrive, with private valuations $s_t \in [0,1]$ and $b_t \in [0,1]$, respectively, for the item.\footnote{Given $x \in \mathbb{N}$, we let $[x]=\{1,\ldots,n\}$.} Unless stated otherwise, we consider the case in which such valuations are chosen by an oblivious adversary.
At each time $t$, the learner proposes two prices: $p_t\in [0,1]$ to the seller and $q_t\in [0,1]$ to the buyer. If the seller agrees to sell the item and the buyer agrees to buy it at the proposed prices, then they trade the item.
Formally, this happens when $\mathbb{I}\{s_t\le p_t, q_t\le b_t\} = 1$.
Then, in the case in which a trade occurred, the learner receives a payment $q_t$ from the buyer and pays $p_t$ to the seller.

The performance of the learner is usually evaluated in terms of two metrics: \emph{gain from trade} and \emph{revenue}. The former is formally defined as follows.

\begin{definition}[Gain From Trade]
The \emph{gain from trade} (GFT) at time step $t \in [T]$ by posting a pair of prices $(p,q)\in [0,1]^2$ is defined as: 
    \[\gft_t(p,q) \coloneqq (b_t-s_t)\mathbb{I} \{s_t\le p, q \le b_t\}.\]
\end{definition}
The gain from trade at time $t$ is defined as the difference between the buyer's valuation $b_t$ and the seller's one $s_t$ whenever the trade happens, while it is equal to zero otherwise.
Notice that maximizing the social welfare is equivalent to maximizing the gain from trade.

The revenue is related to the profit of the learner, and it is defined as follows: 
\begin{definition}[Revenue]
The \emph{revenue} at time step $t \in [T]$ by posting prices $(p,q)\in [0,1]^2$ is: 
    \[\rev_t(p,q) \coloneqq (q-p)\mathbb{I}\{s_t\le p, q \le b_t\}.\]
\end{definition}
The learner is usually required to satisfy some budget balance constraints, which enforce some lower bounds on the revenue. These guarantee that the learner is \emph{not} subsidizing the market.
%

The interaction between the learner and the environment is described by~\Cref{alg: mech env interaction}.

\begin{algorithm}[!htp]\caption{Learner-Environment Interaction} \label{alg: mech env interaction}
    \begin{algorithmic}[1]
        \For{$t=1,\ldots,T$}
        \State The valuations $(s_t,b_t)$ are chosen by the environment
        \State The learner proposes two prices $(p_t,q_t)\in [0,1]$
        \State The learner receives feedback $\mathbb{I}\{s_t\le p_t, q_t\le b_t\}$
        \If{$\mathbb{I}\{s_t\le p_t, q_t\le b_t\}=1$}
        \State $\gft_t(p_t,q_t)= b_t-s_t$
        \State $\rev_t(p_t,q_t)= q_t-p_t$
        \Else 
        \State $\gft_t(p_t,q_t)=0,\rev_t(p_t,q_t)=0$
        \EndIf
        \EndFor
    \end{algorithmic}
\end{algorithm}

\paragraph{Feedback} In this work, we focus on the one-bit feedback, where at each time step $t$, after posting the pair of prices $(p_t,q_t)$ the learner observes $\mathbb{I} \{ s_t\le p_t, q_t\le b_t \}$. 
Despite this, our lower bound works for the stronger two-bit feedback, where at each time step $t$, after posting the pair of prices $(p_t,q_t)$ the learner observes $\mathbb{I} \{ s_t\le p_t\}$ and $\mathbb{I} \{q_t\le b_t \}$.

\paragraph{The Gain From Trade/Revenue Tradeoff}
The most common goal in bilateral trade is to maximize the gain from trade subject to some constraints on the revenue.
In order to evaluate the performance of the learner, we introduce the following notion of regret with respect to the best fixed price in hindsight. Since our baseline is the best fixed price, we can restrict to mechanisms posting the same price to both agents \citep{cesa2023repeated,cesa2024bilateral,bernasconi2024no}.
Notice that this equivalence between posting two different (weakly budget balanced) prices and the same price to both agents exists only from an optimization perspective. From a learning perspective, posting different prices is essential as shown by previous works \citep{cesa2024bilateral,bernasconi2024no}.

\begin{definition}[Regret]
    The \emph{regret} of the learner over the $T$ time steps is defined as:
    \[R_T \coloneqq \max_{p\in [0,1]} \,\, \sum_{t=1}^T \gft_t(p,p) - \sum_{t=1}^T\gft_t(p_t,q_t).\]
\end{definition}

While minimizing the regret, the learner should satisfy some constraint related to budget violation, guaranteeing that the learner is \emph{not} losing money.
%
%
%
This can be of three different types:
\begin{itemize}
    \item \emph{strong budget balance} (SBB) requires $\rev_t(p_t,q_t)=0$ for every $t\in [T]$;
    \item \emph{weak budget balance} (WBB) is equivalent to ensuring that $\rev_t(p_t,q_t)\ge0$ for every $t\in [T]$;
    \item \emph{global budget balance} only requires that $\sum_{t=1}^T \rev_t(p_t,q_t)\ge 0$.
\end{itemize}


It is well known that global budget balance is the weaker budget constraint and the only one that allows to achieve sublinear regret with adversarial valuations.
Despite that, the best achievable regret bound is $\tilde O(T^{3/4})$.
In this paper, we are interested in studying the relation between the regret and how much the learner is subsidizing the market. In order to quantify how much the learner is subsidizing the market, we introduce the following notion of \emph{budget violation}:
\[V_T \coloneqq  - \sum_{t=1}^T \rev_t(p_t,q_t).\]
Intuitively, our goal is to study if slightly subsidizing the market, \emph{i.e.}, allowing for $V_T=o(T)$, is sufficient to guarantee better regret bounds, and what is the relation between the budget violation $V_T$ and the achievable regret $R_T$.

In particular, given a target budget violation $T^\beta$, for some $\beta \in [0,1]$, we want to minimize the regret subject to $V_T\le T^\beta$. Notice that all our algorithms guarantee to keep $V_T\le T^\beta$ deterministically.

    

\section{Warm Up: Stochastic Case} \label{sec:stochastic}

While our result holds for the more general adversarial case, we present a simplified proof for stochastic problems. This allows to highlight some challenges that are already present in the stochastic case, and then extend our approach to the more complex adversarial one.

\subsection{Stochastic Online Bilateral Trade}

%

In this section, we consider a stochastic environment. In particular, we assume that buyer and seller' valuations are independent and identically distributed (i.i.d.) realizations of a fixed joint probability distribution, \emph{i.e.}, $(s_t,b_t) \sim \mathcal{P}$ for every $t \in [T]$.

In stochastic settings, it is much easier to work with the \emph{expected} gain from trade of a pair of prices $(p,q)\in [0,1]^2$, defined as follows:
\[\gft(p,q) \coloneqq \mathbb{E}_{(s,b)\sim \mathcal{P}}\left[(b-s)\mathbb{I} \{ s\le p, q \le b \} \right].\]
Notice that $\mathbb{E}[\gft_t(p,q)]=\gft(p,q)$ for all $t\in \{1,\ldots,T\}$.

Similarly, we define 
\[\rev(p,q) \coloneqq \mathbb{E}_{(s,b)\sim \mathcal{P}}\left[(q-p)\mathbb{I} \{ s\le p, q \le b \} \right].\]

Then, the goal of the learner is to minimize the pseudo-regret
\[\mathcal{R}_T\coloneqq T \cdot \max_{p \in [0,1]} \gft(p,p)-\sum_{t=1}^T \gft(p_t,q_t)\]
while keeping the pseudo-violation $\mathcal{V}_T\coloneqq - \sum_{t=1}^T\rev(p_t,q_t)\le T^\beta$.



\subsection{Main Challenge}

We design a novel algorithm for stochastic bilateral trade, achieving the optimal trade-off between regret and budget violation (see \Cref{fig: fig 1}). In particular, given any $\beta\in [3/4,6/7]$, it achieves regret of the order of $\BigOL{T^{1-\beta / 3}}$ by suffering $T^\beta$ budget violation, matching the lower bound in Section \ref{sec:lowerbound}. 

One of the main challenges in designing algorithms for bilateral trade lies in optimizing a non-Lipschitz function over a continuous domain.
This challenge is highlighted, for instance, in \citep{cesa2024bilateral}, which shows a linear regret lower bound for SBB mechanisms.


\begin{figure}[H]
\begin{minipage}{0.48\textwidth}
    \centering
\begin{tikzpicture}[scale=4]

  \coordinate (O) at (0,0);
  \coordinate (A) at (1,0);
  \coordinate (B) at (1,1);
  \coordinate (C) at (0,1);
  \coordinate (P) at (0.6,0.2);
  \coordinate (Ps) at (0.5,0.5);

  \draw[thick] (O) rectangle (B);

  \draw[thick] (O) -- (B);

  \fill[green!30] (0,0.2) rectangle (0.6,1);
    \fill[red!30]  (P)-- (0.2,0.2) -- (0.6,0.6) --  cycle;
  \fill[yellow!50] (0,0.5) rectangle + (0.5,0.5);

  \node[circle, fill=black, inner sep=0.5pt, label=left:\footnotesize$p^*$] at (Ps) {};

  \node at (P) {\large$\oplus$};
  \node[below] at (P) {\footnotesize$(p,q)$};

\end{tikzpicture}
\end{minipage}
\begin{minipage}{0.48\textwidth}
\centering
\begin{tikzpicture}[scale=4]

  \coordinate (O) at (0,0);
  \coordinate (A) at (1,0);
  \coordinate (B) at (1,1);
  \coordinate (C) at (0,1);
  \coordinate (P) at (0.6,0.2);
  \coordinate (Ps) at (0.5,0.5);
  \coordinate (P') at (0.4,0.2);
  \coordinate (P'') at (0.6,0.4);

  \draw[thick] (O) rectangle (B);

  \draw[thick] (O) -- (B);

  \fill[green!30] (0,0.4) rectangle (0.6,1);
    \fill[red!30]  (P'')-- (0.4,0.4) -- (0.6,0.6) --  cycle;
  \fill[yellow!50] (0,0.5) rectangle + (0.5,0.5);

  \node[circle, fill=black, inner sep=0.5pt, label=left:\footnotesize$p^*$] at (Ps) {};

  \node[fill opacity=0.3] at (P) {\large$\oplus$};
  \node[below,fill opacity=0.3] at (P) {\footnotesize$(p,q)$};
    \node at (P') {\large$\oplus$};
      \node[below] at (P') {\footnotesize$(p',q')$};
       \node at (P'') {\large$\oplus$};
      \node[right] at (P'') {\footnotesize$(p'',q'')$};

\end{tikzpicture}
\end{minipage}
\caption{\emph{Left}: relation between the trade realized by the strong budget balance price $(p^*,p^*)$ and the prices $(p,q)$. The yellow region represent the trades realized by both prices, the green one the trade with positive GFT realized only by $(p,q)$, while the red one the trade realized only by $(p,q)$ with negative GFT. \emph{Right:} a step of the grid construction procedure. The prices $(p,q)$ are replaced by the two couples $(p',q')$ and $(p'',q'')$. This guarantees that at least one of the two couples (in this case $(p'',q'')$) ``covers'' $(p^*,p^*)$. Moreover, the red region related to $(p'',q'')$ decreases.}
\end{figure}
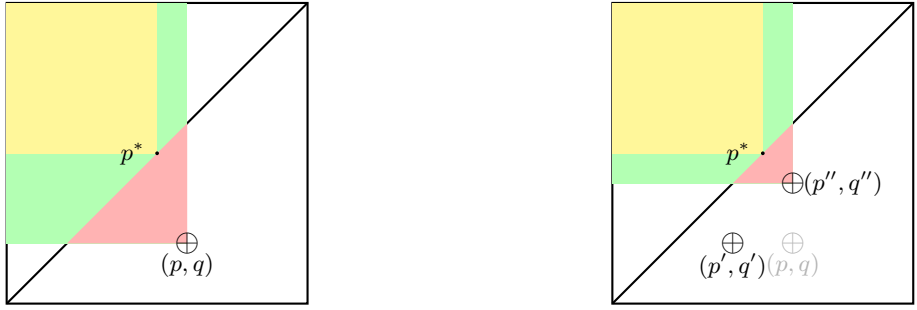

This challenge can be partially circumvented by moving towards global budget balance mechanisms, which use the Lipschitzness of the gain from trade along some directions.
Intuitively, by playing slightly below the diagonal $p=q$, it is possible to build a finite-dimensional grid that approximates the gain from trade of the diagonal, which is a continuous domain.

As an example, consider the left side of \Cref{fig: fig 1}. Assume that the optimal price is $p^*$. If the algorithm posts the prices $(p,q)$ the trade happens when the valuations are in the upper right rectangle defined by $(p,q)$, \emph{i.e.}, $b,s$ such that $s\le p,q\le b$.
This rectangle can be split into three areas: the yellow one, the green one, and the red one.
From a trade in the yellow rectangle, the pair $(p,q)$ gets the same GFT of $p^*$, from the green region it gets an additional positive revenue, while from the red one it generates negative GFT. 
Notice that such an area includes all the couples $s_t,b_t$ such that $q\le b_t \le s_t \le p$, for which $\gft_t(p^*,p^*)=0 \ge \gft_t(p,q)$.
%
Regarding the revenue, the farther we move from the diagonal, the larger is the budget violation.
These two aspects build a tension between playing close to the diagonal and far from it.

One crucial component of our approach is to bound the negative GFT that we get playing the not budget balance prices $(p,q)$.
Past works on global budget balance mechanisms made the simple observation that such a negative gain from trade is upper bounded by the distance from the diagonal.
Formally:
\begin{align} \label{eq:distanceDiagonal}
\mathbb{E}[\mathbb{I}\{q\le b \le s \le p\} (b-s)]\ge q-p,
\end{align}
essentially showing that moving towards the right-bottom direction the GFT does not decrease too much, and in particular by the distance from the diagonal.
This allows us to approximate the gain from trade of all the prices in the interval $[q,p]$, with a single price, reducing from a continuous to a discrete domain. However, in doing so, we are suffering an error $p-q$.
While this approach is sufficient to provide a $\tilde O(T^{3/4})$ regret bound, it cannot be pushed any further, even by relaxing the budget constraint. One of our main contributions is a tighter bound on the negative GFT based also on the probability that a trade with negative GFT occurs.

\paragraph{General Idea}

At a high level, our goal is to find a discrete grid of pairs of prices $\mathcal{F}$ with $p>q$ that approximates the GFT obtained by the couples of prices on the diagonal $\{(p,p)\}_{p \in [0,1]}$ with a given approximation $\alpha>0$. 

We want to achieve the following three goals:
\begin{itemize}
    \item For every $a\in [0,1]$, there is a $(p,q)\in \mathcal{F}$ such that $q\le a\le p$. Let $F(a)$ be such a pair.
    \item For every $a \in [0,1]$, $\gft(a)\le \gft(F(a))+\alpha$.
    \item The cardinality of the grid $\mathcal{F}$ is the minimum possible.
\end{itemize}

The first goal ensures that the pair $(p,q)$ guarantees a superset of the trades guaranteed by $a$. The second goal ensures that we are not loosing too much from trades with negative GFT.
The third goal guarantees to have a small set of possible candidate solutions, making the subsequent learning step easier.

In doing so, we also take into consideration the allowed budget violation, which determines how far points in $\mathcal{F}$ can be from the diagonal. Indeed,
\[\rev(p,q)= \mathbb{E}_{(s,b)\sim \mathcal{P}}\left[(q-p)\mathbb{I}(s\le p, q \le b)\right]\ge q-p.\]
Despite that, we trivially satisfy the budget constraint restricting to couples $(p,q)$ with $p-q\le T^{\beta-1}$, where we recall that $T^{\beta}$ is the allowed budget violation.

The straightforward approach employed in \cite{bernasconi2024no} consider the grid $(p=i\alpha,q=(i-1)\alpha)$ with $i$ in $[1/\alpha]$. The size of the grid is $1-\alpha$, while the suboptimality of the GFT is upper bounded by  their distance from the diagonal line $\alpha$ (see \Cref{eq:distanceDiagonal}). 

However, as we already highlighted in \Cref{eq:distanceDiagonal}, the contribution to the GFT from the red region in \Cref{fig: fig 1} depends on two factors: the probability density in that region and the distance of those valuations from the diagonal.
Hence, we our goal is to guarantee that 
\begin{align}\label{eq:ideal}
(q-p)  \, \mathbb{P}(q\le b \le s \le p)\simeq - \alpha.
\end{align}

Our idea is to build an adaptive grid that satisfies this constraint.
The key intuition is that a region with high error can be addressed by partitioning it into smaller sub-triangles, thereby reducing the negative impact of both the probability of a trade with negative GFT and the gap between the seller and the buyer prices (see right-hand side of \Cref{fig: fig 1}).

Hence, the main idea is thus to iteratively subdivide a triangle into two smaller ones whenever the probability within it exceeds a certain threshold, bearing in mind that smaller triangles require increasingly higher probability to justify further splitting. Since triangles are disjoint and the total probability is at most one, we will show that the final set of triangles is surprisingly small.

\subsection{Algorithm: Stochastic Setting}
\begin{algorithm}[!htp]\caption{}
\label{alg: stoch tot}
    \begin{algorithmic}[1]
        \State Input: \# rounds $T$, budget violation $T^\beta$ , $\delta$
        \State Set parameters: $K\gets T^{1-\beta}$, $\alpha \gets T^{-\frac{1}{3}\beta}$,  $T_0\gets T^{\frac{2}{3}\beta}$   
        \State $\mathcal{F} \gets \texttt{Grid($K,\alpha,\delta$)}$ 
        \Comment{Grid construction}
    \For{$(p,q)\in \mathcal{F}$}\Comment{GFT exploration}
    \State $\widehat{\gft}(p,q) \gets \texttt{$\gft$-Est.rep}((p,q),T_0)$ 
    \EndFor
    
    \State $(\hat{p},\hat{q})\gets \arg \max_{(p,q)\in \Fcal} \widehat{\gft}(p,q)$ 
    \Comment{Exploitation}
    \While{ $t\le T$}
    \State Play $(\hat{p},\hat{q})$ 
    \EndWhile
\end{algorithmic}
\end{algorithm}
In this section we provide our algorithm and we discuss its mechanism. The algorithm can be conceptually divided into three phases:
\begin{enumerate}
    \item \emph{Grid construction}, where the algorithm, through the observations, build an adaptive grid.
    \item \emph{GFT exploration}, where the algorithm builds an estimate of the expected GFT for every point of the grid
    \item \emph{exploitation}, where the algorithm commits to the optimal couple of prices according to the estimations.
\end{enumerate}

\subsubsection{Grid Construction}

In this section, we design an algorithm to build an adaptive grid. Before doing that, we need an additional component to estimate the probability of some events.
In particular, since our algorithm builds the adaptive grid taking into account the probability that the valuations are in a given set, we need a procedure to estimate a probability in a sector of the prices space.

As we already observe, given a couple of prices $(p,q)$, we are interested in estimating the probability of the negative-GFT trades $\mathbb{P}(q\le b \le s \le p)$ (see the red triangles in figure \Cref{fig: fig 1}).
However, for simplicity, we focus on estimating $\mathbb{P}((s,b)\in [q,p]\times [q,p])$, which overestimates the desired probability.
Then, we exploit this tool to build an iterative procedure to discard and generate prices on the grid. 

\begin{figure}
    \centering
    \scalebox{.9}{
    \begin{tikzpicture}[scale=0.60,
	every node/.style={align=center},
	level 1/.style={sibling distance=15cm},
	level 2/.style={sibling distance=15cm},
	level 3/.style={sibling distance=15cm},
	leaf/.style={rectangle, draw=red, fill=white}
	]
\node (A) at (0,0) {$(p,q)\in \Acal_1$}{
    child {node (B) {$(p-\frac{1}{2K},q)\in \Acal_2$}{
        child {node[leaf] (D) {$(p-\frac{1}{4K},q)\in \Acal_3 \cap \Fcal$}
    }
        child {node (E) {$(p-\frac{1}{2K},q+\frac{1}{2K})\in \Acal_3$}
            child {node[leaf] (G) {$(p-\frac{1}{4K},q+ \frac{1}{2K})\in \Acal_4\cap \Fcal$}}
		child {node[leaf] (F) {$(p-\frac{1}{4K},q+\frac{3}{4K})\in \Acal_4\cap \Fcal$}}
        }
    }}
    child {node[leaf] (C) {$(p,q+\frac{1}{2K})\in \Acal_2\cap \Fcal$}}
    };
%
\end{tikzpicture}}
    \caption{Grid construction restricted to a single $(p,q) \in \Acal_1$. The prices in the grid $\Fcal$ are the leafs of the tree and are highlighted in red.}
    \label{fig: grid construction}
\end{figure}
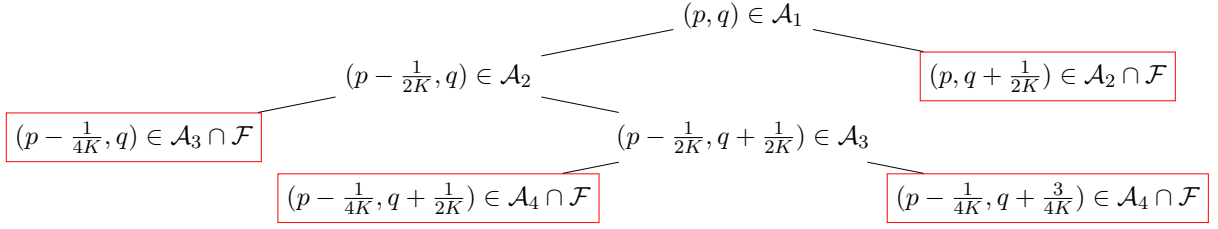

\begin{figure}[H]
\begin{minipage}{0.48\textwidth}
    \centering
\begin{tikzpicture}[scale=6]

  \coordinate (O) at (0,0);
  \coordinate (A) at (1,0);
  \coordinate (B) at (1,1);
  \coordinate (C) at (0,1);
  \coordinate (P1) at (0.5,0);
  \coordinate (P2) at (1,0.5);
  \coordinate (P3) at (0.25,0);
  \coordinate (P4) at (1,0.75);
  \coordinate (P5) at (0.125,0);
  \coordinate (P6) at (1,0.875);

  \coordinate (P21) at (0.5,0.25);
  \coordinate (P22) at (0.75,0.5);

  \coordinate (P31) at (0.5,0.375);
  \coordinate (P32) at (0.375,0.25);

  \coordinate (P33) at (0.875,0.75);

  \draw[color= gray,] (O) rectangle (B);

  \draw[color= gray,] (O) -- (B);
  \draw[color= gray,] (P1) -- (P2);
  \draw[color= gray,] (P3) -- (P4);
  \draw[color= gray,] (P5) -- (P6);

 \draw[line width=0.65mm,color= blue ] (P1) -- (P3);
\draw[line width=0.65mm,color= blue ] (P1) -- (P21);
\draw[line width=0.65mm,color= blue ] (P2) -- (P22);
\draw[line width=0.65mm,color= blue ] (P2) -- (P4);
\draw[line width=0.65mm,color= blue ] (P21) -- (P31);
\draw[line width=0.65mm,color= blue ] (P21) -- (P32);
\draw[line width=0.65mm,color= blue ] (P4) -- (P33);
\draw[line width=0.65mm,color= blue ] (P4) -- (P6);

\draw[color= gray,dashed] (P3) -- (0.25,0.25);
\draw[color= gray,dashed] (P21) -- (0.25,0.25);
\draw[color= gray,dashed] (P1) -- (0.5,0.5);
\draw[color= gray,dashed] (P22) -- (0.75,0.75);
\draw[color= gray,dashed] (P2) -- (0.5,0.5);
\draw[color= gray,dashed] (P4) -- (0.75,0.75);
\draw[color= gray,dashed] (P31) -- (0.375,0.375);
\draw[color= gray,dashed] (P32) -- (0.375,0.375);
\draw[color= gray,dashed] (P33) -- (0.875,0.875);
\draw[color= gray,dashed] (P6) -- (0.875,0.875);

  \node [circle, fill=red, scale=0.5] at (P1) {$\textbf{1.1}$}  ;

  \node [circle, fill=red, scale=0.5] at (P2) {$1.2$}  ;

  \node [circle, fill=yellow, scale=0.5] at (P3) {$\textbf{2.1}$}  ;
  \node [circle, fill=yellow, scale=0.5] at (P21) {$\textbf{2.2}$}  ;
  \node [circle, fill=yellow, scale=0.5] at (P22) {$\textbf{2.3}$}  ;
  \node [circle, fill=yellow, scale=0.5] at (P4) {$\textbf{2.4}$}  ;
  
  
  \node [circle, fill=green, scale=0.5] at (P32) {$\textbf{3.1}$}  ;
   \node [circle, fill=green, scale=0.5] at (P31) {$\textbf{3.2}$}  ;
   \node [circle, fill=green, scale=0.5] at (P33) {$\textbf{3.3}$}  ;
  \node [circle, fill=green,scale= 0.5] at (P6) {$\textbf{3.4}$}  ;

\end{tikzpicture}
\end{minipage}
\begin{minipage}{0.48\textwidth}
\centering
\begin{tikzpicture}[scale=6]
\begin{scope}
[yscale=-1.2]
  \coordinate (O) at (0.5,1);
  \coordinate (11) at (0.175,0.825);
  \coordinate (12) at (0.8,0.825);
  \coordinate (21) at (0,0.5);
  \coordinate (22) at (0.35,0.5);
  \coordinate (23) at (0.65,0.5);
  \coordinate (24) at (1,0.5);
  \coordinate (31) at (0.175,0.325);
  \coordinate (32) at (0.5,0.325);
  \coordinate (33) at (0.85,0.325);
  \coordinate (34) at (1.15,0.325);
  
\draw[line width=0.65mm, color=blue] (11) -- (21);
\draw[line width=0.65mm, color=blue] (11) -- (22);
\draw[line width=0.65mm, color=blue] (12) -- (23);
\draw[line width=0.65mm, color=blue] (12) -- (24);
\draw[line width=0.65mm, color=blue] (22) -- (31);
\draw[line width=0.65mm, color=blue] (22) -- (32);
\draw[line width=0.65mm, color=blue] (24) -- (33);
\draw[line width=0.65mm, color=blue] (24) -- (34); 
  
  \node [circle, fill=red, scale=0.5] at (11) {$\textbf{1.1}$}  ;
  \node [circle, fill=red, scale=0.5] at (12) {$\textbf{1.2}$}  ;
  \node [circle, fill=yellow, scale=0.5] at (21) {$\textbf{2.1}$}  ;
\node [circle, fill=yellow, scale=0.5] at (22) {$\textbf{2.2}$}  ;
\node [circle, fill=yellow, scale=0.5] at (23) {$\textbf{2.3}$}  ;
\node [circle, fill=yellow, scale=0.5] at (24) {$\textbf{2.4}$}  ;
\node [circle, fill=green, scale=0.5] at (31) {$\textbf{3.1}$}  ;
\node [circle, fill=green, scale=0.5] at (32) {$\textbf{3.2}$}  ;
\node [circle, fill=green, scale=0.5] at (33) {$\textbf{3.3}$}  ;
\node [circle, fill=green, scale=0.5] at (34) {$\textbf{3.4}$}  ;

\end{scope}
\end{tikzpicture}
\end{minipage}
\caption{Graphic representation of the forest interpretation of the adaptive grid.}
\label{fig:grid construction 2}
\end{figure}
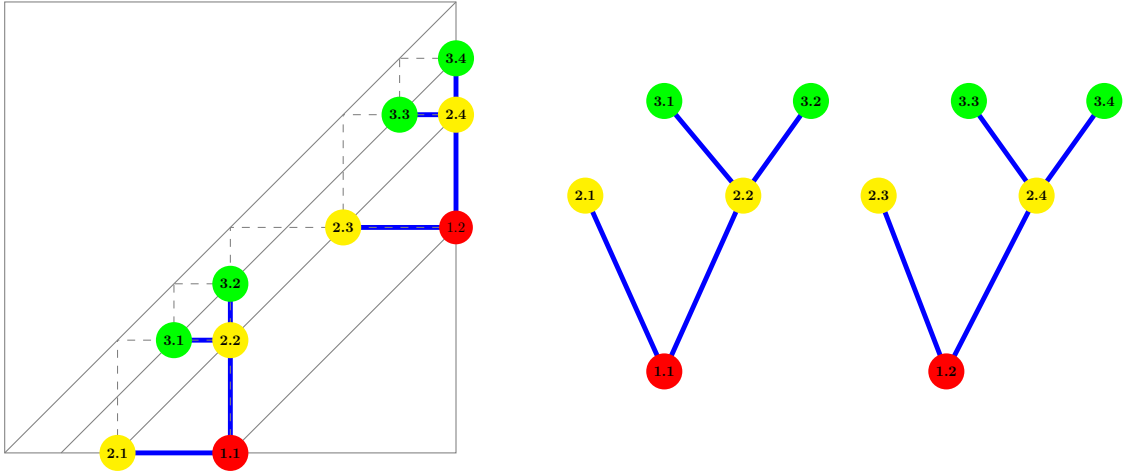

\textbf{Probability Estimation }
%
We present a procedure that, taking as input a couple of prices $(p,q)$ and a number of rounds $\ell$, uses $\ell$ rounds and returns an estimate of $\mathbb{P}((s,b)\in [q,p]\times[q,p])$. Notice that this is an upperbound on the desired probability $\mathbb{P}(q\le b \le s \le p)$.

Our algorithm, called \texttt{Prob.Est},   build empirical estimations of the probability that the valuations are in the regions $\{\mathbb{P}(s\le p,q\le b),\mathbb{P}(s\le q,q\le b),\mathbb{P}(s \le p, b\ge p),\mathbb{P}(s \le q, b\ge p)\}$. It is easy to bound the gap between the true probabilities and the empirical ones through an Hoeffding bound. Finally, it defines the desired probability $\mathbb{P}(q\le s \le p,q\le b \le p)$ as a linear combination of the 4 estimated probabilities. 
For completeness, we present the pseudocode of our algorithm in \Cref{app:stoc}.
In the following, \texttt{Prob.Est} is called with an input $(p,q)$ and a number of rounds $\ell$ outputs a lower bound $\xi(p,q)$ which satisfies the following:\footnote{Notice that $\xi(p,q)$ is a lower bound on $\mathbb{P}(q\le s \le p,q\le b \le p)$ and not on   $\mathbb{P}(q\le b \le s \le p)$.}

\begin{restatable}{lemma}{probest}
\label{lemma: prob est}
    The function \textnormal{\texttt{Prob.Est}} is such that, when called with input $(p,q)\in [0,1]^2$, number of rounds $L$ and confidence parameter $\nu$, it returns a value $\xi(p,q)$ such that with probability at least $1-\nu$
    \[\xi(p,q)\le \mathbb{P}(q\le s \le p,q\le b \le p)\]
    and 
    \[\mathbb{P}(q\le s \le p,q\le b \le p) \le \xi(p,q)+ 8\sqrt{\frac{\ln(\frac{4}{\delta})}{2L}}\]
\end{restatable}

\paragraph{Iterative Grid Construction}

\begin{algorithm}[!htp]\caption{\texttt{Grid}}
\label{alg: grid construction}
    \begin{algorithmic}[1]
        \State Input: grid size $K$, error $\alpha$, confidence parameter $\delta$ 
        \State Initialize $\Acal_1 \gets \left\{\left(\frac{j+1}{K},\frac{j}{K}\right): j= 0,\ldots,K-1\right\}$
        \State $M\gets \log_2\left(\frac{1}{\alpha K}\right)+1$ 
        \For{$i \in [M]$ }
        \State $\Acal_{i+1}\gets \emptyset$  
        \State $H_{i}\gets \emptyset$
        \For{$(p,q)\in \Acal_i$}
        \State $\xi(p,q)\gets\texttt{Prob.Est}((p,q),L= \lceil(\alpha K 2^{i-1})^{-2}\rceil,\nu=\frac{\alpha \delta}{2})$  
        \If{$\xi(p,q)\ge \alpha K 2^i$} 
        \State
        \(\Acal_{i+1}\gets \Acal_{i+1} \cup \left\{\left(p-\frac{2^{-i}}{K},q\right),\left(p,q+\frac{2^{-i}}{K}\right)\right\}\)
        \Else
        \State $H_{i}\gets H_{i}\cup \{(p,q)\}$
        \EndIf
        \EndFor
        \EndFor
        \State \Return $\mathcal{F} \gets \cup_{j=1}^{M}H_j$
\end{algorithmic}
\end{algorithm}

The pseudocode is in \Cref{alg: grid construction}.
The grid construction phase receives as input the initial number of points $K$, a tolerance threshold $\alpha$ and an exploration parameter $T_0$. 

It starts by initializing a set of $K$ equidistant points $\Acal_1=\left\{\left(\frac{j+1}{K},\frac{j}{K}\right): j= 0,\ldots,K-1\right\}$. Then, it  works by cycling over the set of nodes $K$ nodes, \emph{i.e.}, $\Acal_1$, and replacing the points which are too suboptimal with points closer to the diagonal, which are added to the set $\Acal_2$. Then, the procedure iteratively cycles over the set $\Acal_2$. See \Cref{fig: grid construction} for a representation restricted to a single couple of prices $(p,q) \in \Acal_1$ and \Cref{fig:grid construction 2} for a graphical representation of the whole algorithm.
Intuitively, the algorithm builds a forest, \emph{i.e.}, a union of trees, including couples of prices.
Each tree has (at most) $M=\log_2\left(\frac{1}{\alpha K}\right)+1$ levels. In the following, we will refer informally to depth and width of a forest. With depth, we refer to the depth of the deeper tree in the forest, while with width we refer to the cumulative number of nodes in a level of the forest.

In detail, the algorithm does the following for each level $i \in [M]$ of the forest. For each point $(p,q)$ on the $i$-th level, the algorithm calls the function \texttt{Prob.Est}, that returns $\xi(p,q)$, an estimate of the probability of the valuations to assume values in the lower triangle defined by $q\le s< b \le p$.
If the condition $\xi(p,q)\cdot (p-q)\ge \alpha$ is verified, the point $(p,q)$ is discarded and the algorithm generates two novel points $(p-\frac{(p-q)}{2},q),(p,q+\frac{(p-q)}{2})$ (see \Cref{fig: fig 1} right). This guarantees discarding all the prices whose error is too big (see \Cref{eq:ideal}). Moreover, it is easy to see that in doing so, we are always guaranteeing that for each $a\in [0,1]$, there is a $(p,q)$ in the grid such that $q\le a\le p$.

%
Hence, the main challenge is to show that the size of the built grid is small.  Keeping the dimension of the grid contained is a crucial step, as in the following phase the algorithm has to estimate the expected GFT for each point on the final grid.

Two main aspects have a role in upper bounding the size of the grid. First, the closer we get to the diagonal, the larger the probability of getting a trade with negative GFT. For instance, by \Cref{eq:distanceDiagonal}, we never need $q-p\le T^{1-\frac{1}{3}\beta}$. Indeed, playing such prices, the cumulative regret from negative-GFT would be of the order $T\cdot T^{1-\frac{1}{3}\beta}$. Intuitively, this bounds the depth of the forest which builds the grid (See \Cref{fig: grid construction}).
Second, each point on the grid $\mathcal{F}$ satisfies $(q-p)  \, \mathbb{P}(q\le b \le s \le p)\simeq \alpha$. Hence, since the cumulative probability is bounded by $1$, we can upperbound the width of the forest. In particular,  we split a node at depth $i$ only if probability $\mathbb{P}(q\le b \le s \le p)$ is greater than $\alpha K 2^i$. Hence, there must be at most $\frac{1}{\alpha K 2^{i-1}}$ sets at each level $i$ of the forest or equivalently at distance $\frac{2^{-i+1}}{K}$ from the diagonal.
The key observation is that when we go closer to the diagonal the required probability for splitting a node into two children grows exponentially. This makes the number of nodes per level small, and the forest very sparse.


The first step in the formal analysis of algorithm \texttt{Grid} is to restrict to the event in which all the calls to \texttt{Prob.Est} are successful, \emph{i.e.}, they provide a good approximation of the true probability.
To do so, we define a clean event $\mathcal{E}_1^S$ under which each time the procedure \texttt{Prob.Est} is called inside algorithm \texttt{Grid} the high probability event defined by \Cref{lemma: prob est} is verified.

To compute the probability of the clean event $\mathcal{E}_1^S$ defined as such, it is useful to bound the maximum number of calls to \texttt{Prob.Est}. By construction, the procedure is called for each element in $\bigcup_{i\in [M]}\mathcal{A}_i$. We can bound the cardinality of this set as follows:

\begin{lemma}
\label{lemma: grid stoch naive bound}
    If the Algorithm \texttt{Grid} is initialized with values $K,\alpha,\delta$, then \[\bigcup_{i}|\mathcal{A}_i| \le \frac{2}{\alpha}.\]
\end{lemma}
\begin{proof}
    It is trivial to observe that by construction each set $\Acal_i$ has at most $K2^i$ elements, and that $\{\Acal_i\}_i$ is a disjoint family. Hence, \[\sum_{i=1}^{M}|\Acal_i|\le\sum_{i=1}^{M}K2^{i-1}= \frac{2}{\alpha}.\]
\end{proof}

Then, we can formally define the clean even $\mathcal{E}^S_1$ and lowerbound its probability.

\begin{lemma}
\label{lemma: stoch event 1}
    The event $\mathcal{E}^S_1$ defined as 
    \[\mathcal{E}^S_1 \coloneq\bigcap_{i\in [M]}\bigcap_{(p,q)\in \mathcal{A}_i}\bigg\{\xi(p,q)\le \mathbb{P}\left(q\le s \le p,q\le b \le p)\right) \le \xi(p,q)+ 8\sqrt{\frac{\ln(\frac{4}{\alpha\delta})}{2\lceil(\alpha K 2^{i-1})^{-2}\rceil}}\bigg\}\]
    holds with probability at least $1-\delta$.
\end{lemma}
\begin{proof}
     \Cref{lemma: grid stoch naive bound} guarantees that the clean event $\mathcal{E}^S$ happens with probability at least $1-\delta$, as it is the intersection of at most $\frac{2}{\alpha}$ events, each of probability at least $1-\frac{\alpha \delta}{2}$.
\end{proof}

In the remaining of the section, we will state and prove the three following crucial results on Algorithm \texttt{Grid} under the clean event $\mathcal{E}^S$:
\begin{enumerate}

    \item The dimension of output grid $\mathcal{F}$ is of the order $\BigOL{K+\frac{1}{\alpha K}}$, implying that the resulting grid is not too big to explore and exploit.
    \item The Algorithm \texttt{Grid} employs at most $\BigOL{K+\frac{1}{\alpha K}+\frac{1}{\alpha^2 K}+\frac{1}{\alpha^2 K^3}}$ time steps, implying that the cost payed in term of regret to build the grid is sufficiently small
    \item The difference in terms of GFT between the true optimum $p^*$ and the optimum on the grid is at most $\BigOL{\alpha}$, implying that focusing only on the resulting grid $\mathcal{F}$ leads to a regret at most $\BigOL{\alpha T}$.
    
\end{enumerate}
To prove the first and second result, it is useful to bound the maximum cardinality of all sets $\Acal_i$ for all possible $i\in [M]$.
\begin{lemma}
\label{lemma: Ai bound stoch}
    Under the event $\mathcal{E}^S$ \Cref{alg: grid construction} initialized with parameters $K,\alpha,\delta$ build the family of sets $\{\Acal_i\}_{i=1}^{M+1}$ such that 
    \begin{equation}
        |\Acal_1|=K \quad \textnormal{and} \quad |\mathcal{A}_i|\le \frac{1}{\alpha K2^{i-3}} \quad \forall i\in \left\{2,\ldots,M\right\}
    \end{equation}
\end{lemma}
\begin{proof}
    The general idea of this proof is that as the index $i$ increases, the cardinality of the points in $\mathcal{A}_i$ that surpass the set probability threshold of $\alpha K 2^i$ must decrease due to one key observation: all probabilities should sum to 1.
    
    In particular, let's define $n_i$ as $n_i= \lvert \Acal_i\backslash H_i\rvert$, i.e., the number of elements in $\Acal_{i}$ such that $\xi(p,q)\ge \alpha K 2^{i-1}$, for all $i\in [M]$, and let $i$ be an arbitrary element in $[M-1].$
    
    Under the clean event $\mathcal{E}^S$ we know that for all $(p,q)\in \Acal_i$ it holds that $\xi(p,q)\le \mathbb{P}(q\le s_t\le p, q \le b_t \le p)$, and that by definition of $n_i$ it also holds 
    \[\sum_{(p,q)\in \Acal_i}\xi(p,q) \ge n_i \cdot \alpha K 2^{i-1}.\]
    This implies that $n_i\le \frac{1}{\alpha K 2^{i-1}}$ under the event $\mathcal{E}^S$ and therefore $|\mathcal{A}_{i+1}|=2|\mathcal{A}_i|\le \frac{1}{\alpha K 2^{i-2}} $. 
    This is true for all possible choices of $i\in [M-1]$.
\end{proof}
The first interesting result, the bound on the dimension of the grid $\mathcal{F}$ is easily inferred from \Cref{lemma: Ai bound stoch}.
\begin{lemma}
\label{cor: grid stoch dim}
    Under the event $\mathcal{E}^S$ \Cref{alg: grid construction} initialized with parameters $K,\alpha,\delta$ return a grid $\mathcal{F}$ such that 
    \[|\mathcal{F}|\le K+\frac{4}{\alpha K}\]
\end{lemma}
\begin{proof}
    By construction  $|\mathcal{F}|= |\bigcup_{i=1}^{M} H_i|\le |\bigcup_{i=1}^{M} \Acal_i|$, which, in addition to \Cref{lemma: Ai bound stoch} concludes the proof.
    \[\sum_{i=0}^{M}|\Acal_i|= K + \sum_{i=1}^{M}\frac{1}{\alpha K 2^{i-3}}\le K + \frac{4}{\alpha K}. \]
\end{proof}

While \Cref{lemma: grid stoch naive bound} and \Cref{lemma: Ai bound stoch} aim to bound the same quantity—the maximum dimension of the family of sets $\{\Acal_i\}_i$—their guarantees differ by a multiplicative factor of $\BigOL{1/K}$. This discrepancy arises because \Cref{lemma: grid stoch naive bound} provides a naive, deterministic bound that holds with probability $1$, whereas the refined bound in \Cref{lemma: Ai bound stoch} holds only under the event $\mathcal{E}^S$, which occurs with probability at least $1 - \delta$.

Both results are essential to the analysis: the deterministic bound is critical for defining the clean event $\mathcal{E}^S$, while the high-probability bound enables a tighter upper bound of the regret of \Cref{alg: stoch tot}.

\begin{lemma}
\label{lemma: grid stoc time}
    Under the event $\mathcal{E}^S$, \Cref{alg: grid construction} initialized with parameters $K,\alpha,\delta$ guarantees
     \[T_{grid}\le 4\left(K+\frac{4}{\alpha K}\right) + \frac{4}{\alpha^2 K} + \frac{16}{\alpha^2 K^3},\]
     where $T_{\textnormal{grid}}$ is the number of time steps employed by the algorithm
\end{lemma}
\begin{proof}
    For each element in $\mathcal{A}_i$ \Cref{alg: grid construction} call the procedure \texttt{Prob.Est} with parameter $L=\lceil(\alpha K 2^{i-1})\rceil$, which runs for $4\lceil(\alpha K 2^{i-1})^{-2}\rceil$ time steps, for all $i\in [M]$.
    To compute the total running time we can simply apply \Cref{lemma: Ai bound stoch}, which implies that, under the event $\mathcal{E}^S$
    \begin{align*}
        T_{\textnormal{grid}}&= \sum_{i=1}^{M}|\Acal_i|(4\lceil(\alpha K 2^{i-1})^{-2}\rceil)\\
        &\le 4 \sum_{i=1}^{M}|\Acal_i| + \sum_{i=1}^{M}|\Acal_i|(4(\alpha K 2^{i-1})^{-2})\\
        & \le 4\left(K+\frac{4}{\alpha K}\right) + \frac{4}{\alpha^2 K} + \sum_{i=1}^{\log_2\left(\frac{1}{\alpha K}\right)+1}\frac{4}{\alpha^2 K^3 2^{(i-3)+2i}}\\
        & \le 4\left(K+\frac{4}{\alpha K}\right) + \frac{4}{\alpha^2 K} + \frac{16}{\alpha^2 K^3}.
    \end{align*}
\end{proof}

Finally, we show that our grid provides a good approximation of the optimal GFT.

\begin{lemma}
\label{lemma: grid stoc good approx}
    Under the event $\mathcal{E}^S$ \Cref{alg: grid construction} initialized with parameters $K,\alpha,\delta$ returns a grid $\mathcal{F}$ such that 
    \[\max_{(p,q)\in \mathcal{F}}\gft(p,q) \ge \max_{p\in [0,1]}\gft(p,p)-\alpha \left(1+8\sqrt{\ln\left(\frac{8}{\alpha \delta}\right)}\right)\]
\end{lemma}
\begin{proof}
    Define $p^*\in \arg\max_{p\in[0,1]}\gft(p,p)$ and the couple of prices $(p,q)\in \mathcal{F}$ such that $q\le p^* \le p $. By construction given a $p^*$, $(p,q)$ exists and is univocally defined.
    
    Under the clean event $\mathcal{E}_1^S$ we know that \[(p,q)\in \mathcal{F}\implies \mathbb{P}(q\le s_t \le p, q\le b_t \le p)\le \frac{\alpha}{(p-q)}+\frac{\alpha}{(p-q)}8\sqrt{\ln\left(\frac{8}{\alpha \delta}\right)}.\] Indeed, by construction if $(p,q)\in \Acal_i$ then $p-q=(K2^{i-1})^{-1}$.
    
    Then 
    \[\gft(p^*,p^*)-\gft((p,q))\le (p-q)\mathbb{P}(q\le b_t\le s_t \le p)\le  \alpha \left(1+8\sqrt{\ln\left(\frac{8}{\alpha \delta}\right)}\right).\]
\end{proof}

\subsubsection{GFT Exploration}

The second phase of the algorithm, called GFT~exploration, takes the grid $\mathcal{F}$ built in the previous phase and, for each element in it, estimates the relative GFT. To do so, we use a procedure inspired by other works in the literature (see, e.g. \cite{cesa2021regret}), that allows us to estimate the values $s_t,b_t$ (which are never directly observable). 
We call \texttt{$\gft$-Est.Rep} this algorithm, whose pseudocode of the algorithm in \Cref{alg: gft expl}.
The algorithm takes as input a couple of prices $(p,q)$ and the number of time steps to employ $T_0$, and outputs an estimate $\widehat \gft(p,q)$ of $\gft(p,q)$.

While our approach is based on \cite{cesa2021regret}, our estimator is much more powerful since it is able to estimate the GFT also for prices that are not strongly budget balanced. One step in this direction was made by \cite{bernasconi2024no} that build a slightly biased estimator. Our approach is able to remove the bias by introducing a third component in the estimator. In particular, we notice that: 
\begin{align}
\label{eq: unbiased est}
    \gft_t(p,q)= & \, p \int_{0}^p \frac{1}{p} \mathbb{I}\{ s_t\le x,q\le b\} dx \nonumber \\
    &+  (1-q) \int_{q}^1\frac{1}{1-q} \mathbb{I}\{ s_t\le p,x\le b_t\} dx + (q-p) \mathbb{I}\{s_t\le p,q\le b_t\}
\end{align}   

The first integral can be estimated sampling $p_t$ uniformly from $[0,p]$ and setting $q_t=q$, the second by setting $p_t=p$ and sampling $q_t$ uniformly from $[q,1]$. Finally, the last one can be estimated simply playing $p_t=p$ and $q_t=q$.
Interesting, our sampling procedure allow to have revenue at least $q-p$, which would not be possible sampling uniform from the interval $[0,1]$.

Formally, we prove the following:

\begin{lemma}
\label{lemma: unbiased est}
    Fix any couple of prices $(p,q)\in [0,1]^2$. Then, for any $s_t,b_t\in [0,1]$, if $U\sim \mathcal{U}([0,p])$, $V\sim \mathcal{U}([q,1])$ it holds
    \begin{align*}
        \gft_t(p,q)&= (b_t-s_t)\mathbb{I}\{s_t\le p\}\mathbb{I}\{q\le b_t\}\\
        &=\mathbb{E}[(p)\mathbb{I}\{s_t\le U,q\le b_t\}]+\mathbb{E}[(1-q)\mathbb{I}\{s_t\le p, V\le b_t\}] + \mathbb{E}[ (q-p)\mathbb{I}\{s_t\le p, q\le b_t\}].
    \end{align*}
\end{lemma}
\begin{proof}
    Proving the lemma is equivalent to prove \Cref{eq: unbiased est}.
    To do so, we observe that 
    \[\mathbb{E}_{U\sim \mathcal{U}([0,p])}[p\mathbb{I}\{s_t\le U, q\le b\}]=p \int_{0}^p \frac{1}{p} \mathbb{I}\{ s_t\le x,q\le b\} dx= \frac{p}{p}\left(p-s_t\right)\mathbb{I}\{s_t\le p,q\le b\}\]
    and
    \begin{align*}
    \mathbb{E}_{V\sim ([q,1])}[(1-q)\mathbb{I}\{s_t\le p,V\le b_t\}]&=(1-q) \int_{q}^1\frac{1}{1-q} \mathbb{I}\{ s_t\le p,x\le b_t\} dx\\
    &=\frac{1-q}{1-q}(b_t-q)\mathbb{I}\{s_t\le p,x\le b_t\}.
    \end{align*}
    Hence 
\begin{align*}
    (b_t-s_t)\mathbb{I}\{s\le p\}\mathbb{I}\{q\le b_t\}& = \bigg((p-s_t)+(b_t-q)+(q-p)\bigg)\mathbb{I}\{s_t\le p,q\le b\}\\
    & = p \int_{0}^p \frac{1}{p} \mathbb{I}\{ s_t\le x,q\le b\} dx \nonumber \\
    &+  (1-q) \int_{q}^1\frac{1}{1-q} \mathbb{I}\{ s_t\le p,x\le b_t\} dx + (q-p) \mathbb{I}\{s_t\le p,q\le b_t\}\\
    & = \mathbb{E}[(p)\mathbb{I}\{s_t\le U,q\le b_t\}]\\
    &+\mathbb{E}[(1-q)\mathbb{I}\{s_t\le p, V\le b_t\}] + \mathbb{E}[ (q-p)\mathbb{I}\{s_t\le p, q\le b_t\}].
\end{align*}

\end{proof}


\begin{restatable}{lemma}{stochgftest}
\label{lemma: stoch gft est}
    Let $\delta\in (0,1)$. Procedure \texttt{$\gft$-Est.Rep}, called with input $(p,q)\in [0,1]^2$ and $T_0\in \mathbb{N}$, runs for $T_0$ time steps and returns an estimate $\widehat{\gft}(p,q)$ such that 
    \[\lvert\gft(p,q)-\widehat{\gft}(p,q)\rvert\le 3 \sqrt{\frac{\ln({\frac{2}{\delta})}}{2T_0}},\]
    with probability at least $1-\delta$.
\end{restatable}

Then, we define the event $\mathcal{E}^S_2$ under which all the estimations are accurate and show that this event holds with high probability.

\begin{lemma}
\label{cor: E s 2}
    Define the event $\mathcal{E}^S_2$ as 
    \[\mathcal{E}^S_2 \coloneqq \bigcup_{(p,q)\in \mathcal{F}}\left\{\lvert\gft(p,q)-\widehat{\gft}(p,q)\rvert\le 3 \sqrt{\frac{\ln\left(\frac{2(K+\frac{4}{\alpha K})}{\delta}\right)}{2T_0}}\right\}.\]
    Then, $\mathcal{E}^S_2$ holds with probability at least $1-\delta$ under $\mathcal{E}^S_1$.
\end{lemma}
\begin{proof}
    By \Cref{lemma: stoch gft est}, each inequality  $\lvert\gft(p,q)-\widehat{\gft}(p,q)\rvert\le 3 \sqrt{\frac{\ln\left(\frac{2(K+\frac{4}{\alpha K})}{\delta}\right)}{2T_0}}$ holds with probability at least $1-\frac{\delta}{K+\frac{4}{\alpha K}}$, and under event $\mathcal{E}^S_1$ the number of considered inequalities is $|\mathcal{F}|\le K+\frac{4}{\alpha K}$ by \Cref{cor: grid stoch dim}.
\end{proof}

\subsection{Putting Everything Together}
\Cref{alg: stoch tot} receives as input the maximum contraint violation acceptable $T^\beta$, and the number of time steps $T$.
Based on $T^\beta$ the algorithm set the number of initial points $K$.Indeed on one hand,the higher $K$ the higher the regret will be, as the estimation effort increases significantly, on the other hand the higher the number of initial points $K$ the smaller the constraints violation will be (since these points are closer to the diagonal). 

Hence, the algorithm proposed presents a trade-off between regret and violation. However, this trade-off is quite complex and counterintuitive since there are many moving components and parameters.
For instance, the trade-off impacts only a subregion of the possible regret-violation region and presents some discontinuities (See \Cref{fig: fig 1}). Indeed. the trade-off identified by the algorithm interests only the subspace $(\log_T(R_T),\log_T(V_T))\in [\frac{5}{7},\frac{3}{4}]\times[\frac{3}{4},\frac{6}{7}]$, and in this range we observe a trade-off of type $R_T=\BigOL{T^{1-\beta / 3}},V_T=\BigOL{T^\beta}$ for $\beta\in [3/4,6/7]$. 

Notice that our result is suboptimal for $\beta=3/4$. Indeed, we are getting regret $\tilde O(T^{3/4})$ and violation $\tilde O(T^{3/4})$. This does not match the result of \cite{bernasconi2024no} that get the same regret but zero violation.  However, we can recover their result adding an initial revenue maximization phase similar to the one of \cite{bernasconi2024no}. 
Interestingly, for any other range of parameters, this initial revenue maximization phase is useless since our results are tight.


\begin{restatable}{theorem}{stochtot}
   Fix $\beta\in [3/4,6/7]$ and $\delta>0$. There exists an algorithm that with probability at least $1-3\delta$ guarantees
    \[R_T \le \BigOL{T^{1-\frac{1}{3}\beta}}\quad V_T \le \BigOL{T^{\beta}}.\]
\end{restatable}
\begin{proof}
    We run \Cref{alg: stoch tot} with  input parameters $T,T^\beta,\delta$. 
    
    We can decompose the regret into two components: the regret suffered in the grid construction and exploration phase, and the regret in the exploitation phase.

    We can bound the regret in the first two phases by simply considering the number of rounds spent in these phases.
    By \Cref{lemma: grid stoc time} we know that under the event $\mathcal{E}^S$ the phase of grid construction has length at most \[T_{grid}=4\left(K+\frac{4}{\alpha K}\right)+\frac{4}{\alpha^2K}+\frac{16}{\alpha^2 K^3}\] and that under event $\mathcal{E}^S_1$ the number of rounds used to estimate the GFT on the grid $\mathcal{F}$ is \[T_{gft}=|\mathcal{F}|\cdot T_0\le T_0\left(K + \frac{4}{\alpha K}\right)\].

    Then we can apply under the event $\mathcal{E}^S_1\cap \mathcal{E}^S_2$  \Cref{lemma: grid stoc good approx} and \Cref{lemma: stoch gft est} to obtain, defining the couple of prices $(p,q)\in \mathcal{F}$ such that $p\le p^* \le q$ as $F(p^*)$,
    \begin{align*}
        \sum_{t=T_{grid}+T_{gft}+1}^T \gft(p^*,p^*)-\gft(p_t,q_t)&= \sum_{t=T_{grid}+T_{gft}+1}^T \left(\gft(p^*,p^*)-\gft(F(p^*))\right)\\
        &+ \sum_{t=T_{grid}+T_{gft}+1}^T\left(\gft(F(p^*))-\widehat{\gft}(F(p^*))\right) \\
        & + \sum_{t=T_{grid}+T_{gft}+1}^T \left(\widehat{\gft}(F(p^*))-\widehat{\gft}(p_t,q_t)\right)\\
        & + \sum_{t=T_{grid}+T_{gft}+1}^T\left(\widehat{\gft}(p_t,q_t)-\gft(p_t,q_t)\right)\\
        & \le \left(T-T_{grid}-T_{gft}\right)\cdot \left(\alpha \left(1+8\sqrt{\ln\left(\frac{8}{\alpha \delta}\right)}\right)\right)\\
        &+ 2\left(T-T_{grid}-T_{gft}\right)\cdot 3\sqrt{\frac{\ln\left(\frac{2(K+\frac{4}{\alpha K})}{\delta}\right)}{2T_0}}. 
    \end{align*}

    To bound the difference between the expected GFT and the realizations we apply Hoeffding inequalities obtaining that with probability at least $1-\delta$:
    \begin{equation*}
        \sum_{t=T_{grid}+T_{gft}+1}^T \left(\gft_t(p^*,p^*)-\gft_t(p_t,q_t) \right)-\sum_{t=T_{grid}+T_{gft}+1}^T \left(\gft(p^*,p^*)-\gft(p_t,q_t) \right) \le \sqrt{T\frac{\ln(\frac{2}{\delta})}{2}}.
    \end{equation*}

    Hence, since by \Cref{lemma: stoch event 1} and \Cref{cor: E s 2} event $\mathcal{E}^S_1\cap\mathcal{E}^S_2$ is verified with probability at least $1-2\delta$ then with probability at least $1-3\delta$
    \begin{align*}
        R_T &\le T_{grid}+T_{gft}+\sum_{t=T_{grid}+T_{gft}+1}^T \left(\gft_t(p^*,p^*)-\gft_t(p_t,q_t) \right)\\
        & \le \BigOL{K+\frac{1}{\alpha K}+\frac{1}{\alpha^2 K}+\frac{1}{\alpha^2K^3}+T_0K+\frac{T_0}{\alpha K}+\sqrt{T}+\alpha T+ \frac{T}{\sqrt{T_0}}}.
    \end{align*}
    Moreover, it is easy to see that all the coupes $(p,q)$ that we play guarantee $q-p\le \frac{1}{K}$ deterministically and hence 
    \[V_T \le \frac{T}{K}.\]

    Substituting for $\gamma\in [3/4,6/7]$ $T_0=T^{\frac{2}{3}\beta},K=T^{1-\beta},\alpha=T^{-\frac{1}{3}\beta}$ it becomes that with probability at least $1-3\delta$:
    \begin{align*}
        \begin{cases}
            &R_T \le \BigOL{T^{1-\frac{1}{3}\beta}+T^{\frac{5}{3}\beta-1}+T^{2\beta-1}}\le \BigOL{T^{1-\frac{1}{3}\beta}}\\
            & V_T\le T^{\beta}
        \end{cases}
    \end{align*}
    since $\frac{5}{3}\beta-1\le1-\frac{1}{3}\beta$ is true for all $\beta\le1$ and $2\beta-1\le 1-\frac{1}{3}\beta$ is true if $\beta \le \frac{6}{7}$ .
    Notice that the condition $\beta\le\frac{6}{7}$ is perfectly coherent with the discontinuity in the lower bound, for which allowing violation greater than $T^{6/7}$ does not improve the regret bound.

\end{proof}

One interesting special case of the theorem holds for $\gamma=\frac{6}{7}$.  Indeed, \Cref{alg: stoch tot} attains regret $R_T \le \BigOL{T^{5/7}}$.
This matches the lower bound of \cite{bernasconi2024no} for settings without budget balance.

\begin{corollary}
   Fix $\delta>0$. There exists an algorithm that with probability at least $1-3\delta$ guarantees
    \[R_T \le \BigOL{T^{5/7}}\quad V_T \le \BigOL{T^{6/7}}.\]
\end{corollary}


\section{Adversarial Valuations} \label{sec:adv}

In this section, we extend the previous approach to handle adversarial valuations.
We begin by outlining the main challenges introduced by the adversarial nature of the valuations and the key ideas used to address them in \Cref{sec:advChal}. Then, in \Cref{sec:sleeping}., we present an interlude on sleeping experts. In particular, we provide a novel algorithm that guarantees low dynamic regret in this setting. This will be essential to build a regret minimizer for the adaptive grid that we build at runtime.
We then present the algorithm in detail and analyze its theoretical guarantees in \Cref{sec:advAlg}.

\subsection{Main Challenges} \label{sec:advChal}

We build  on the algorithm for stochastic valuations. 
However, \cref{alg: stoch tot}, is not directly applicable to adversarial environments. Indeed, it relies on a standard exploration-then-exploitation paradigm: it first constructs an adaptive discretization grid based on early observations, then estimates the gain-from-trade (\gft) at each grid point, and finally exploits this information in the remainder of the rounds.

It is well known that algorithms that work in phases fail in adversarial environments. For instance, the initial observations used to build the grid and estimate the GFT can be misleading, mimicking a stationary distribution. Once the algorithm enters the exploitation phase, the adversary can abruptly shift the distribution, causing a significant increase in regret. This behavior is analogous to the failure of uniform exploration multi-armed bandit algorithms in adversarial settings (see, \emph{e.g.}, \cite{slivkins2019introduction}).

This issue has been addressed in prior works using block decomposition techniques (see, e.g., \cite{bernasconi2024no} for an example related to bilateral trade), which convert stochastic algorithms into ones that are more robust in adversarial environments while preserving some of their original performance guarantees. This approach is effective to make the GFT estimation robust to adversarial environments. In particular, the $T$ rounds are divided into $N$ blocks and, in each block $j\in [N]$, some random rounds are used to build an estimation of, for instance, the average GFT in block $j$.

However, this approach alone is inadequate in our setting. To see that, consider the natural extension of the adaptive grid to the adversarial setting in which we replace the probabilities with the empirical frequency. In this scenario, the grid, \emph{i.e.}, the set of arms, is built at runtime. This makes classical bandit algorithms and previous learning algorithms for bilateral trade Ineffective, since the set of ``arms'' in the underlying bandit structure evolve over time.
This introduce two additional challenges: i) build the grid at runtime and ii) design a learning algorithm for settings in which the available arms changes over time.

We address this challenges building a different grid $\Fcal_{j}$ for each block $j \in [N]$.
Then, we bound the regret decomposing it into two components which can be bounded separately: 
\begin{itemize}
\item \textbf{grid regret}, related to the suboptimality of our sequence of grids $\Fcal_j$ 
\item \textbf{dynamic regret}, related on the regret with respect to play the optimal sequence of prices on the grid (with some restrictions). 
\end{itemize}

\paragraph{Grid Regret}
The first component, which we refer to as the \emph{grid regret}, quantifies the cumulative loss incurred due to the discretization of the price set. Specifically, it measures the cumulative difference between the GFT achieved at the true optimal price pair and the best available approximation on the discretization grids in blocks $j\in [N]$.

Formally, consider a time step $t$ in a block $j$.
Then, we let $F_t(p^*)=(p_t^*, q_t^*) $ denote the couple $(p, q) \in \Fcal_{j}$ such that $q \leq p^* \leq p$. Notice that by construction, there will be only one of such prices.
Intuitively, $(p_t^*, q_t^*)$ denotes the grid point in $\Fcal_j$ that better approximates the optimal pair $(p^*, p^*)$.
Then, the grid regret is defined as:
\[
R^{\textnormal{grid}} = \sum_{t=1}^T \gft_t(p^*, p^*) - \sum_{t=1}^T \gft_t(F_t(p^*)).
\]

To ensure low grid regret, the grid construction uses a mechanism inspired by the stochastic setting, but that works at runtime considering the empirical frequency of getting a negative trade.
This is combined with a block decomposition strategy that incorporates randomized mappings, which allows us to get an unbiased estimator of the empirical frequency.

\paragraph{Dynamic Regret}
The second component captures the cumulative difference between the GFT achieved by the algorithm and that of the best  sequence of grid-based approximation of $p^*$.
This component is exactly the dynamic regret with respect to the sequence of  price pairs $\{(p_t^*,q_t^*)\}_{t=1}^T$.

Formally, it is defined as:
\[
R^{\textnormal{Dyn}}(\{(p_t^*,q_t^*)\}_{t=1}^T) = \sum_{t=1}^T \gft_t(p_t^*, q_t^*) - \sum_{t=1}^T \gft_t(p_t, q_t).
\]

Clearly, dynamic regret might be linear in general. Hence, to provide any meaningful result, we  must guarantee that the sequence $\{(p_t^*,q_t^*)\}_{t=1}^T$ is well-behaved. To do so, we observe that in our problem this sequence change prices, \emph{i.e.}, action, a small number of times.
This holds since $p^*_t$ changes only when we replace $(p_t^*,q_t^*)$ with two child nodes (see \Cref{fig: fig 1} Right for an example in the stochastic case) and that, by construction, the depth of the forest is bounded by $\log_2(\frac{4N}{K\alpha})$.

To ensure low sequence regret, we adopt techniques inspired by \cite{cesa2012mirror} to sleeping experts. This will allow to obtain novel regret bounds which depend only the number of changes of action. This improves the classical guarantees on sleeping bandits \cite{kleinberg2010regret} in ``slightly-nonstationary'' environments.

\subsection{Dynamic Regret in Sleeping Experts} \label{sec:sleeping}
In this section, we discuss some novel analysis relative to the sleeping experts problem (\cite{kleinberg2010regret}), with the goal to provide no-dynamic regret for this problem.

\begin{algorithm}[H] \caption{\texttt{Dynamic-Sleeping-Expert}}\label{alg: seq omd fix-sha}
    \begin{algorithmic}[1]
    \State \textbf{input:} number of rounds $T$, set of arm $\Acal$
    \State $\eta \gets \sqrt{\frac{\ln(|\Acal|T)}{T}}$
    \State $\gamma \gets \frac{1}{T}$
        \State Initialize $x_0(a)=\frac{1}{|\Acal|}$ for all $a\in  \Acal$
        \State set $\tilde \ell_0(a)=0$ for all $a\in  \Acal$
        \For{$t=1,\ldots,T$}
        \State Observe $\Acal_{t}$
        \State For all $a\in \mathcal{A}$
        \[\hat{x}_{t}(a)\gets \frac{\bar x_{t-1}(a)e^{-\eta\tilde \ell_{t-1}(a)}}{\sum_{a'\in \Acal_{t}}\bar x_{t-1}(a')e^{-\eta \tilde \ell_{t-1}(a')}} \]
        \State For all $a\in \mathcal{A}$
        \[\bar x_{t-1}(a)\gets \frac{\gamma}{|\mathcal{A}|}+(1-\gamma)\hat{x}_{t-1}(a)\]
        \State For all $a \in \Acal$ \label{line:project}
        \begin{align*}
             x_t(a)\gets  \begin{cases}
             \frac{\bar x_{t}(a)}{\sum_{a \in \Acal_{t}} \bar x_{t}(a)}\quad & \textnormal{if}\ a\in \Acal_t
            \\  0 & \textnormal{otherwise}
        \end{cases} 
        \end{align*}
        \State Play $a_t\sim x_t$
        \State Observe $\ell_t(a)$ for all $a\in \Acal_{t}$
        \State for each arm $a$ \label{line:extend}
        \begin{align*}
            \tilde \ell_t(a)\gets  \begin{cases}
             \ell_t(a)\quad & \textnormal{if}\ a\in \Acal_t
            \\  1 & \textnormal{otherwise}
        \end{cases}
        \end{align*}
        \EndFor
    \end{algorithmic}
\end{algorithm}

\paragraph{Sleeping Experts Setting}
At each round $t\in [T]$, a learner observe a set of available actions $\Acal_t \subseteq \Acal$, where $\Acal=\cup_{t=1}^T \Acal_t$ is the set of all actions. The set of actions $\Acal$ is known to the learner. Then, the learner chooses an action $a_t$ from a set of action $\Acal_t$.
Finally, a loss function $\ell_t:\Acal_t\rightarrow [0,1]$ is observed by the learner. We assume that the sequence of loss functions is built by an oblivious adversary.
One of the most considered performance metric is the policy regret, where a policy consists in the best ordering of all possible experts. Then, the baseline receives the loss of the first available expert in each round, following the policy ordering.
Here, we focus on a different baseline, and in particular a sequence of actions $\mathbf{a}^*=(a^*_t)_{t \in [T]}$, such that $a^*_t\in \mathcal{A}_t$ for each $t$.
Our goal is to bound the regret with respect to a sequence $\mathbf{a}^*$ as a function of the number of switches $S(\mathbf{a}^*)=1+\sum_{t=1}^{T-1} I\{a_t^*\neq a_{t+1}^* \}$. In particular, we aim at bounding the dynamic regret:
\[  R^{\textnormal{Dyn}}(\mathbf{a}^*) = \sum_{t=1}^T \ell_t(a^*_t) - \sum_{t=1}^T \ell_t(a_t)   \]
as a function of $S(\mathbf{a}^*)$.

\paragraph{Special case: Dynamic Regret on the Sequence of Grids}
Playing on a sequence of changing grids $\{\mathcal{F}_j\}_{j=1}^N$ while competing with the corresponding sequence of optimal grid approximations $\{F_j(p^*)\}_{j=1}^N$ can be reformulated as an instance of a \emph{sleeping bandit problem}. In this setting, the set of all possible actions is denoted by $\mathcal{A}$, where each action $a \in \mathcal{A}$ can be uniquely identified with a price pair $(p, q)$. The action space satisfies $\mathcal{A} \supseteq \bigcup_{j \in [N]} \mathcal{F}_j$, and at each time step $j \in [N]$, the set $\mathcal{F}_j$ defines the available (active) actions. The time horizon is $N$, and the loss associated with an action $a$ is assumed to be proportional to $-\gft(p, q)$, where $(p, q)$ corresponds to the price pair represented by $a$. 

Under this formulation, the challenge of selecting actions from dynamically changing grids—while staying competitive with the optimal choices at each round—can be reduced to a sleeping expert problem studied through the framework of \emph{dynamic regret}.

\paragraph{Our Algorithm}

The pseudocode of our algorithm is in \Cref{alg: seq omd fix-sha}. Our algorithm uses as a black box the algorithm of \cite{cesa2012mirror}, which guarantees no dynamic regret in a standard expert problem. To extend the algorithm to the sleeping expert problem, we introduce the following two modifications: i) we set the loss of arms not in $\Acal_t$ to $1$ (see Line \ref{line:extend}) and ii) we project the randomized decision on the simplex over $\cA_t$ (see Line \ref{line:project}).

Our result reads as follows.

\begin{restatable}{theorem}{lemmaSleeping}
\label{lemma: sleeping}
    \Cref{alg: seq omd fix-sha} is such that for all sequences of arms $\mathbf{a}=\{a^*_t\}_{t=1}^T$ such that $a_t^*\in \Acal_t$ with probability at least $1-\delta$, it holds
    \[R^{\textnormal{Dyn}}(\mathbf{a})\le \tilde O\left(S(\mathbf{a}) \sqrt{\log\left(\frac{|\Acal|}{\delta}\right)T}\right).\]
\end{restatable}

\begin{proof}
    For the ease of notation, we will use $\ell(x)=\mathbb{E}_{a\sim x}[\ell(a)]$ to indicate the expected loss generated by the randomized decision $x\in \Delta_{\Acal}$, and we use a similar notation for the fictitious loss $\tilde \ell(x)$. 
    
    Notice that the sequence of strategies $\bar x_{t}$ are exactly the one employed in Theorem 3 of \citep{cesa2012mirror}, where we use the fictitious loss $\tilde \ell$. 
    Hence, given a sequence of arms $\mathbf{a}=(a_t^*)_{t \in [T]}$ it is the case that
    \[ \sum_{t \in [T]} \left(\tilde \ell_t(a^*_t)-\tilde\ell_t(\bar x_t))\le  \tilde O(S(\mathbf{a}) \sqrt{\log(|\Acal|)T}\right). \]

    Moreover, notice that  $\tilde \ell_t(\bar x_t)\ge \tilde \ell_t(x_t)$, since we are moving probability from arms not in $\Acal_t$ (with artificial loss $1$) to arm in $\Acal_t$ (with loss at most $1$).
    Hence:
     \[ \sum_{t \in [T]} (\tilde \ell_t(a^*_t)-\tilde \ell_t(x_t))\le  \tilde O(S(\mathbf{a}) \sqrt{\log(|\Acal|)T}). \]
     Noticing that for all the actions with positive probability $\tilde \ell_t(a)= \ell_t(a)$ we get 
      \[ \sum_{t \in [T]} ( \ell_t(a^*_t)-\ell_t(x_t))\le  \tilde O(S(\mathbf{a}) \sqrt{\log(|\Acal|)T}), \]
      and fillay we conclude the proof by applying Azuma-Hoeffdimg inequality, obtaining that with probability at least $1-\delta$ it holds
      \[\sum_{t \in [T]} ( \ell_t(a^*_t)- \ell_t(a_t))\le  \tilde O\left(S(\mathbf{a}) \sqrt{\log\left(\frac{|\Acal|}{\delta}\right)T}\right).\]
\end{proof}

\paragraph{On the Inefficiency of \Cref{alg: seq omd fix-sha}}
In order to simplify the exposition, we presented a computationally inefficient implementation of \Cref{alg: seq omd fix-sha} which update the loss for the set of arm $\Acal$, and not only the set of arms which appears at least ones. This will be particularly important in the next section, where $|\cup_t \Acal_t|\ll |\Acal|$.
Despite that, it is easy to design a version of the algorithm that runs in time $\poly(|\cup_t \Acal_t|)$. For instance, one simple approach could be to apply a doubling trick on the number of arms. To do so, set $|\Acal'|=2 |\Acal|$ every time the maximum number of arms $|\Acal|$ is reached, and reset the algorithm with a new set of arms of cardinality $|\Acal'|$.\footnote{Notice that our algorithm could equivalently take as input the number of arms $|\Acal|$ instead of the set of arms $\Acal$, since arms yet to appear are indistinguishable.}

\subsection{Algorithm: Adversarial Setting} \label{sec:advAlg}

This section is organized as follows. We begin by introducing the general functioning and features of the algorithm. Then we present and discuss the merits of the two auxiliary functions \texttt{Ind.Est} and \texttt{\gft.est}, studying also how the total dimension of the adaptive grid can be bounded. Finally we provide and prove the desired theoretical guarantees of \Cref{alg: adv}.

\subsubsection{Structure of the Algorithm}
\Cref{alg: adv} is based on a block decomposition structure. The algorithm divides the $T$ time steps in $N$ equally long subsets of rounds. Then, through the use of random mappings the algorithm assigns each round in the block to three possible functions:
\begin{enumerate}
    \item Grid update
    \item GFT estimation
    \item exploitation
\end{enumerate}

\paragraph{Grid Update (from Line \ref{line: adv grid 1} to Line \ref{line: adv grid 2})} 
The algorithm begins with a uniform grid of $K$ points \(\mathcal{F}_0 = \left\{ \left(\frac{w+1}{K}, \frac{w}{K} \right) : w = 0, \ldots, K-1 \right\}\), similar to the initialization in the stochastic case. In each block, the algorithm selects one round for every active point on the grid and uses it to invoke the \texttt{Ind.Est} function. The goal is to obtain an optimistic estimate of the number of times the valuations fall within the triangular region defined by the price pair \((p, q)\) and the diagonal—that is, the event \(\mathbb{I}_t(q \leq b_t \leq s_t \leq p)\). See \Cref{fig: fig 1} for a graphical representation.
The mechanism follows the same principle as in the stochastic setting. The algorithm starts with a tolerance parameter \(\alpha\), and iteratively updates an optimistic estimate of the \emph{negative} gain-from-trade (GFT) contribution for each active price pair. If this estimate exceeds the threshold \(\alpha\) (i.e., the condition at Line \cref{line: adv condition grid}), the corresponding price pair \((p, q)\) is discarded. Two new price pairs are then added to the grid:  
\(
\left\{ \left(p - \frac{p - q}{2}, q \right), \left(p, q + \frac{p - q}{2} \right) \right\}.
\) 
As in the stochastic case, deeper generations—those closer to the diagonal \(p = q\)—require higher frequency of valuations in the associated negative triangle to trigger the discard condition.

Notice that differently from the stochastic setting, this update of the grid is uniquely based on the empirical frequency and it is distributed over the whole time horizon.

\paragraph{Gain For Trade Estimate (from Line \ref{line: adv gft 1} to Line \ref{line: adv gft 2})}
In each block, the algorithm selects one random round for every active point on the grid and uses it to invoke the \texttt{\gft.est} function, which returns an unbiased estimate of the true \gft corresponding to the selected price pair.

\paragraph{Exploitation (from Line \ref{line: adv expl. 1} to Line \ref{line: adv expl. 2} and Line \ref{line: adv expl. 3} and Line \ref{line: adv expl. 4})} 
At the end of each block, the algorithm records the updated set of active price pairs on the grid along with their corresponding estimated \gft values. It then applies the algorithm for sleeping expert designed in the previous section to compute the couple of prices to be played in the following block. 
Notice that since the block decomposition approach build an estimation of all the available arm, we artificially have full feedback in this step.

\begin{algorithm}[]\caption{} \label{alg: adv}
\small
    \begin{algorithmic}[1]
        \State Input: $T,\beta,\delta$
        \State Set: $K\gets T^{1-\beta},N\gets T^{\frac{2}{3}\beta},\alpha\gets T^{\frac{1}{3}}$ 
        \State Divide the $T$ episodes in $N$ blocks
        \[\mathcal{B}_j \gets \bigg\{(j-1)\frac{T}{N}+1,\ldots, j\frac{T}{N}\bigg\}\quad \forall j\in [N]\]
        \State Initialize $\mathcal{F}_1 \gets \left\{\left(\frac{w+1}{K},\frac{w}{K}\right): w= 0,\ldots,K-1\right\}$ 
        \State $\ell_0(a)=0$ for each $(p,q)\in \Fcal_1$
        \State $\hat n_0(p,q)\gets 0 $ for all $(p,q)\in \mathcal{F}_1$
        \For{$j= 1,\ldots,N$}
        \State $(p_j,q_j)\gets$ action returned by \texttt{Dynamic-Sleeping-Expert} with set of active arm $\Fcal_{j}$ and loss $\ell_{j-1}(\cdot)$ \label{line: adv expl. 3}
        \State Initialize two random mappings $f_{j}:\mathcal{F}_{j} \rightarrow \mathcal{S}^f_{j} \subseteq \mathcal{B}_{j}$, $g_{j}:\mathcal{F}_j \rightarrow \mathcal{S}^g_j \subseteq \mathcal{B}_j$
        \For{$t\in \mathcal{B}_j$}
        \If{$t \notin \mathcal{S}_j^f \cup \mathcal{S}_j^g$}\label{line: adv expl. 1} \Comment{Exploitation}
        \State Play $(p_t,q_t)=(p_j,q_j)$ \label{line: adv expl. 2}
        \ElsIf{$t\in \mathcal{S}_j^f$}\label{line: adv grid 1} \Comment{Grid update}
        \State select $(p,q)$ such that  $f_j((p,q))=t$
        \State $i \gets \log_2(\frac{1}{(p-q)K})$ \hspace{1cm}/*\textit{i is the depth of $(p,q)$, i.e. $p-q=\frac{1}{K2^i}$}*/
        \State $I_j(p,q)\gets\cdot\texttt{Ind.Est}((p,q))$
        \State $\hat n_j(p,q) \gets \hat{n}_{j-1}(p,q)+ I_{j}(p,q)$ 
        \State $\underline{n}_j(p,q)\gets  \widehat{n}_j(p,q) - 4\sqrt{\frac{\ln(\frac{2T}{\delta})}{2}N}$
        \If{$\underline{n}_j(p,q)>2^i K \alpha$}\label{line: adv condition grid} 
        \State \(\mathcal{F}_{j+1}\gets \mathcal{F}_{j+1} \cup \left\{\left(p-\frac{2^{-(i+1)}}{K},q\right),\left(p,q+\frac{2^{-(i+1)}}{K}\right)\right\}\)
        \State $\hat{n}_j\left(p-\frac{2^{-(i+1)}}{K},q\right)\gets 0,~ \hat{n}_j\left(p,q+\frac{2^{-(i+1)}}{K}\right)\gets 0$
        \Else 
        \State $\mathcal{F}_{j+1}\gets \mathcal{F}_{j+1} \cup \{(p,q)\}$
        \EndIf \label{line: adv grid 2}
        \ElsIf{$t\in \mathcal{S}_j^g$}\label{line: adv gft 1} \Comment{GFT estimation}
        \State select $(p,q)$ such that  $g_j((p,q))=t$
        \State $\widehat{\gft}_j(p,q)\gets \texttt{gft.est}(p,q)$ 
        \EndIf \label{line: adv gft 2}
        \EndFor
        \State \[\ell_j(p,q)\gets \frac{3-\widehat{\gft}_j(p,q)}{6} \quad \forall (p,q)\in \mathcal{F}_{j}\]
        \label{line: adv expl. 4}
        \EndFor
        
    \end{algorithmic}
\end{algorithm}

\subsubsection{Auxiliary Functions: GFT and Density Estimation }
In this section, we introduce the general functioning of the procedures used to estimate the gain from trade and the frequency with which valuations fall within specific regions of the space $[0,1]$, namely \texttt{GFT.Est} and \texttt{Ind.Est}. We also discuss the theoretical guarantees associated with these procedures, as well as the implications for the performance of \Cref{alg: adv}, which relies on them.

The procedure \texttt{GFT.Est} is based on the same decomposition of the gain from trade (GFT) used in the stochastic setting, namely:
\[
\gft_t(p,q) = p\,\mathbb{E}[\mathbb{I}\{s_t \leq U,\ q \leq b_t\}] + \mathbb{E}[(1 - q)\,\mathbb{I}\{s_t \leq p,\ V \leq b_t\}] + \mathbb{E}[(q - p)\,\mathbb{I}\{s_t \leq p,\ q \leq b_t\}],
\]
where \( U \sim \mathcal{U}([0, p]) \) and \( V \sim \mathcal{U}([q, 1]) \).  
The key difference in the adversarial setting is that the estimation is performed in a single time step via a randomized procedure, rather than relying on repeated sampling as in the stochastic case (the detailed procedure is reported in \Cref{app: adv}).

\begin{restatable}{lemma}{gftEstAdv}\label{lemma: adv gft est}
    For all  $(p,q)\in [0,1]^2$ and for all $t\in [T]$, if the procedure \texttt{GFT.Est} is employed at round $t\in [T]$ with input $(p,q)$ it holds
    \[ \mathbb{E}[\texttt{GFT.Est}(p,q)]= \gft_t(p,q)\]
\end{restatable}
\begin{proof}
    The proof can be directly derived from \Cref{lemma: stoch gft est}.
\end{proof}

Since \Cref{alg: adv} employs the procedure \texttt{GFT.Est} to construct the estimates $\widehat{\gft}_j(p,q)$, \Cref{lemma: adv gft est} can be used to show that these estimates are unbiased estimators of the empirical average of the generated gain for trade ($\gft$) by $(p,q)$, restricted to the $j$-th block $\mathcal{B}_j$.

\begin{corollary}
\label{cor: gft adv}
    \Cref{alg: adv} is such that for all blocks $j\in [N]$ and for all $(p,q)\in \mathcal{F}_j$ it holds
    \[\mathbb{E}[\widehat{\gft}_j(p,q)]=\frac{1}{|\mathcal{B}_j|}\sum_{t\in |\mathcal{B}_j|}\gft_t(p,q)\]
\end{corollary}
\begin{proof}
    Thanks to the random mapping $g_j$ for each $j\in [N]$ and for all $(p,q)\in \mathcal{F}_j$ the probability of calling the procedure \texttt{GFT.Est} with input $(p,q)$ is $1/|\mathcal{B}_j|$, and by \Cref{lemma: adv gft est} \texttt{GFT.Est} is unbiased estimator of the actual GFT, which concludes the proof.
\end{proof}

We focus now on the other procedure \texttt{Ind.Est}, and the dynamical construction of the grids $\{\mathcal{F}_j\}_{j=1}^N$ by \Cref{alg: adv}.

The adaptive grid is designed to estimate the probability of the region below the diagonal defined by each price pair \((p, q)\), corresponding to the event \((q \leq b \leq s \leq p)\), where \(s\) and \(b\) denote the seller's and buyer's valuations, respectively (see the red area in \Cref{fig: grid construction}).  
In practice, we estimate the broader event \((q \leq b_t \le p, q \leq s_t \leq p)\), which contains the target event but is more tractable.

In the stochastic setting, we directly estimate the probability \(\mathbb{P}_{(s,b) \sim \mathcal{P}}(q \leq b \leq p, q\leq s \leq p)\). In contrast, in the adversarial setting—where no distributional assumptions can be made—we aim to construct an unbiased estimator for the indicator function \(\mathbb{I}(q \leq b_t \le p, q \leq s_t \leq p)\) using the procedure \texttt{Ind.Est}.

To estimate this indicator, \texttt{Ind.Est} exploits the following four-part decomposition:
\[
\mathbb{I}(q\le s_t \le p,q\le b_t\le p)= \mathbb{I}(s_t\le p, q \le b_t) - \mathbb{I}(s_t\le q, q \le b_t) - \mathbb{I}(s_t\le p, p \le b_t) + \mathbb{I}(s_t\le q, p \le b_t),
\]
and uses a randomized procedure to estimate all four components simultaneously in a single time step.

This mechanism ensures that the optimum over the grid families \(\{\mathcal{F}_j\}_{j=1}^N\) is not significantly worse than the true optimum over the continuous space along the diagonal.

\begin{lemma}
\label{lemma: adv ind est}
    For every possible $(p,q)\in [0,1]^2$ such that $q\le p$, and for all $t\in [T]$, if the procedure \texttt{Ind.Est} is called with input $(p,q)$ at the time step $t$ then
    \[\mathbb{E}\left[\texttt{Ind.Est}(p,q)\right]= \mathbb{I}(q\le s_t \le p,q\le b_t\le p)\]
\end{lemma}
\begin{proof}
    The lemma is easily proven thanks to the following equation
    \[\mathbb{I}(q\le s_t \le p,q\le b_t\le p)= \mathbb{I}(s_t\le p, q \le b_t)-\mathbb{I}(s_t\le q, q \le b_t)-\mathbb{I}(s_t\le p, p \le b_t)+\mathbb{I}(s_t\le q, p \le b_t).\]
\end{proof}
Since \Cref{alg: adv} employs the procedure \texttt{Ind.Est} to construct the estimates $I_j(p,q)$, \Cref{lemma: adv ind est} can be used to show that these estimates are unbiased estimators of the empirical average of the target indicator functions $\mathbb{I}\{q\le s_t \le p,q\le b_t\le p\}$, restricted to the $j$-th block $\mathcal{B}_j$.

\begin{corollary}
\label{cor: adv ind est}
    \Cref{alg: adv} is such that, for all blocks $j\in [N]$ and for all $(p,q)\in \mathcal{F}_j$ it holds
    \[\mathbb{E}[I_j(p,q)]=\frac{1}{|\mathcal{B}_j|}\sum_{t\in |\mathcal{B}_j|}\mathbb{I}(q\le s_t \le p,q\le b_t\le p)\]
\end{corollary}
\begin{proof}
    Thanks to the random mapping $f_j$ for each $j\in [N]$ and for all $(p,q)\in \mathcal{F}_j$ the probability of calling the procedure \texttt{Ind.Est} with input $(p,q)$ is $1/|\mathcal{B}_j|$, and by \Cref{lemma: adv ind est} \texttt{Ind.Est} is unbiased estimator of the actual identity $\mathbb{I}(q\le s_t \le p,q\le b_t\le p)$, which concludes the proof.
\end{proof}
Similarly, we can establish a cumulative version of \Cref{cor: adv ind est}. Specifically, we show that the estimates \(\hat{n}_j(p,q)\) built by \Cref{alg: adv} are unbiased estimators of the cumulative sum of the block-wise empirical means of the relevant indicator functions over all blocks up to block \(j\), restricted to those blocks \(k \leq j\) where the price pair \((p,q)\) was active, i.e., \((p,q) \in \mathcal{F}_k\). Formally:

\begin{corollary}
    \Cref{alg: adv} is such that, for all blocks $j\in [N]$ and for all $(p,q)\in \mathcal{F}_j$ it holds
    \[\mathbb{E}[\hat n_j(p,q)]=\frac{N}{T}\sum_{j=0}^k\mathbb{I}((p,q)\in \mathcal{F}_j)\sum_{t\in \mathcal{B}_j}\mathbb{I}(q\le s_t \le p,q\le b_t \le p)\]
\end{corollary}
\begin{proof}
    The corollary is a direct consequence of \Cref{cor: adv ind est}, after considering that by construction $|\mathcal{B}_j|=\frac{T}{N}$ for all $j\in [N]$.
\end{proof}

Moreover, we can exploit the bounded support of \( I_j(p,q) \in [-4,4] \) to apply the Azuma-Hoeffding inequality and derive concentration bounds. This allows us to construct optimistic estimates \(\underline{n}_j(p,q)\), and to show that they lie within an \(\BigOL{\sqrt{N}}\)-neighborhood of the expected value \(\mathbb{E}[\hat{n}_j(p,q)]\).

\begin{restatable}{lemma}{lemmaProbAdv}
\label{lemma: n hat adv}
    \Cref{alg: adv} is such that for all $k\in [N]$, $\forall (p,q)\in \mathcal{F}_k$ it holds with probability at least $1-\frac{\delta}{T}$
    \begin{equation*}
        \underline{n}_k(p,q)\le \frac{N}{T}\sum_{j=1}^k\mathbb{I}((p,q)\in \mathcal{F}_j)\sum_{t\in \mathcal{B}_j}\mathbb{I}(q\le s_t \le p,q\le b_t \le p),
    \end{equation*}
    and 
    \begin{equation*}
        \underline{n}_k(p,q) \ge \frac{N}{T}\sum_{j=1}^k\mathbb{I}((p,q)\in \mathcal{F}_j)\sum_{t\in \mathcal{B}_j}\mathbb{I}(q\le s_t \le p,q\le b_t \le p) - 8\sqrt{\frac{\ln(\frac{2T}{\delta})}{2}N}.
    \end{equation*}
\end{restatable}
The proof of \Cref{lemma: n hat adv} can be found in appendix.

\begin{lemma}
    Define the event $\mathcal{E}^A_1$ as

    \begin{align*}
        \mathcal{E}^A_1\coloneq\bigcup_{k\in [N],(p,q)\in \mathcal{F}_k}\bigg\{\frac{N}{T}\sum_{j=1}^k&\mathbb{I}((p,q)\in \mathcal{F}_j)\sum_{t\in \mathcal{B}_j}\mathbb{I}(q\le s_t \le p,q\le b_t \le p) - 8\sqrt{\frac{\ln(\frac{2T}{\delta})}{2}N}\\
        &  \hspace{3cm}\le\underline{n}_k(p,q) \le \\
        &\hspace{0.5cm}  \frac{N}{T}\sum_{j=1}^k\mathbb{I}((p,q)\in \mathcal{F}_j)\sum_{t\in \mathcal{B}_j}\mathbb{I}(q\le s_t \le p,q\le b_t \le p)\bigg\}.
        \end{align*}
        Then $\mathcal{E}^A_1$ holds with probability at least $1-\delta$.
\end{lemma}
\begin{proof}
    By \Cref{lemma: n hat adv} $\mathcal{E}^A_1$ is the union of at most $T$ events, each of probability at least $1-\frac{\delta}{T}$.
\end{proof}

Then, working with the empirical frequency instead of the probabilities, similarly to the result stated in  \Cref{cor: grid stoch dim} for the stochastic case we get the analogous lemma in the adversarial one:

\begin{restatable}{lemma}{LemmaGridBoundAdv}
\label{lemma: adv mathcalN}
    Consider \Cref{alg: adv} with input $T,\beta,\delta$ and define $\mathcal{N}= \lvert \bigcup_{j=1}^{N}\mathcal{F}_j\rvert$ then under the event $\mathcal{E}^A_1$
    \[\mathcal{N}\le K + 4\frac{N}{\alpha K}\]
\end{restatable}
\begin{proof}
    The proof is based on the same idea of \Cref{cor: grid stoch dim} that the number of prices $(p,q)$ such that for some block $j$  they have to be discarded in favour of two new couple of prices, decrease exponentially with the depth of the tree, thank to the splitting condition being $\underline{n}_j(p,q)>2^iK\alpha $.
    The complete proof can be founded in \Cref{app: adv}.
\end{proof}
\begin{corollary}
\label{cor: adv T}
    Consider \Cref{alg: adv} with input $T,\beta,\delta$, under the event $\mathcal{E}^A_1$  \[\sum_{j=1}^N|\mathcal{F}_j|\le NK + 4\frac{N^2}{\alpha K} \]
\end{corollary}
\begin{proof}
    The proof trivially descend from \Cref{cor: adv T} since $\mathcal{F}_j\subseteq \bigcup_{j=1}^{N}\mathcal{F}_j$ for all $j\in [N]$.
\end{proof}

The guarantees presented in \Cref{cor: grid stoch dim} and \Cref{lemma: adv mathcalN} differ by a factor of \(N\) appearing in the numerator of the second term on the right-hand side. This discrepancy stems from the fact that, in the stochastic setting, the algorithm estimates probabilities, which necessarily sum to \(1\), whereas in the adversarial setting, it estimates the sum over the block of the block-wise frequency of the indicator functions with positive values, which can sum to at most the number of blocks \(N\).

To reconcile this difference, \Cref{alg: stoch tot} sets the parameter \(\alpha = T^{-\frac{1}{3}\beta}\), while \Cref{alg: adv} uses \(\alpha = T^{\frac{1}{3}\beta}\), which can equivalently be written as \(\alpha = N \cdot T^{-\frac{1}{3}\beta}\).

\Cref{lemma: adv mathcalN} will be useful to bound the regret of \Cref{alg: adv} generated by its explorations rounds.

To conclude this section, we present one final result that will be instrumental in bounding the regret incurred during the exploitation phase.  
Recall the decomposition of the regret in the exploit phase from \Cref{sec:advChal} into two components: \(R^{\text{dyn}}\left(\{F_{j(t)}\}_{t=1}^N\right)\) and \(R^{\text{grid}}\).  
While the next section will provide a detailed discussion on how to bound \(R^{\text{dyn}}\left(\{F_{j(t)}(p^*)\}_{t=1}^N\right)\), our focus here is to highlight how \Cref{alg: adv} makes use of \Cref{alg: seq omd fix-sha}, whose guarantees depend on the number of change points in the sequence relative to which the dynamic regret is computed.

At a high level, this is why it is important that the maximum length of any path from the root to a leaf in the tree representation of the general grid \(\bigcup_{j=1}^N \mathcal{F}_j\) is bounded.  
Indeed, by the definition of the sequence \(\{F_{j(t)}\}_{t=1}^N\), the number of change points—denoted \(S(\{F_{j(t)}(p^*)\}_{t=1}^T)\)—corresponds to the length of the path in the tree from \(F_{1}(p^*)\) to \(F_N(p^*)\).

\begin{lemma}
\label{lemma: adv i bar}
    If \Cref{alg: adv} is initialized with parameters \(T\), \(\beta\), and \(\delta\), then for any \(a \in [0,1]\), define \(F_j(a)\) as the pair of prices \((p, q) \in \mathcal{F}_j\) such that \(q \le a \le p\), for each \(j \in [N]\).  
    By construction, such a pair always exists in \(\mathcal{F}_j\) for every \(j\).  
    Let \(\mathbf{F}(a) = \{F_1(a), \ldots, F_N(a)\}\), and define the number of changes in this sequence as \(
S(\mathbf{F}(a)) = \sum_{j=1}^{N-1} \mathbb{I}\{F_j(a) \ne F_{j+1}(a)\}.
\)
Then, under the event $\mathcal{E}^A_1$ the following holds:

    \[\mathbf{S}(a)\le \log_2\left(\frac{4N}{K\alpha}\right).\]
    
\end{lemma}
\begin{proof}
    Consider that $S(a)$ can be seen as the length of the path on the grid tree from a root in $\mathcal{F}_0$ to a leaf of the tree belonging to $\mathcal{F}_N$.
    \\
    Hence $\max_{a\in[0,1]}S(\mathbf{F}(a))$ is the maximum depth of the generated grid tree.
    A node of the tree $(p,q)$ at depth $i$ is a father node, i.e. it is not a leaf as it is discarded before block $N$ ($(p,q)\in \bigcup_{j=1}^{N-1}\mathcal{F}_{j}$ and $(p,q)\notin \,\mathcal{F}_N$) , only if it satisfy for some $j\in [N]$ condition: 
    \[\underline{n}_j>2^i K \alpha.\]
    Moreover, by \Cref{lemma: n hat adv}, under event $\mathcal{E}^A_1$ it holds that $\underline{n}_j\le \hat{n}_j\le 4N$, hence if we fix $i=\log_2\left(\frac{4N}{K\alpha}\right)$ and for some $j\in [N]$ it really exists a point such that $(p-q)=\frac{1}{K2^i}$, this node will never satisfy the discarding condition as \[\underline{n}_j \le 4N = 2^i K \alpha.\]
\end{proof}

\subsection{Regret Exploitation Phase}

In this section we will give the first result relative to the regret of the exploitation phase of \Cref{alg: adv}. We will use the decomposition presented in \Cref{sec:advChal} to bound separately the contribution of $R^{dyn}(\{F_t(p^*)\}_{t=1}^T)$ and $R^{grid}$.

To bound the first we will use the guarantees of \Cref{alg: seq omd fix-sha}, where the time horizon of \Cref{alg: seq omd fix-sha} is $N$ number of blocks and it receive full feedback for all couple of prices active in the block $j\in [N]$ considered (i.e.  for all $(p,q)$ s.t. $(p,q)\in \mathcal{F}_j$). 

\begin{lemma}
\label{lemma: adv dyn reg}
    Initialize \Cref{alg: adv} with values $T,\beta, \delta$, and let $p^*$ be $p^* \in \arg \max_{p\in [0,1]} \sum_{t=1}^T\gft_t(p,p)$ and let $F_{j(t)}(p^*)$ be a couple of prices $(p,q)\in \mathcal{F}_{j(t)}$ such that $q\le p^* \le p$, where $j(t)$ is the $j\in [N]$ such that $t\in \mathcal{B}_j$. 
    
    Then, under event $\mathcal{E}^A_1$ with probability at least $1-2\delta$:
    \begin{equation*}
        \sum_{t=1}^T\left(\gft_{t}(F_{j(t)}(p^*))-\gft_t(p_{j(t)},q_{j(t)})\right)=\tilde {\mathcal{O}}\left(\frac{T}{N}(1+S(\{F_{j}(p^*)\}_{j=1}^N))\sqrt{\ln\left(K+\frac{N}{\alpha K}\right)N}\right).
    \end{equation*}
\end{lemma}
\begin{proof}
    By \Cref{cor: gft adv} it holds that $\mathbb{E}[\ell_j(p,q)]=\frac{1}{6}\left(3-\frac{T}{N}\sum_{t\in \mathcal{B}_j}\gft_t(p,q)\right)$ and $\ell_j(p,q)\in [0,1]$ for all $j\in [N]$ and for all $(p,q)\in \mathcal{F}_j$.
    Under event $\mathcal{E}^A_1$ by \Cref{lemma: adv mathcalN} and by \Cref{lemma: sleeping} with probability at least $1-\delta$
    \[\sum_{j=1}^N\left(\ell_j(F_j(p^*))-\ell_j(p_j,q_j)\right)\le \tilde{\mathcal{O}}\left((1+S(\{F_{j}(p^*)\}_{j=1}^N))\sqrt{\ln\left(K+\frac{N}{\alpha K}\right)N}\right).\]
    Then we can apply Hoeffding inequality to bound to bound the difference between $\ell_j(p,q)$ and their expected values for all $(p,q)\in \mathcal{F}_j,j\in [N]$ with probability at least $1-\delta$.
\end{proof}
\begin{corollary}
    \Cref{alg: adv} with values $T,\beta, \delta$ is such that under event $\mathcal{E}^A_1$ with probability at least $1-2\delta$ it holds
    \[\sum_{t=1}^T\left(\gft_{t}(F_{j(t)}(p^*))-\gft_t(p_{j(t)},q_{j(t)})\right)=\tilde {\mathcal{O}}\left(\frac{T}{\sqrt{N}}\ln\left(\frac{N}{K\alpha}\right)\sqrt{\ln\left(K\right)}\right)\]
\end{corollary}
\begin{proof}
    The corollary derives directly from \Cref{lemma: adv dyn reg}, after using \Cref{lemma: adv i bar} to bound $S(\mathbf{F}(p^*))$ with $\log_2\left(\frac{4N}{K\alpha}\right)$.
\end{proof}


We now establish an upper bound on the regret incurred by restricting the set of allowable price choices from the continuum along the diagonal \((p, p)\), with \(p \in [0,1]\), to the discrete sequence of grids \(\{\mathcal{F}_j\}_{j=1}^N\). Specifically, we show that with high probability, the grid construction performed by \Cref{alg: adv} provides a sufficiently accurate approximation over the entire time horizon.  
The reasoning closely parallels the argument used in the stochastic setting.

\begin{lemma}
\label{lemma: adv seq reg}
    Let $p^*$ be $p^* \in \arg \max_{p\in [0,1]} \sum_{t=1}^T\gft_t(p,p)$.
    Then \Cref{alg: adv} is such that under the event $\mathcal{E}^A_1$ it guarantees
    \begin{equation*}
        \sum_{t=1}^T\left(\gft_t(p^*,p^*)-\gft_t(F_{j(t)}(p^*))\right) \le 
        \tilde{\mathcal{O}}\left(\frac{\alpha T}{N}\ln\left(\frac{N}{K\alpha}\right)+\frac{T}{\sqrt{N}}\ln\left(\frac{1}{\delta}\right)\right)
    \end{equation*}
    where $F_j(p^*)$ is a couple of prices $(p,q)\in \mathcal{F}_j$, the corresponding pair of prices on the grid in the episode considered, where $j(t)$ is the block $j$ that contains episode $t$.
\end{lemma}

\begin{proof}
    Let $(p,q)\in \{(p,q)\in [0,1]^2: \exists j\in [N] s.t. (p,q)=F_j(p^*)\}$. Then, for this $(p,q)$ 
    \[\underline{n}_j(p,q)\le 2^iK\alpha+4 \quad \forall j \in [N]~ s.t. ~ (p,q)\in \mathcal{F}_j,\]
    with $i$ such that $(p-q)=\frac{1}{K2^i}$.
    This observation is straightforward: as soon as $\underline{n}_j(p,q)$ exceeds the threshold $2^iK\alpha$, the point $(p,q)$ is removed from future grids. Moreover, by construction, $\underline{n}_j(p,q)$ can increase by at most $4$ between consecutive blocks.

We can then leverage the fact that if a trade occurs for the pair $(p^*, p^*)$, it also occurs for any pair $(p,q)$ such that $q \leq p^* \leq p$. Therefore, it suffices to bound the negative contribution to the gain from trade (GFT) that may arise when using the pair $(p,q)$, specifically when the condition $q \leq b_t \leq s_t \leq p$ holds. This scenario is analogous to the one introduced for stochastic valuations in \Cref{fig: grid construction}.

As previously mentioned, the negative contribution to the GFT is closely tied to how far the point $(p,q)$ lies below the diagonal. Indeed, 

    \[(b_t-s_t)\mathbb{I}\{q\le b_t\le s_t\le p\}\ge -(p-q)\mathbb{I}\{q\le b_t\le s_t\le p\},\]
    which also implies
    \begin{align*}
        \sum_{j=1}^N\mathbb{I}\{(p,q)\in \mathcal{F}_j\}\sum_{t\in \mathcal{B}_j}(\gft_t(p^*,p^*)-\gft_t(p,q))\le
        \sum_{j=1}^N\mathbb{I}\{(p,q)\in \mathcal{F}_j\}\sum_{t\in \mathcal{B}_j}(p-q)\mathbb{I}\{q\le b_t\le s_t\le p\}.
    \end{align*}
    Moreover, by \Cref{lemma: n hat adv} with probability under the event $\mathcal{E}^A_1$
    \[ \frac{N}{T}\sum_{j=1}^N\mathbb{I}\{(p,q)\in \mathcal{F}_j\}\sum_{t\in \mathcal{B}_j}\mathbb{I}\{q\le b_t\le s_t\le p\}\le \left(2^iK\alpha+4\right)+ 8 \sqrt{\frac{\ln\left(\frac{2T}{\delta}\right)}{2}N},\]
    and therefore using $(p-q)\le 1$ and $p-q=\frac{1}{2^iK}$:
    \[\sum_{j=1}^N\mathbb{I}\{(p,q)\in \mathcal{F}_j\}\sum_{t\in \mathcal{B}_j}(\gft_t(p^*,p^*)-\gft_t(p,q))\le \frac{T}{N}\left(\alpha + 4\right)+\frac{T}{N}8 \sqrt{\frac{\ln\left(\frac{2T}{\delta}\right)}{2}N}.\]
    Notice that this reasoning holds identically for all $(p,q)\in \{(p,q)\in [0,1]^2: \exists j\in [N] s.t. (p,q)=F_j(p^*)\}$, and that by \Cref{lemma: adv i bar} under the event $\mathcal{E}^A_1$ it holds that $|\{(p,q)\in [0,1]^2: \exists j\in [N] s.t. (p,q)=F_j(p^*)\}|\le \log_2\left(\frac{4N}{K\alpha}\right)$, thus we obtain the statement of the lemma:
    \[\sum_{t=1}^T\left(\gft_t(p^*,p^*)-\gft_t(p_{t}^*,q_{t}^*)\right) \le \log_2\left(\frac{4N}{K\alpha}\right)\left(\frac{T(\alpha+4)}{N}+ \frac{T}{\sqrt{N}}4\sqrt{2\ln\left(\frac{2T}{\delta}\right)}\right).\]
\end{proof}

\subsection{Putting Everything Together}

In this section we will join all pieces of the analysis to provide the general high probability guarantees relative to the challenging scenario of adversarial valuations.

\begin{theorem}
\label{theo: adv tot}
   Fix $\beta\in [3/4,6/7]$ and $\delta\in (0,1)$. There exists an algorithm that guarantees 
   \[V_T\le \BigOL{T^{\beta}}\]
   and with probability at least $1-3\delta$ guarantees
    \[R_T \le \BigOL{T^{1-\frac{1}{3}\beta}}.\]
\end{theorem}
\begin{proof}

First we bound the regret and then the violations.

\textbf{Regret }
    Consider \Cref{alg: adv} initialized with values $T,\beta,\delta$. The regret incurred by the algorithm can be decomposed in four parts: $R^{dyn}(\mathbf{F}(p^*))$, $R^{grid}$, and the regret cumulated during the time steps used in the \gft and frequency estimation $\mathcal{S}_j^f,\mathcal{S}_j^g$.
    The last two component can be upperbounded by $T^{ind}= \sum_{j=1}^N|\mathcal{S}_j^f|$ and $T^{gft}= \sum_{j=1}^N|\mathcal{S}_j^g|$.

    First we focus on $T^{ind}, T^{gft}$.

    Since $f_j,g_j$ are bijective functions it holds that $T^{ind}=T^{gft}=\sum_{j=1}^N|\mathcal{F}_j|\le KN+\frac{4N^2}{K\alpha}$ by \Cref{cor: adv T} under the event $\mathcal{E}^A_1$.
    
    Then, by \Cref{lemma: adv dyn reg} and \Cref{lemma: adv seq reg} it holds with probability at least $1-3\delta$
    \begin{align*}
        \sum_{t=1}^T \gft_t(p^*,p^*)-\gft_t(p_t,q_t) &\le R^{grid}+R^{dyn}(\mathbf{F}(p^*))+T^{ind}+T^{gft}\\
        & \le \tilde{\mathcal{O}}\left(\frac{\alpha T}{N}+\frac{T}{\sqrt{N}}+KN+\frac{N^2}{K\alpha}\right).
    \end{align*}
    \Cref{alg: adv} set $K=T^{1-\beta},N=T^{\frac{2}{3}\beta},\alpha=T^{\frac{1}{3}\beta}$, hence the inequality above becomes
    \begin{equation}
        R_T\le \tilde{\mathcal{O}}\left(T^{1-\frac{1}{3}\beta}+T^{2\beta-1}\right)\le \tilde{\mathcal{O}}\left(T^{1-\frac{1}{3}\beta}\right),
    \end{equation}
    since $2\beta-1\le 1-\frac{1}{3}\beta$ if $\beta\le \frac{6}{7}$. Notice that, similar to the stochastic case, the regret for adversarial environment presents a discontinuity for $\beta=\frac{6}{7}$, which is consistent with the lower bound.
    
\textbf{Violation } Observe that both the functions \texttt{Ind.Est}$(p,q)$ and \texttt{\gft.EST}$(p,q)$ generate at most instant violation $p-q$. 
Hence, at most $V_T\le T \cdot \frac{1}{K}\le \BigOL{T^{\beta}}$.
\end{proof}

\section{Lower Bound Tradeoff} \label{sec:lowerbound}

In this section, we provide a parametric lowerbound that matches the positive result of the previous section. In particular, we show that for each $T^{\beta}$, $\beta\in [3/4,6/7]$, the best possible regret obtainable violating the budget constraint by at most $T^{\beta}$ is $\Omega(T^{1-\beta/3})$. This result interpolates between the $\Omega(T^{5/7})$ result with arbitrary budget violation of \cite{bernasconi2024no} and the $\Omega(T^{3/4})$ with no valuation of the concurrent work of \cite{chen2025tight}.

\begin{theorem}
\label{thm: lower bound}
    Assume two-bit feedback. Any algorithm that guarantees $\mathcal{V}_T\le T^\beta$ with $\beta\in [\frac{3}{4},\frac{6}{7}]$ suffer at least pseudo-regret $\mathcal{R}_T\ge \Omega(T^{1-\frac{1}{3}\beta})$. 
\end{theorem}

In the rest of the section we will discuss the proof of \Cref{thm: lower bound}.
We will fix a $\beta\in [\frac{3}{4},\frac{6}{7}]$ and for such $\beta$ we will prove \Cref{thm: lower bound}.
We will do so in $5$ steps:
\begin{enumerate}
    \item Construct \( N \) well-defined hard instances, inspired by the apple-tasting framework. Each instance separates the space into disjoint regions for exploration and exploitation, making it particularly costly for the algorithm to distinguish between perturbed instances.
    \item Partition the decision space \( [0,1]^2 \) according to the gain-from-trade (GFT), the feedback structure, and the revenue, to better understand the trade-offs and learning challenges.
    \item Characterize the regret of an algorithm in all $N$ instances
    \item Relate the behavior of the same algorithm across the instances
    \item Conclude the proof by showing the regret trade-off across instances
\end{enumerate}

\subsection{Building a Set of Hard Instances}
\label{sec: lb subsec 1}

\begin{figure}[H]
\begin{minipage}{0.48\textwidth}
    \centering
\begin{tikzpicture}[scale=4]

  \coordinate (O) at (0,0);
  \coordinate (A) at (1,0);
  \coordinate (B) at (1,1);
  \coordinate (C) at (0,1);
  \coordinate (P11) at (0,0.3);
  \coordinate (P12) at (0,0.35);
  \coordinate (P13) at (0,0.4);
  \coordinate (P14) at (0,0.45);
  \coordinate (P15) at (0,0.5);
  \coordinate (P16) at (0,0.55);
  \coordinate (P17) at (0,0.6);
  \coordinate (P18) at (0,0.65);
  \coordinate (P19) at (0,0.7);
  
  \coordinate (P21) at (0.1,0.3);
  \coordinate (P22) at (0.1,0.35);
  \coordinate (P23) at (0.1,0.4);
  \coordinate (P24) at (0.1,0.45);
  \coordinate (P25) at (0.1,0.5);
  \coordinate (P26) at (0.1,0.55);
  \coordinate (P27) at (0.1,0.6);
  \coordinate (P28) at (0.1,0.65);
  \coordinate (P29) at (0.1,0.7);

  \coordinate (P31) at (0.3,1);
  \coordinate (P32) at (0.35,1);
  \coordinate (P33) at (0.4,1);
  \coordinate (P34) at (0.45,1);
  \coordinate (P35) at (0.5,1);
  \coordinate (P36) at (0.55,1);
  \coordinate (P37) at (0.6,1);
  \coordinate (P38) at (0.65,1);
  \coordinate (P39) at (0.7,1);

  \coordinate (P41) at (0.3,0.9);
  \coordinate (P42) at (0.35,0.9);
  \coordinate (P43) at (0.4,0.9);
  \coordinate (P44) at (0.45,0.9);
  \coordinate (P45) at (0.5,0.9);
  \coordinate (P46) at (0.55,0.9);
  \coordinate (P47) at (0.6,0.9);
  \coordinate (P48) at (0.65,0.9);
  \coordinate (P49) at (0.7,0.9);

  \coordinate (P5) at (0.3,0.7);

  \draw[thick] (O) rectangle (B);

  \draw[thick] (O) -- (B);


  \node [circle, fill=red, scale=0.5] at (P11) {}  ;
  \node [circle, fill=red, scale=0.5] at (P13) {}  ;
  \node [circle, fill=red, scale=0.5] at (P15) {}  ;
  \node [circle, fill=red, scale=0.5] at (P17) {}  ;
  \node [circle, fill=red, scale=0.5] at (P19) {}  ;

  \node [circle, fill=teal, scale=0.5] at (P21) {}  ;
  \node [circle, fill=teal, scale=0.5] at (P23) {}  ;
  \node [circle, fill=teal, scale=0.5] at (P25) {}  ;
  \node [circle, fill=teal, scale=0.5] at (P27) {}  ;
  \node [circle, fill=teal, scale=0.5] at (P29) {}  ;


 \node [circle, fill=orange, scale=0.5] at (P31) {}  ;
  \node [circle, fill=orange, scale=0.5] at (P33) {}  ;
  \node [circle, fill=orange, scale=0.5] at (P35) {}  ;
  \node [circle, fill=orange, scale=0.5] at (P37) {}  ;
  \node [circle, fill=orange, scale=0.5] at (P39) {}  ;


 \node [circle, fill=olive, scale=0.5] at (P41) {}  ;
  \node [circle, fill=olive, scale=0.5] at (P43) {}  ;
  \node [circle, fill=olive, scale=0.5] at (P45) {}  ;
  \node [circle, fill=olive, scale=0.5] at (P47) {}  ;
  \node [circle, fill=olive, scale=0.5] at (P49) {}  ;


\node [circle, fill=magenta, scale=0.5] at (P5) {}  ;
  

\node [circle, fill=cyan, scale=0.5] at (O) {}  ;
\node [circle, fill=cyan, scale=0.5] at (A) {}  ;
\node [circle, fill=cyan, scale=0.5] at (B) {}  ;
\node [circle, fill=cyan, scale=0.5] at (C) {}  ;

\matrix [draw, above left] at (-0.1,0.3) {
\node[circle, fill=orange, scale=0.5] (l1)  {};
  \node[right,font=\tiny ] at (l1.east) {$\mathcal{W}_1$};\\
    \node[circle, fill=olive, scale=0.5] (l2)  {};
  \node[right,font=\tiny] at (l2.east) {$\mathcal{W}_2$};\\
  \node[circle, fill=red, scale=0.5] (l3)  {};
  \node[right, font=\tiny] at (l3.east) {$\mathcal{W}_3$};\\
  \node[circle, fill=teal, scale=0.5] (l4)  {};
  \node[right,font=\tiny] at (l1.east) {$\mathcal{W}_4$};\\
  \node[circle, fill=magenta, scale=0.5] (l3)  {};
  \node[right, font=\tiny] at (l3.east) {$\mathcal{W}_5$};\\
  \node[circle, fill=cyan, scale=0.5] (l4)  {};
  \node[right,font=\tiny] at (l1.east) {$\mathcal{W}_6$};\\
};

\end{tikzpicture}
\end{minipage}
\begin{minipage}{0.48\textwidth}
\centering
\vspace{0.4cm}
\begin{tikzpicture}[remember picture, scale=4]

  \coordinate (O) at (0,0);
  \coordinate (A) at (1,0);
  \coordinate (B) at (1,1);
  \coordinate (C) at (0,1);
  \coordinate (P11) at (0,0.3);
  \coordinate (P12) at (0,0.35);
  \coordinate (P13) at (0,0.4);
  \coordinate (P14) at (0,0.45);
  \coordinate (P15) at (0,0.5);
  \coordinate (P16) at (0,0.55);
  \coordinate (P17) at (0,0.6);
  \coordinate (P18) at (0,0.65);
  \coordinate (P19) at (0,0.7);
  
  \coordinate (P21) at (0.1,0.3);
  \coordinate (P22) at (0.1,0.35);
  \coordinate (P23) at (0.1,0.4);
  \coordinate (P24) at (0.1,0.45);
  \coordinate (P25) at (0.1,0.5);
  \coordinate (P26) at (0.1,0.55);
  \coordinate (P27) at (0.1,0.6);
  \coordinate (P28) at (0.1,0.65);
  \coordinate (P29) at (0.1,0.7);

  \coordinate (P31) at (0.3,1);
  \coordinate (P32) at (0.35,1);
  \coordinate (P33) at (0.4,1);
  \coordinate (P34) at (0.45,1);
  \coordinate (P35) at (0.5,1);
  \coordinate (P36) at (0.55,1);
  \coordinate (P37) at (0.6,1);
  \coordinate (P38) at (0.65,1);
  \coordinate (P39) at (0.7,1);

  \coordinate (P41) at (0.3,0.9);
  \coordinate (P42) at (0.35,0.9);
  \coordinate (P43) at (0.4,0.9);
  \coordinate (P44) at (0.45,0.9);
  \coordinate (P45) at (0.5,0.9);
  \coordinate (P46) at (0.55,0.9);
  \coordinate (P47) at (0.6,0.9);
  \coordinate (P48) at (0.65,0.9);
  \coordinate (P49) at (0.7,0.9);

  \coordinate (P5) at (0.3,0.7);

  \draw[thick] (O) rectangle (B);

  \draw[thick] (O) -- (B);


  \node [circle, fill=red, scale=0.5] at (P11) {}  ;
  \node [circle, fill=red, scale=0.5] at (P13) {}  ;
  \node [circle, fill=red, scale=0.5] at (P15) {}  ;
  \node [circle, fill=red, scale=0.5] at (P17) {}  ;
  \node [circle, fill=red, scale=0.5] at (P19) {}  ;

  \node [circle, fill=teal, scale=0.5] at (P21) {}  ;
  \node [circle, fill=teal, scale=0.5] at (P23) {}  ;
  \node [circle, fill=teal, scale=0.5] at (P25) {}  ;
  \node [circle, fill=teal, scale=0.5] at (P27) {}  ;
  \node [circle, fill=teal, scale=0.5] at (P29) {}  ;


 \node [circle, fill=orange, scale=0.5] at (P31) {}  ;
  \node [circle, fill=orange, scale=0.5] at (P33) {}  ;
  \node [circle, fill=orange, scale=0.5] at (P35) {}  ;
  \node [circle, fill=orange, scale=0.5] at (P37) {}  ;
  \node [circle, fill=orange, scale=0.5] at (P39) {}  ;


 \node [circle, fill=olive, scale=0.5] at (P41) {}  ;
  \node [circle, fill=olive, scale=0.5] at (P43) {}  ;
  \node [circle, fill=olive, scale=0.5] at (P45) {}  ;
  \node [circle, fill=olive, scale=0.5] at (P47) {}  ;
  \node [circle, fill=olive, scale=0.5] at (P49) {}  ;


\node [circle, fill=magenta, scale=0.5] at (P5) {}  ;
  

\node [circle, fill=cyan, scale=0.5] at (O) {}  ;
\node [circle, fill=cyan, scale=0.5] at (A) {}  ;
\node [circle, fill=cyan, scale=0.5] at (B) {}  ;
\node [circle, fill=cyan, scale=0.5] at (C) {}  ;


\end{tikzpicture}

\begin{tikzpicture}[xshift=-2cm, yshift=0.55cm, remember picture, overlay, scale=4]

\coordinate (G1) at (0.3,0.3);
\coordinate (G2) at (0.3,0.7);
\coordinate (G3) at (0.7,0.3);
\coordinate (G4) at (0.7,0.7);

\coordinate (M11) at (0.3,0.3);
\coordinate (M12) at (0.4,0.3);
\coordinate (M13) at (0.5,0.3);
\coordinate (M14) at (0.6,0.3);
\coordinate (M15) at (0.7,0.3);

\coordinate (M21) at (0.3,0.4);
\coordinate (M22) at (0.4,0.4);
\coordinate (M23) at (0.5,0.4);
\coordinate (M24) at (0.6,0.4);
\coordinate (M25) at (0.7,0.4);

\coordinate (M31) at (0.3,0.5);
\coordinate (M32) at (0.4,0.5);
\coordinate (M33) at (0.5,0.5);
\coordinate (M34) at (0.6,0.5);
\coordinate (M35) at (0.7,0.5);

\coordinate (M41) at (0.3,0.6);
\coordinate (M42) at (0.4,0.6);
\coordinate (M43) at (0.5,0.6);
\coordinate (M44) at (0.6,0.6);
\coordinate (M45) at (0.7,0.6);

\coordinate (M51) at (0.3,0.7);
\coordinate (M52) at (0.4,0.7);
\coordinate (M53) at (0.5,0.7);
\coordinate (M54) at (0.6,0.7);
\coordinate (M55) at (0.7,0.7);

  \coordinate (P21) at (0.1,0.3);
  \coordinate (P22) at (0.1,0.4);
  \coordinate (P23) at (0.1,0.5);
  \coordinate (P24) at (0.1,0.6);
  \coordinate (P25) at (0.1,0.7);

  \coordinate (P11) at (0,0.3);
  \coordinate (P12) at (0,0.4);
  \coordinate (P13) at (0,0.5);
  \coordinate (P14) at (0,0.6);
  \coordinate (P15) at (0,0.7);

  \coordinate (P31) at (0.3,1);
  \coordinate (P32) at (0.4,1);
  \coordinate (P33) at (0.5,1);
  \coordinate (P34) at (0.6,1);
  \coordinate (P35) at (0.7,1);

  \coordinate (P41) at (0.3,0.9);
  \coordinate (P42) at (0.4,0.9);
  \coordinate (P42) at (0.4,0.9);
  \coordinate (P43) at (0.5,0.9);
  \coordinate (P44) at (0.6,0.9);
  \coordinate (P45) at (0.7,0.9);

\draw[thick] (G1) rectangle (G4);

\node [diamond, scale=0.4, fill=black] at (M11) {}  ;
\node [diamond, scale=0.4, fill=black] at (M12) {}  ;
\node [diamond, scale=0.4, fill=black] at (M13) {}  ;
\node [diamond, scale=0.4, fill=black] at (M14) {}  ;
 \node [diamond, scale=0.4, fill=black] at (M15) {}  ;

 \node [diamond, scale=0.4, fill=black] at (M21) {}  ;
\node [diamond, scale=0.4, fill=black] at (M22) {}  ;
\node [diamond, scale=0.4, fill=black] at (M23) {}  ;
\node [diamond, scale=0.4, fill=black] at (M24) {}  ;
 \node [diamond, scale=0.4, fill=black] at (M25) {}  ;

 \node [diamond, scale=0.4, fill=black] at (M31) {}  ;
\node [diamond, scale=0.4, fill=black] at (M32) {}  ;
\node [diamond, scale=0.4, fill=black] at (M33) {}  ;
\node [diamond, scale=0.4, fill=black] at (M34) {}  ;
 \node [diamond, scale=0.4, fill=black] at (M35) {}  ;

\node [diamond, scale=0.4, fill=black] at (M41) {}  ;
\node [diamond, scale=0.4, fill=black] at (M42) {}  ;
\node [diamond, scale=0.4, fill=black] at (M43) {}  ;
\node [diamond, scale=0.4, fill=black] at (M44) {}  ;
 \node [diamond, scale=0.4, fill=black] at (M45) {}  ;

 \node [diamond, scale=0.4, fill=black] at (M51) {}  ;
\node [diamond, scale=0.4, fill=black] at (M52) {}  ;
\node [diamond, scale=0.4, fill=black] at (M53) {}  ;
\node [diamond, scale=0.4, fill=black] at (M54) {}  ;
 \node [diamond, scale=0.4, fill=black] at (M55) {}  ;
 
\draw[] (M21) rectangle (M25);
\draw[] (M31) rectangle (M35);
\draw[] (M41) rectangle (M45);

\draw[] (M12) rectangle (M52);
\draw[] (M13) rectangle (M53);
\draw[] (M14) rectangle (M54);

\draw[thick,fill=pink!30] (P11) rectangle (P25);
\draw[fill=pink!30] (P11) rectangle (P22);
\draw[fill=pink!30] (P12) rectangle (P23);
\draw[fill=pink!30] (P13) rectangle (P24);
\draw[fill=pink!30] (P14) rectangle (P25);

\draw[thick,fill=pink!30] (P31) rectangle (P45);
\draw[fill=pink!30] (P31) rectangle (P42);
\draw[fill=pink!30] (P32) rectangle (P43);
\draw[fill=pink!30] (P33) rectangle (P44);
\draw[fill=pink!30] (P34) rectangle (P45);

\matrix [draw, above left] at (-0.1,0.7) {
\node[diamond, fill=black, scale=0.5] (l1)  {};
  \node[right,font=\tiny ] at (l1.east) {$\{M_{i,j}\}$};\\
    \node[draw=black, fill=pink, scale=0.8] (l2)  {};
  \node[right,font=\tiny] at (l2.east) {$\{N_k\}$};\\
};

\end{tikzpicture}
\end{minipage}
\caption{ \emph{Left}: the support $\mathcal{W}$ of $\{\mu_k\}_{k\in \{0,\ldots,N-1\}}$. \emph{Right}: the generated exploitation grid $\{M_{i,j}\}$ and the exploration areas $\{N_k\}$.}
\label{fig: lb instances}
\end{figure}
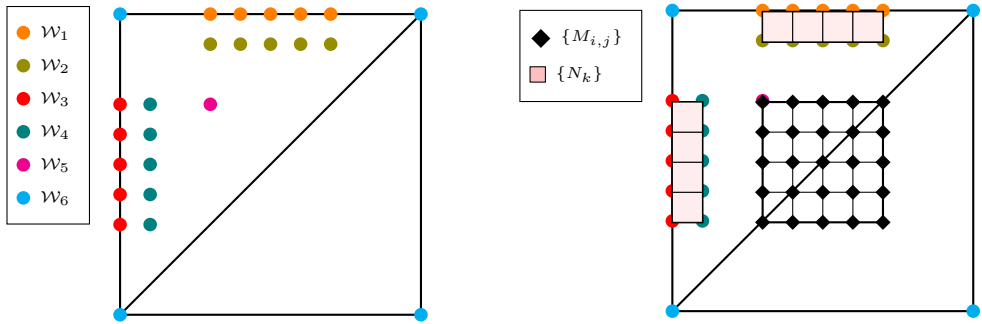

We begin by introducing a collection of \( N \) hard instances of the bilateral trade problem, where \( N \) is a function of the time horizon \( T \), to be set later on. Our objective is to demonstrate that any learning algorithm guaranteeing a violation budget \( V_T \le T^\beta \) must incur a regret of at least \( \Omega(T^{1 - \frac{1}{3}\beta}) \) on at least one of these \( N \) instances.

In the construction, we use the following parameters: $\ell=\frac{1}{8}$, $\Delta=\frac{\ell}{N}$,$g=\frac{1}{24}T^{1-\frac{4}{3}\beta}$, $N=\frac{1}{200}T^{1-\beta}$. 
%
Each instance is defined by a probability distribution \( \mu_k \in \Delta([0,1]^2) \), for \( k = 0, \ldots, N-1 \). We begin by constructing the base instance, denoted by \( \mu_0 \), which we refer to as the \emph{unperturbed instance}. The remaining \( N - 1 \) distributions \( \mu_k \) for \( k = 1, \ldots, N - 1 \) are constructed as perturbations of \( \mu_0 \).

\paragraph{Instances Support}
All distributions \( \{\mu_k\}_{k=0}^{N-1} \) will be supported on a finite set of valuation pairs, which we denote by \( \mathcal{W} \).
\( \mathcal{W} \) can be seen as the union of six different sets $\mathcal{W}_1,\mathcal{W}_2,\mathcal{W}_3,\mathcal{W}_4,\mathcal{W}_5,\mathcal{W}_6$. 
An illustration of these sets is presented in \Cref{fig: lb instances}.


First, we define the four sets of valuations $\cW_1$, $\cW_2$ , $\cW_3$ and $\cW_4$ as follows:
\begin{align*}
    &\cW_1\defeq\left\{w^i_1\defeq\left(\tfrac{1-\ell}2+i \Delta, 1\right):\, i=0,\ldots, N\right\}\\ 
    &  \cW_2\defeq\left\{w^i_2\defeq\left(\tfrac{1-l}2+i \Delta, 1-3\ell\right):\, i=0,\ldots, N\right\}\\ 
    &
    \cW_3\defeq\left\{w^i_3\defeq\left(0,\tfrac{1-\ell}2+i \Delta\right):\, i=0,\ldots, N\right\}\\
    &
    \cW_4\defeq\left\{w^i_4\defeq\left(3\ell,\tfrac{1-\ell}2+i \Delta \right):\, i=0,\ldots, N\right\}.
\end{align*}
The set $\cW_5$ includes a single valuation
\[\cW_5 \coloneqq \bigg\{\left(\frac{1-\ell}{2},\frac{1+\ell}{2}\right)\bigg\},\]

and $\cW_6$ is defined as

\[\cW_6 \coloneqq \{(0,0),(0,1),(1,0),(1,1)\}.\]

\paragraph{Unperturbed Instance}
To define the probability $\mu_0\in \Delta(\cW)$, we will the following quantities $\gamma_1= g/(4(N+1))$, $\gamma_5=\frac{1}{2}$ and $\gamma_6=\frac{1}{4}\left(1-\sum_{w\in \mathcal{W}_1\cup\mathcal{W}_2\cup\mathcal{W}_3\cup\mathcal{W}_4}\mu_0(w)-\gamma_5\right)$. 
Then,  we define the probability $\mu_0(w)$ assigned to the couple of prices $w$ as:
\begin{align*}
&\mu_0(w_1^i)= \gamma_1\left(1+\frac{2i}{3(N)}\right) &&  \forall i\in \{0,\ldots,N\}\\
&\mu_0(w_2^i)= \gamma_1\left(1-\frac{2i}{3(N)}\right) &&  \forall i\in \{0,\ldots,N\}\\
&\mu_0(w_3^i)= \gamma_1\left(1+\frac{2(N-i)}{3(N)}\right) &&  \forall i\in \{0,\ldots,N\}\\
&\mu_0(w_4^i)= \gamma_1\left(1-\frac{2(N-i)}{3(N)}\right) &&  \forall i\in \{0,\ldots,N\}\\
&\mu_0(w_5)=\gamma_5\\
&\mu_0(w_6^i)=\gamma_6 && \forall w_6^i\in \mathcal{W}_6.
\end{align*}
In the following, we will prove that this is actually a probability distribution.

\paragraph{Perturbed Instances}
Building on the $0$-th instance we define the $\mu_k\in \Delta(\cW)$  characterizing the $k$-th instance with $k=1,\ldots,N-1$ as an $\epsilon\in (0,\frac{1}{3}\gamma_1]$ perturbation of $\mu_0$. In particular
\begin{equation*}
\label{eq:LBmu1}
    \mu_k(w^i_j)= \mu_0(w^i_j),\quad \forall j\in\{1,2,3,4\}, i\notin \{k,k+1\},
    \end{equation*}
    while we perturb by $\epsilon$ the probability of the following valuations:

\begin{align*}\label{eq:LBmu2}
    &\mu_k(w^k_1)=\mu_0(w^k_1)+\epsilon,\quad \mu_k(w^{k+1}_1)=\mu_0(w^{k+1}_1)-\epsilon\\
    &\mu_k(w^k_2)=\mu_0(w^k_2)-\epsilon,\quad
    \mu_k(w^{k+1}_2)=\mu_0(w^{k+1}_2)+\epsilon\\
    & \mu_k(w^k_3)=\mu_0(w^k_3)+\epsilon,\quad \mu_k(w^{k+1}_3)=\mu_0(w^{k+1}_3)-\epsilon\\
    &\mu_k(w^k_4)=\mu_0(w^k_4)-\epsilon,\quad
    \mu_k(w^{k+1}_4)=\mu_0(w^{k+1}_4)+\epsilon
    \end{align*}

All the remaining probabilities are the same as the base instance probability $\mu_0(\cdot)$.
We conclude the section showing that these probability distributions are well-defined.

\begin{restatable}{lemma}{probwelldef}
    $\{\mu_k\}_{k=0}^{N-1}$ are a set of well-defined probabilities over $\cW$.
\end{restatable}

\subsection{Relation with the Constructions of \cite{bernasconi2024no} and \cite{chen2025tight}}

Our result builds on the construction of \cite{bernasconi2024no}. However, in their construction, global budget balanced mechanisms are disincentivized making their GFT small. This would require to ``waste'' lot of probability introducing high probability trades with negative GFT, making the instance easier. 

This problem has been circumvented by the concurrent work of \cite{chen2025tight} for exactly global budget balance mechanisms, which make ineffective playing non-balanced prices which are compensate by balance one, through the introduction of additional trades that makes global budget balance mechanisms worse than strongly budget balanced ones.

In our construction, we follow a different approach and we replace the multi-apple tasting gadget in \cite{bernasconi2024no} with two copies of the problem. As we will see in the remaining of the proof, this makes global budget balance mechanisms worse than strongly budget balanced ones whenever the learner is unable to identify the underline instance.


\subsection{Partition the Decision Set}
In this section, we analyze how different regions of the decision space \([0,1]^2\) behave in terms of expected gain-from-trade (GFT), expected revenue, and feedback across the various instances introduced in \Cref{sec: lb subsec 1}. 

We begin by observing that for any instance \( k \in \{0, \ldots, N-1\} \), any algorithm operating over the full decision space \([0,1]^2\) is dominated by another algorithm restricted to the finite space
\[
\mathcal{G}^\cW \coloneqq \left\{s \in [0,1] : (s,x)\in \cW \text{ for some } x \in [0,1] \right\} \times \left\{b \in [0,1] : (x,b)\in \cW \text{ for some } x \in [0,1] \right\},
\]
where we recall that \(\cW\) is the support of the instance distributions. 

This simplification holds because, by construction, for any \((p,q) \in [0,1]^2\), there exists a point \((p',q') \in \mathcal{G}^\cW\) that produces the same GFT, yields the same feedback, and achieves at least as much revenue. Therefore, restricting the algorithm's decisions to \(\mathcal{G}^\cW\) is without loss of generality.
In the remaining of the section we define a  three-way partition of \(\mathcal{G}^\cW\) into \emph{exploration points}, \emph{exploitation points}, and \emph{dominated points}. Then, we show how w.l.o.g. we can restrict the analysis to algorithms that takes prices only from the first two sets.

\subsubsection{GFT Analysis}

In this section we focus on the analysis of the expected gain for trade GFT for the set of prices $(p,q)\in \mathcal{G}^\cW$ for any possible instance $k\in \{0,\ldots,N-1\}$.

First, we define the family of points $\{M_{i,j}\}_{i,j\in \{0,\ldots,N\}}\subseteq \mathcal{G}^\cW$ as $M_{i,j}\coloneqq (\frac{1-\ell}{2}+i\Delta, \frac{1-\ell}{2}+j \Delta)$ (the black diamonds in Figure \ref{fig: lb instances} right).

For this family of points $\{M_{i,j}\}_{i,j\in \{0,\ldots,N\}}$ we study their expected gain for trade in each possible instances $\{0,\ldots,N-1\}$. 
Onward, we use the notation \(\mathbb{E}_k\) to denote the expected value with respect to the distribution \((s_t, b_t) \sim \mu_k\), for \(k \in \{0, \ldots, N-1\}\). With a slight abuse of notation, we will write \(\mathbb{E}_k[\gft(p, q)]\) to indicate the expected gain from trade under distribution \(\mu_k\), rather than simply writing \(\gft(p, q)\) as in previous sections. We use a similar notation $\mathbb{E}_k[\rev(p,q)]$ for revenue. This notation is useful to clarify the probability distribution with respect to which the expectation is taken. 

\begin{restatable}{lemma}{lemmaGFTZero}
    \label{lemma: lb gft}
    It holds:
    \[\mathbb{E}_0[\gft(M_{i,j})]= c + \gamma_1 (1-2\ell)(i-j) \quad \forall i,j\in\{0,\ldots,N\},\]
    where $c\coloneqq \gamma_6 + \ell \gamma_5 + \gamma_1(N+2) (1-2\ell)$.
    Moreover, for all $k\in \{1,\ldots,N-1\}$
    \[\mathbb{E}_k[\gft(M_{i,j})]=\mathbb{E}_0[\gft(M_{i,j})]+3 \ell\epsilon\left(\mathbb{I}(i= k)+\mathbb{I}(j= k)\right).\]
\end{restatable}

In practice, \Cref{lemma: lb gft} implies that for each point \( M_{i,j} \), its gain from trade (GFT) in the unperturbed instance defined by \( \mu_0 \) differs from the reference value \( c \) by a quantity proportional to its distance from the diagonal. Specifically, the deviation is given by \( \gamma_1(1 - 2\ell)(i - j) \), where playing \emph{below} the diagonal (\( i > j \)) increases the expected GFT, while playing \emph{above} the diagonal (\( i < j \)) decreases it.

Moreover, in the unperturbed (0-th) instance, selecting a point \( M_{i,j} \) that lies \( \propto |i - j| \) above the diagonal and another point \( M_{i',j'} \) that lies the same distance \( |i' - j'| = |i - j| \) below the diagonal yields, in expectation, the same total GFT as choosing twice a point \( M_{i'', i''} \) on the diagonal. In other words, for any strategy that selects only among the family of points \( \{ M_{i,j} \}_{i,j=0}^{N-1} \), its expected gain from trade depends solely on the average distance from the diagonal. We emphasize here—an we will expand on it later—that the average distance from the diagonal is inherently connected to the algorithm's revenue, and therefore to its violations. By construction, the algorithm must ensure that this average distance does not fall too far below the diagonal, where ``too far'' is quantified in terms of the parameter \( \beta \).

However, this is no longer true in each of the \( N - 1 \) perturbed instances. These instances are deliberately constructed to favor the point \( M_{k,k} \) in the \( k \)-th instance, for \( k \in \{1, \ldots, N-1\} \), which yields an expected gain from trade that exceeds that of any other diagonal point \( M_{i,i} \) by a factor of \( \mathcal{O}(\epsilon) \). As a result, failing to correctly identify this optimal point becomes significantly costly. We will later discuss how identifying this point also incurs an exploration cost.

Moreover, in the perturbed instances, compensating choices below the diagonal with points above the diagonal becomes ineffective due to the double apple-tasting structure we introduce. 

The expected GFT in the \(k\)-th instance, for \(k \in \{1, \ldots, N-1\}\), can be decomposed into two components: one matching the expected GFT in the unperturbed instance, and a second component introduced by the perturbation. Specifically, this perturbation adds \(6\ell\epsilon\) to the point \(M_{k,k}\) ; it adds \(3\ell\epsilon\) to all points \(M_{i,k}\) such that \(i\neq k\) and to all points \(M_{k,j}\) such that \(k \neq j\); and it contributes zero elsewhere.

In effect, this means that any point \(M_{i,j}\) located outside the diagonal forfeits at least one of the potential \(3\ell\epsilon\) GFT increments, making such points suboptimal.

We emphasize that this systematic suboptimality of points outside the diagonal is precisely why our lower bound construction employs a \emph{double apple-tasting} structure. A single perturbation would not suffice to ensure this behavior.



\begin{lemma}
    Any algorithm that under the perturbed probability $\mu_k$, $k\in \{1,\ldots,N-1\}$, plays $\mathcal{M}_{i,j}\in [T]$ times the point $M_{i,j}$ with $i<j$ and $\mathcal{M}_{i',j'}\in [T]$ times the point $M_{i',j'}$ with $i'>j'$ such that $\mathcal{M}_{i,j}\mathbb{E}_k[\rev(M_{i,j})]+\mathcal{M}_{i',j'}\mathbb{E}_k[\rev(M_{i',j'})]\ge0$ performs strictly worse than playing $(\mathcal{M}_{i,j}+\mathcal{M}_{i',j'})$ times the point $M_{k,k}$:
    \[\mathcal{M}_{i,j}\mathbb{E}_k[\gft(M_{i,j})]+\mathcal{M}_{i',j'}\mathbb{E}_k[\gft(M_{i',j'})]\le (\mathcal{M}_{i,j}+\mathcal{M}_{i',j'})\mathbb{E}_k[\gft(M_{k,k})]-3\ell\epsilon\mathcal{M}_{i',j'}\]
\end{lemma}

To conclude the analysis of the GFT, we show that choosing prices in $\{[0,1]\times[1-3\ell,1)\} \bigcup \{(0,3\ell]\times[0,1]\}$ results in a loss of expected GFT across all \(N\) instances. This observation will be crucial in the following sections, as that region of the square \([0,1]^2\) is precisely the only area that provides feedback useful for distinguishing between different instances, which makes exploration costly in this family of instances.

\begin{lemma}
    For any instances $k\in \{0,\ldots,N-1\}$, and couple of prices $(p,q)\in \{[0,1]\times[1-3\ell,1)\} \bigcup \{(0,3\ell]\times[0,1]\} $ and for every $i \in \{0,\ldots,N\}$ it holds
    \[\mathbb{E}_k[\gft(p,p)]\ge \mathbb{E}_k[\gft(p,q)]+\gamma_5 \ell+ (1-2\ell)\gamma_1\]
\end{lemma}
\begin{proof}
    The lemma is easy to prove after noticing that each couple of prices $(p,q)$ in the areas considered lose all trade that happens when $(s_t,b_t)\in \mathcal{W}_5$ and at least two trade points in $\mathcal{W}_1\cup\mathcal{W}_2\cup\mathcal{W}_3\cup\mathcal{W}_4$, without making any more trade at positive GFT than the prices $(p,p)$.
\end{proof}

\subsubsection{Feedback Analysis}

In this section, we analyze the feedback structure over the space of possible prices \(\mathcal{G}^\cW\), with the goal of identifying which regions of the space can be used to distinguish between instances, and which regions yield identical feedback across all instances.
We will show that with two bit feedback, i.e. posting prices $(p,q)\in [0,1]^2$ reveals feedback $z=(\mathbb{I}\{s_t\le p\},\mathbb{I}\{q\le b_t\})$, only a small portion of the space can actually be used to distinguish between instances.

Indeed, for each instance $k\in \{1,\ldots,N-1\}$ let $N_K$ be a region of the space $[0,1]^2$ defined as the following:

\begin{align*}
    N_K\coloneqq &\bigg\{(p,q)\in [0,1]^2: p\in \bigg[\frac{1-\ell}{2}+k\Delta,\frac{1-\ell}{2}+(k+1)\Delta\bigg),q\in \bigg(1-3\ell,1\bigg]\bigg\}\\
    &\bigcup \quad \bigg\{(p,q)\in [0,1]^2: p\in \bigg[0,3\ell\bigg), q\in \bigg[\frac{1-\ell}{2}+k\Delta,\frac{1-\ell}{2}+(k+1)\Delta\bigg)\bigg\}
\end{align*}

Then it is immediate to see that the unperturbed distribution $\mu_0$ and the $k$-th perturbed probability distribution $\mu_k$ generates undistinguishable feedback in all the space $[0,1]^2\backslash N_K$.

\begin{lemma}
    For all $(p,q)\in [0,1]^2\backslash\bigcup_{k'\in \{1,\ldots,N-1\}}N_{k'}$ it holds 
    \[\mathbb{P}_0[(\mathbb{I}\{s_t\le p\},\mathbb{I}\{q\le b_t\})=z]= \mathbb{P}_k[(\mathbb{I}\{s_t\le p\},\mathbb{I}\{q\le b_t\})=z] \quad \forall k\in \{1,\ldots,N-1\},\forall z\in \{0,1\}^2\]
\end{lemma}

\subsubsection{Revenue Analysis}

In the previous section, following a similar approach to prior work, we studied the decision space \([0,1]^2\) according to two key criteria: the expected gain from trade (GFT) and the informativeness of the feedback. Based on this analysis, we identified two relevant subsets of the space: (i) the grid \(\{M_{i,j}\}_{i,j \in \{0, \ldots, N-1\}}\), which offers the highest expected GFT, and (ii) the exploration regions \(\{N_k\}_{k=1}^{N-1}\), which provide the most informative feedback for distinguishing between instances.

However, unlike previous works, in our setting this characterization alone does not rule out the possibility that the algorithm might benefit from choosing actions outside these two regions.

Specifically, we must also consider a third metric: the expected revenue. One might speculate that playing in the upper-left triangle of the space, defined by the condition \(q \ge p\), could allow the algorithm to accumulate revenue that could then potentially be used to support the selection of more prices from the grid \(\{M_{i,j}\}\) characterized by negative revenue, without incurring a large budget violation.

In this section, we show that—given the way we have constructed our instances—playing prices in the \emph{GFT-dominated region},
\[
\left([0,\tfrac{1-\ell}{2}) \times [0,1]\right) \cup \left([0,1] \times (\tfrac{1+\ell}{2},1]\right), \quad \text{with } q \ge p,
\]
as a way to accumulate revenue, is a strictly suboptimal strategy.

In particular, we prove that under the unperturbed instance \(0\), any algorithm that plays a combination of two prices—one from below the diagonal in the exploitation grid \(\{M_{i,j}\}\), and another from the positive-revenue but GFT-dominated region—such that the mixture is budget-balanced in expectation, can always be matched or outperformed in terms of expected GFT by an algorithm that plays exclusively on the strongly budget-balanced points \(\{M_{i,i}\}\) along the diagonal.

To do so, first we will prove the following auxiliary result which relates the ratio between the difference in terms of GFT and the revenue of all possible combination of two points: one in the exploitation unbalanced grid and one in the GFT dominated positive revenue area.

\begin{restatable}{lemma}{ratiogftrev}
    \label{lemma: ratio gft rev}
    Under the unperturbed instance characterized probability distribution $\mu_0$, $\forall (p,q)\in \left([0,\tfrac{1-\ell}{2}) \times [0,1]\right) \cup \left([0,1] \times (\tfrac{1+\ell}{2},1]\right)$ with $ q \ge p$, for all $(i,j)\in \{0,\ldots,N-1\}\times \{0,\ldots,N-1\}$ with $i>j$ it holds
    \[\frac{\mathbb{E}_0[\gft(M_{i,i})]-\mathbb{E}_0[\gft(p,q)]}{\mathbb{E}_0[\gft(M_{i,j})]-\mathbb{E}_0[\gft(M_{ii})]}\ge \frac{\mathbb{E}_0[\rev(p,q)]}{-\mathbb{E}_0[\rev(M_{i,j})]}.\]
\end{restatable}
The analogous lemma can be formulated also for the perturbed instances $\mu_k$ with $k\in [N-1]$. Indeed, by construction 
\[\mathbb{E}_k[\gft(M_{i,j})]-\mathbb{E}_k[\gft(M_{ii})]\le \mathbb{E}_0[\gft(M_{i,j})]-\mathbb{E}_0[\gft(M_{ii})] + 3\ell \epsilon.\]
\begin{restatable}{lemma}{ratiogftrevk}
\label{lemma: ratio gft rev k}
    Set $k\in [N-1]$, then under the perturbed instance characterized probability distribution $\mu_k$, $\forall (p,q)\in \left([0,\tfrac{1-\ell}{2}) \times [0,1]\right) \cup \left([0,1] \times (\tfrac{1+\ell}{2},1]\right)$ with $ q \ge p$, for all $(i,j)\in \{0,\ldots,N\}\times \{0,\ldots,N\}$ with $i>j$ it holds
    \[\frac{\mathbb{E}_k[\gft(M_{i,i})]-\mathbb{E}_k[\gft(p,q)]}{\mathbb{E}_k[\gft(M_{i,j})]-\mathbb{E}_k[\gft(M_{ii})]}\ge \frac{\mathbb{E}_k[\rev(p,q)]}{-\mathbb{E}_k[\rev(M_{i,j})]}.\]
\end{restatable}

In practice, the above lemmas can be interpreted as establishing a conversion rate between the loss in potential GFT and the gain in revenue budget. This trade-off is such that it is never advantageous to select a suboptimal pair of prices outside the exploitation grid \(\{M_{i,j}\}\) solely for the purpose of accumulating revenue.

To conclude this section we will prove how \Cref{lemma: ratio gft rev} and \Cref{lemma: ratio gft rev k} actually implies the desired result.

\begin{lemma}
    Set $k\in \{0,\ldots,N-1\}$. Under the $k$-th instance characterized by the probability distribution $\mu_k$, $\forall (i,j)\in \{0,\ldots,N-1\}^2$ and $\forall (p,q)\in \left([0,\tfrac{1-\ell}{2}) \times [0,1]\right) \cup \left([0,1] \times (\tfrac{1+\ell}{2},1]\right)$ with $ q \ge p$, for all $O,M\in [T]$ such that $O+M\le T$ and $O\mathbb{E}_k\left[\rev(p,q)\right]+M\mathbb{E}_k\left[\rev(M_{i,j})\right]\ge 0$ it holds:
    \[O\mathbb{E}_k\left[\gft(p,q)\right]+M\mathbb{E}_k\left[\gft(M_{i,j})\right]\le (M+O)\mathbb{E}_k[\gft(M_{i,i})].\]
\end{lemma}
\begin{proof}
    Suppose by asbusrd that it exists $(p,q), M,O, M_{i,j}$ such that 
    \[O\mathbb{E}_k\left[\gft(p,q)\right]+M\mathbb{E}_k\left[\gft(M_{i,j})\right]> (M+O)\mathbb{E}_k[\gft(M_{i,i})].\]
    This can be written also as
    \[\frac{M}{O}\ge \frac{\left(\mathbb{E}_k[\gft(M_{i,i})]-\mathbb{E}_k[\gft(p,q)]\right)}{\left(\mathbb{E}_k[\gft(M_{i,j})]-\mathbb{E}_k[\gft(M_{i,i})]\right)}\]
    and by assumption 
    \[O\mathbb{E}\left[\rev(p,q)\right]+M\mathbb{E}\left[\rev(M_{i,j})\right]\ge 0\]
    which implies 
    \[\frac{M}{O}\le \frac{\mathbb{E}_k[\rev(p,q)]}{-\mathbb{E}_k[\rev(M_{i,j})]}.\]
    Hence, to prove the contradiction it is sufficient the statement of \Cref{lemma: ratio gft rev} if $k=0$ and \Cref{lemma: ratio gft rev k} otherwise.
\end{proof}
This lemma is equivalent to stating that any algorithm that plays \(M\) times a pair of prices with negative revenue (e.g., from below the diagonal) and \(O\) times a pair of prices with positive revenue in the GFT-dominated region—such that the overall combination is globally budget-balanced—can be matched or outperformed in terms of expected GFT by an algorithm that instead plays \(M + O\) times a strongly budget-balanced price (i.e.,  \(\{M_{i,i}\}\)).

\subsubsection{Space Division}

To conclude the section, here we will use the analysis we made in this section to define some random variables, which will be helpful for the rest of the analysis, in particular we will define the random variables :
\begin{itemize}
    \item $\{\mathcal{M}_{i,j}\}_{i,j\in\{0,\ldots,N-1\}}$ where $ \mathcal{M}_{i,j}$ represents how many time the algorithm has chosen the exploitation point $M_{i,j} \quad \forall i,j\in \{0,\ldots,N-1\}$
    \item $\{\mathcal{N}_k\}_{k\in \{0,\ldots,N-1\}}$  where $ \mathcal{N}_k $ represents the number of time the algorithm has chose the exploration area $ N_k \quad \forall k\in \{1,\ldots,N-1\}$
    \item $\mathcal{O}$ that represents how many time the algorithm has chosen a point in  $\mathcal{G}^\cW $ not included in the two cases above.
\end{itemize}

Notice that by the arguments presented before  we can restrict the analysis without loss of generality to algorithms such that $\sum_{i,j\in \{0,\ldots,N-1\}}\mathcal{M}_{i,j}+\sum_{k\in\{1,\ldots,N-1\}}\mathcal{N}_k = T$.

\subsection{Regret Analysis}

In this section we will compute the regret of an algorithm in each instance $k\in \{0,\ldots,N-1\}$, using the random variables $\{\mathcal{M}_{i,j}\}_{i,j\in\{0,\ldots,N\}},\{\mathcal{N}_k\}_{k\in \{1,\ldots,N-1\}}$.

First we compute the expected gain for trade of an algorithm in each perturbed instances $k\in\{1,\ldots,N-1\}$ as a function of the expected value of the above random variables.
\begin{lemma}
For each $k\in \{1,\ldots,N-1\}$ it holds
\begin{align*}
    \sum_{t=1}^T\mathbb{E}_k\bigg[\gft(p_t,q_t)\bigg]& \le \mathbb{E}_k\bigg[c\sum_{i,j\in \{1,\ldots,N-1\}}\mathcal{M}_{i,j}+(\gamma_1(1-2\ell))\sum_{i,j\in \{1,\ldots,N-1\}}(i-j)\mathcal{M}_{i,j}&
    \\ &\hspace{1.5cm}+d\sum_{j\in \{1,\ldots,N-1\}}\mathcal{N}_j+3\ell \epsilon \sum_{i,j\in \{1,\ldots,N-1\}}(\mathbb{I}\{i= k\}+\mathbb{I}\{j=k\})\mathcal{M}_{i,j} \bigg],
\end{align*}
    with $d=(1-2\ell)(N+1)\gamma_1+\gamma_6$.
\end{lemma}

Additionally, we know the following: 
\begin{lemma} It holds:
  \[  \max_{p\in [0,1]}\mathbb{E}_k \left[\gft_t(p,p)\right]= \mathbb{E}_k\left[\gft_t(M_{k,k})\right]= c+ 6\ell \epsilon  . \]
\end{lemma}
Therefore, we can relate the pseudo-regret of an algorithm in each instance to $k\in \{0,\ldots,N-1\}$, using the random variables $\{\mathcal{M}_{i,j}\}_{i,j\in\{0,\ldots,N-1\}},\{\mathcal{N}_k\}_{k\in \{0,\ldots,N-1\}}$, in the following way.
\begin{lemma}
\label{lemma: lb reg decomp}
For all perturbed instances $k\in \{1,\ldots,N-1\}$:
    \begin{align*}
        \sum_{t=1}^T\left(\mathbb{E}_k\left[\gft(M_{k,k})\right]-\mathbb{E}_k\left[\gft(p_t,q_t)\right]\right) &\ge \mathbb{E}_k\bigg[(3\ell\epsilon)\left(2T-\sum_{i,j\in \{1,\ldots,N-1\}}\left(\mathbb{I}\{j= k\}+\mathbb{I}\{k= i\}\right)\mathcal{M}_{i,j}\right)\\
        & \hspace{1.5cm}+(\gamma_5\ell+\gamma_1(1-2\ell))\sum_{j\in \{1,\ldots,N-1\}}\mathcal{N}_j \\
        & \hspace{1.5cm}-(\gamma_1(1-2\ell))\sum_{i,j\in \{1,\ldots,N-1\}}(i-j)\mathcal{M}_{i,j}\bigg]
    \end{align*}
\end{lemma}
Intuitively, the three terms composing the pseudo-regret can be interpreted as follows: the first term corresponds to the regret incurred by exploiting without correctly identifying the underlying instance; the second term represents the exploration cost; and the third term captures the advantage gained over the baseline by violating the budget balance constraint.

And specifically to address this third component we define the expected budget violation of an algorithm that play only on the exploitation grid $\{M_{i,j}\}$ and in the exploration areas $\{N_k\}$.

\begin{lemma}
\label{lemma: lb VT}
For every instance $k\in \{0,\ldots,N-1\}$, it holds
    \begin{align*}
        \mathbb{E}_k[\mathcal{V}_T]& \ge \sum_{i,j\in \{0,\ldots, N\}}(i-j)\Delta\left(\gamma_5+\gamma_6+2\gamma_1(N+2)\right)\mathbb{E}_k[\mathcal{M}_{i,j}]-\sum_{j\in \{0,\ldots,N-1\}}((N+1)\gamma_1+\gamma_6)\mathbb{E}_k[\mathcal{N}_j]
    \end{align*}
\end{lemma}
\begin{proof}
    Consider the point $M_{i,j}$ with $i,j\in \{0,\ldots,N\}$. Then 
    \[\mathbb{E}_k[\rev(M_{i,j})]=(j-i)\Delta(\gamma_5+2\gamma_1(N+2)+(i-j))+\gamma_6),\] 
    which can be upper bounded in the following way 
    \[\mathbb{E}_k[\rev(M_{i,j})]\ge(j-i)\Delta(\gamma_5+2\gamma_1(N+1)+\gamma_6)\] 
    since:
    \begin{itemize}
        \item If $i>j$ then $(j-i)\Delta<0$ and $(\gamma_5+2\gamma_1(N+2+(i-j))+\gamma_6)> (\gamma_5+2\gamma_1(N+2)+\gamma_6)$
        \item If $i<j$ then $(j-i)\Delta>0$ and $(\gamma_5+2\gamma_1(N+2+(i-j))+\gamma_6)< (\gamma_5+2\gamma_1(N+2)+\gamma_6)$.
    \end{itemize}

    In addition, the revenue of a point $(p,q)\in N_j$ for all $j\in \{0,\ldots,N-1\}$  can always be upper bounded by 

    \[\mathbb{E}_k[\rev(p,q)]\le ((N+1)\gamma_1+\gamma_6).\]
\end{proof}

Therefore we can rewrite the regret as an explicit function of the violation $\mathcal{V}_T$.

\begin{lemma}
\label{lemma: lb reg decomp v2}
For all instances $k\in \{1,\ldots,N-1\}$  
    \begin{align*}
        \sum_{t=1}^T\left(\mathbb{E}_k\left[\gft(M_{k,k})\right]-\mathbb{E}_k\left[\gft(p_t,q_t)\right]\right) &\ge \mathbb{E}_k\bigg[(3\ell\epsilon)\left(2T-\sum_{i,j\in \{1,\ldots,N-1\}}\left(\mathbb{I}\{j= k\}+\mathbb{I}\{k= i\}\right)\mathcal{M}_{i,j}\right)\\
        & \hspace{1.5cm}+\frac{1}{32}\sum_{j=1}^{N-1}\mathcal{N}_j-3 g\cdot \mathcal{V}_T^k\bigg]
    \end{align*}

\end{lemma}

\begin{proof}
    By \Cref{lemma: lb VT} we have that the third part of the decomposition in \Cref{lemma: lb reg decomp} can be written through the violation as
    \begin{align*}
        -\gamma_1(1-2\ell)\sum_{i,j\in \{0,\ldots,N\}}(i-j)\mathbb{E}_k[\mathcal{M}_{i,j}]& \ge  - \frac{\gamma_1(1-2\ell)}{\Delta(\gamma_5+\gamma_6+2\gamma_1(N+2))} \cdot \mathcal{V}_T \\
        & \quad - \frac{((N+1)\gamma_1+\gamma_6)\gamma_1(1-2\ell)}{\Delta(\gamma_5+\gamma_6+2\gamma_1(N+2))}\cdot \sum_{j\in \{0,\ldots,N\}}\mathbb{E}_k[\mathcal{N}_j]\\
        & \ge -3g \mathcal{V}_T^k-\frac{g}{2}\sum_{j=1}^{N-1}\mathbb{E}_k[\mathcal{N}_j].
    \end{align*}
    Indeed, by definition $(N+1)\gamma_1+\gamma_6= \frac{\gamma_5}{4}=\frac{1}{8}$, and $\Delta=\ell/N$, and $g\le \frac{1}{16}$ for all $\beta\in [\frac{3}{4},\frac{6}{7}]$.
\end{proof}

To conclude this part, we analyze the regret of the unperturbed instance.
\begin{lemma} It holds
    \begin{align*}
        \mathcal{R}_T^0 \ge \frac{1}{32}\sum_{k=1}^N \mathbb{E}_0[\mathcal{N}_k]-3g\cdot \mathcal{V}_T^0.
    \end{align*}
\end{lemma}
\begin{proof}
It is easy to check that
    \begin{align*}
        \mathcal{R}_T^0 &\ge \gamma_5\ell \sum_{k=1}^{N-1}\mathbb{E}_0[\mathcal{N}_k]-\left(\mathcal{V}_T^0+\sum_{k=1}^{N-1}\mathbb{E}_0[\mathcal{N}_k](\gamma_6+(N+1)\gamma_1)\right)\frac{(1-2\ell)\gamma_1}{\Delta\left(\gamma_5+\gamma_6+2(N+2)\gamma_1\right)}\\
        & \ge \sum_{k=1}^{N-1}\mathbb{E}_0[\mathcal{N}_k]\left(\frac{1}{32}\right)-3g \mathcal{V}_T^0.
    \end{align*}
\end{proof}

\subsection{Relating the Instances}

Similar to \cite{bernasconi2024no} we can use KL decomposition to relate the behavior of an algorithm in different instances.
\begin{lemma}
\label{lemma: KL decomp}
    For every $k\in\{1\ldots,N-1\}$ and for all $i,j\in\{0,\ldots,N\}$ it holds: 
    \[\mathbb{E}_k\left[\sum_{i,j}(\mathbb{I}\{j= k\}+\mathbb{I}\{k = i\})\mathcal{M}_{i,j}\right]-\mathbb{E}_0\left[\sum_{i, j}(\mathbb{I}\{j= k\}+\mathbb{I}\{k = i\})\mathcal{M}_{i,j}\right]\le T\epsilon \sqrt{\frac{2}{\gamma_6}\mathbb{E}_0[\mathcal{N}_k]}\]
\end{lemma}

\begin{proof}
First we define $\{Z_t\}_{t\in T}$ the feedback observed at each round .
    Then we can observe the following result:
    \begin{align*}
        \mathbb{E}_k&\left[\sum_{i,j}(\mathbb{I}\{j= k\}+\mathbb{I}\{k = i\})\mathcal{M}_{i,j}\right]-\mathbb{E}_0\left[\sum_{i,j}(\mathbb{I}\{j= k\}+\mathbb{I}\{k = i\})\mathcal{M}_{i,j}\right]= \\
        & \le  \sum_{t\in \mathcal{T}^M_L}\lVert\mathbb{P}^k_{Z_0,\ldots,Z_{t-1}}-\mathbb{P}^0_{Z_0,\ldots,Z_{t-1}}\rVert_{TV}
    \end{align*}
    By Pilsken inequality 
    \[\lVert\mathbb{P}^k_{Z_0,\ldots,Z_{t-1}}-\mathbb{P}^0_{Z_0,\ldots,Z_{t-1}}\rVert_{TV}\le \sqrt{\frac{1}{2}KL(\mathbb{P}^0,\mathbb{P}^k)}\]
    
    Similar to \cite{bernasconi2024no} by KL decomposition (Lemma 15.1 in \citep{cesa2023repeated})
    \begin{align*}
    KL(\mathbb{P}_0,\mathbb{P}_k)= \mathbb{E}_0[\mathcal{N}_k]\cdot KL(\mathcal{H}_0,\mathcal{H}_k),
    \end{align*}
    where $\mathcal{H}_0,\mathcal{H}_k$ are the feedback distributions.
    \begin{align*}
        \mathcal{H}_k(z)= \begin{cases}
             (k+1)\gamma_1 + \gamma_6 + \epsilon \mathbb{I}(k\neq0)  \quad &\text{if }z=(1,1)\\
             \gamma_6 + \gamma_5 + 2(N+1)\gamma_1 + (k+1)\gamma_1 - \epsilon \mathbb{I}(k\neq0) & \text{if } z=(0,1)\\
             \gamma_6 + (N-2-k)\gamma_1 - \epsilon \mathbb{I}(k\neq0) &  \text{if } z=(1,0)\\
             \gamma_6  + (N-2-k)\gamma_1 + \epsilon \mathbb{I}(k\neq0) &  \text{if } z=(0,0)
        \end{cases}
    \end{align*}
    Hence, by upperbounding KL with $\chi^2$-distance we get
    \begin{equation*}
        KL(\mathcal{H}_0,\mathcal{H}_k) \le \chi^2(\mathcal{H}_0,\mathcal{H}_k) = \sum_{z\in \{0,1\}^2}\frac{(\mathcal{H}_0(z)-\mathcal{H}_k(z))^2}{\mathcal{H}_0(z)}\le 4\epsilon^2\frac{1}{\gamma_6}.
    \end{equation*}
\end{proof}

\Cref{lemma: KL decomp} implies that the behavior of a fixed algorithm when run on two different instances—namely, the unperturbed instance $0$ and the $k$-th perturbed instance characterized by the distribution $\mu_k$—can differ only by a quantity proportional to the square root of the expected number of times the algorithm selects prices within the region $N_k$ under $\mu_0$. This quantity is denoted by $\mathcal{N}_k$. Intuitively, this follows from the fact that the only way to distinguish between instance $0$ and instance $k$ is to choose prices in $N_k$, the region where the distributions differ.

\begin{lemma} It holds
\label{lemma: lb R_T0 V0}
    \begin{align*}
        &\frac{1}{N-1}\sum_{k=1}^{N-1}\mathbb{E}_k\left[\sum_{i, j}(\mathbb{I}\{j= k\}+\mathbb{I}\{k = i\})\mathcal{M}_{i,j}\right]\\
        & \hspace{4cm}\le  \frac{1}{N-1}2T+ T\epsilon \sqrt{\frac{2}{\gamma_6}\frac{1}{N-1}\sum_{k=1}^{N-1}\mathbb{E}_0[\mathcal{N}_k]}.
    \end{align*}
\end{lemma}

\subsection{Putting Everything Together}

Combining all the previous components we finally get that:

\begin{theorem}
    For each algorithm that guarantees $\mathcal{V}_T\le T^\beta$, there is an instance such that 
    \[\mathcal{R}_T\ge  \min\bigg\{\frac{1}{2048}\frac{N}{\epsilon^2}-3gT^\beta,\frac{3}{32}\epsilon T-3gT^\beta\bigg\}\]
\end{theorem}
\begin{proof}
    We have that  
    \[\frac{1}{N-1}\sum_{k=1}^{N-1}\mathcal{R}_T^k  \ge 3\ell T \epsilon \left(\frac{1}{2}-\epsilon\sqrt{\frac{2}{\gamma_6}\frac{1}{N-1}\sum_{k=1}^{N-1}\mathbb{E}_0[\mathcal{N}_k]}\right)-3g T^\beta\]
    with $N\ge2$.

    We can distinguish two cases:
    \begin{itemize}
        \item if $\epsilon\sqrt{\frac{2}{\gamma_6}\frac{1}{N-1}\sum_{k=1}^{N-1}\mathbb{E}_0[\mathcal{N}_k]}<\frac{1}{4}$ then 
        \[\frac{1}{N-1}\sum_{k=1}^{N-1}\mathcal{R}_T^k\ge \frac{3}{4}\ell T \epsilon-3g\cdot T^\beta\]
        \item if $\epsilon\sqrt{\frac{2}{\gamma_6}\frac{1}{N-1}\sum_{k=1}^{N-1}\mathbb{E}_0[\mathcal{N}_k]}\ge\frac{1}{4}$ then 
        \[\mathcal{R}_T^0\ge \frac{1}{32} \sum_{k=1}^{N-1}\mathbb{E}_0[\mathcal{N}_k]-3 g \mathcal{V}_T^0\ge \frac{1}{32}\gamma_6 (N-1) \left(\frac{1}{4 \epsilon }\right)^2 -3 g T^\beta.\]
    \end{itemize}
    
\end{proof}

Finally, setting $\epsilon=\frac{1}{3}\gamma_1=\frac{g}{12(N+1)}$, with $g=\frac{1}{24}T^{1-\frac{4}{3}\beta}$ and $N= \frac{1}{200}T^{1-\beta}$, we conclude the proof.




\newpage
\printbibliography

\newpage
\appendix

\section*{Appendix}
\section{Omitted Proofs from \Cref{sec:stochastic}: Stochastic Valuations} \label{app:stoc}

\begin{algorithm}[H]
    \begin{algorithmic}[1] 
        \Function{ \texttt{Prob.Est}}{$(p,q),L,\nu$}
        \State $p^1\gets 0,p^2 \gets 0,p^3\gets 0, p^4\gets0$
        \For{$t=1,\ldots,L$}
        \State Play $(p,q)$
        \State $p^1 \gets p^1 + \mathbb{I}(s_t\le p, q \le b_t)$
        \EndFor
        \State $p^1 \gets \frac{1}{L}p^1$
        \For{$t=L+1,\ldots,2L$}
        \State Play $(q,q)$
        \State $p^2 \gets p^2 + \mathbb{I}(s_t\le q, q \le b_t)$
        \EndFor
        \State $p^2 \gets \frac{1}{L}p^2$
        \For{$t=2L+1,\ldots,3L$}
        \State Play $(p,p)$
        \State $p^3 \gets p^3 + \mathbb{I}(s_t\le p, p \le b_t)$
        \EndFor
        \State $p^3 \gets \frac{1}{L}p^3$
        \For{$t=3L+1,\ldots,4L$}
        \State Play $(q,p)$
        \State $p^4 \gets p^4 + \mathbb{I}(s_t\le q, p \le b_t)$
        \EndFor
        \State $p^4 \gets \frac{1}{L}p^4$
        \State $\hat{p}\gets p^1-p^2-p^3+p^4$
        \State $\underline{p}=(p^1-p^2-p^3+p^4)-4 \sqrt{\frac{\ln(\frac{4}{\delta})}{2L}}$
        \State\Return $\underline{p}{(p,q)}$
        
        \EndFunction
    \end{algorithmic}
\end{algorithm}

\probest*

    \begin{proof}
    We focus on the definition of the probability we are interested in $\mathbb{P}(q\le s_t \le p,q\le b_t \le p)$, which is well defined and nontrivial whenever $p\ge q$. This probability can be seen as the decomposition of four different probability easier to estimate:

    \begin{equation*}
        \mathbb{P}(q\le s_t \le p,q\le b_t \le p)= \mathbb{P}(s_t\le p,q\le b_t)-\mathbb{P}(s_t\le q,q\le b_t)-\mathbb{P}(s_t \le p, b_t\ge p)+ \mathbb{P}(s_t \le q, b_t\ge p).
    \end{equation*}
    It is immediate to see that $p^1,p^2,p^3,p^4$ are unbiased estimators of the above probabilities.
    Hence, applying Hoeffding inequality with probability at least $1-\nu/4$ it holds
    \begin{align*}
        \lvert p^1- \mathbb{P}(s_t\le p,q\le b_t)\rvert \le \sqrt{\frac{\ln(\frac{4}{\nu})}{2L}},
    \end{align*}
    and analogously the following inequalities hold, each with probability at least $1-\delta/4$
    \[\lvert p^2- \mathbb{P}(s_t\le q,q\le b_t)\rvert \le \sqrt{\frac{\ln(\frac{4}{\nu})}{2L}}\]
    \[\lvert p^3- \mathbb{P}(s_t\le p,p\le b_t)\rvert \le \sqrt{\frac{\ln(\frac{4}{\nu})}{2L}}\]
    \[\lvert p^4- \mathbb{P}(s_t\le q,p\le b_t)\rvert \le \sqrt{\frac{\ln(\frac{4}{\nu})}{2L}}.\]
    Finally, consider the joint event through an union bound we get that with probability at least $1-\nu$
    \begin{align*}
        \lvert\mathbb{P}(q\le s_t \le p,q\le b_t \le p)-(p^1-p^2-p^3+p^4)|&  \le 4 \sqrt{\frac{\ln(\frac{4}{\nu})}{2L}},
    \end{align*}
   which concludes the proof, given the definition of $\underline{p}$ as $\underline{p}=(p^1-p^2-p^3+p^4)-4 \sqrt{\frac{\ln(\frac{4}{\nu})}{2L}}$.
\end{proof}

\stochgftest*

\begin{proof}
    By  \Cref{lemma: unbiased est}, $\widehat{GFT}(p,q)$ is the empirical mean of $T_0$ unbiased estimators of the true value of $\gft(p,q)$, each bounded in the interval $[-3,3]$ by construction. Therefore it is sufficient to apply Hoeffding inequality obtaining that with probability at least $1-\delta$
    \[\lvert\gft(p,q)-\widehat{\gft}(p,q)\rvert\le 3 \sqrt{\frac{\ln({\frac{2}{\delta})}}{2T_0}}.\]
\end{proof}

\begin{algorithm}[H]
    \begin{algorithmic}[1]
    \label{alg: gft expl}
        \Function{ $\gft$-Est.rep}{$(p,q) \in [0,1]^2,T_0\in \mathbb{N}$}
        \For{$t=1,\ldots,T_0$}
        \State Sample $D_t\sim \mathcal{U}(\{0,1,2\})$
        \If{$D_t=0$}
        \State Sample $\widetilde{p}\sim \mathcal{U}([0,p])$
        \State Play $(\widetilde{p},q)$
        \State $\widehat{\gft}(p,q) \gets \widehat{\gft}(p,q) + 3 (p)\mathbb{I}(s_t\le \widetilde{p},q\le b_t)$
        \ElsIf{$D=1$}
        \State Sample $\widetilde{q}\sim \mathcal{U}([q,1])$
        \State Play $(p ,\widetilde{q})$
        \State $\widehat{\gft}(p,q) \gets \widehat{\gft}(p,q) + 3 (1-q)\mathbb{I}(s_t\le p , \widetilde{q}\le b_t)$
        \ElsIf{$D=2$}
        \State Play $(p,q)$
        \State $\widehat{\gft}(p,q) \gets \widehat{\gft}(p,q) +3 (q-p)\mathbb{I}(s_t\le p , q\le b_t)$
        \EndIf
        \EndFor
        \State $\widehat{\gft}(p,q)\gets \frac{1}{T_0}\widehat{\gft}(p,q)$
        
        \State\Return $\widehat{\gft}(p,q)$
        \EndFunction
    \end{algorithmic}
\end{algorithm}

\section{Omitted Proofs from \Cref{sec:adv}: Adversarial Valuations}
\label{app: adv}
\begin{algorithm}[H]\label{funct: ind est}
\small
    \begin{algorithmic}[1]
        \Function{ \texttt{Ind.Est}}{$(p,q) $}
        \State $i \gets \log_2(\frac{1}{K(p-q)})$
        \State $p^1\gets 0,p^2 \gets 0,p^3\gets 0, p^4\gets0$
        \State Sample $D\in \mathcal{U}_{\{0,1,2,3\}}$ 
        \If{$D=0$}
        \State Play $(p,q)$
        \State $p^1 \gets  \mathbb{I}(s_t\le p, q \le b_t)$
        \ElsIf{$D=1$}
        \State Play $(q,q)$
        \State $p^2 \gets  \mathbb{I}(s_t\le q, q \le b_t)$
        \ElsIf{$D=2$}
        \State Play $(p,p)$
        \State $p^3 \gets  \mathbb{I}(s_t\le p, p \le b_t)$
        \ElsIf{$D=3$}
        \State Play $(q,p)$
        \State $p^4 \gets \mathbb{I}(s_t\le q, p \le b_t)$
        \EndIf
        \State $\hat{\mathbb{I}}_{j(t)}\gets p^1-p^2-p^3+p^4$
        \State\Return $\hat{\mathbb{I}}_{j(t)}$
        
        \EndFunction
    \end{algorithmic}
\end{algorithm}
\begin{algorithm}[H]
\small
    \begin{algorithmic}[1]
        \Function{ $\gft$.est}{$(p,q) \in [0,1]^2$}
        \State Sample $D\sim \mathcal{U}(\{0,1,2\})$
        \If{$D=0$}
        \State Sample $\widetilde{p}\sim \mathcal{U}([0,p])$
        \State Play $(\widetilde{p},q)$
        \State $\widehat{\gft}(p,q) \gets  3 (p)\mathbb{I}(s_t\le \widetilde{p},q\le b_t)$
        \ElsIf{$D=1$}
        \State Sample $\widetilde{q}\sim \mathcal{U}([q,1])$
        \State Play $(p ,\widetilde{q})$
        \State $\widehat{\gft}(p,q) \gets  3 (1-q)\mathbb{I}(s_t\le p , \widetilde{q}\le b_t)$
        \ElsIf{$D=2$}
        \State Play $(p,q)$
        \State $\widehat{\gft}(p,q) \gets 3 (q-p)\mathbb{I}(s_t\le p , q\le b_t)$
        \EndIf
        
        \State\Return $\widehat{\gft}(p,q)$
        \EndFunction
    \end{algorithmic}
\end{algorithm}

\gftEstAdv*

\begin{proof}
\begin{align*}
    \mathbb{E}[\widehat{\gft}_j(p,q)]& = \sum_{t\in \mathcal{B}_j}\mathbb{P}(f_j^{-1}(p,q)=t)\bigg(\mathbb{P}(D=0)\mathbb{E}_{\tilde{p}\sim\mathcal{U}([0,p])}\left[3p\mathbb{I}(s_t\le \tilde{p},q\le b_t)\right] \\
    &\hspace{5cm}+ \mathbb{P}(D=1)\mathbb{E}_{\tilde{q}\sim\mathcal{U}([q,1])}\left[3(1-q)\mathbb{I}(s_t\le p,\tilde{q}\le b_t)\right] \\
    & \hspace{5cm}+ \mathbb{P}(D=2)\mathbb{E}\left[3(q-p)\mathbb{I}(s_t\le p,q\le b_t)\right]\bigg)\\
    & = \frac{1}{\mathcal{B}_j}\sum_{t\in \mathcal{B}_j}\bigg(\frac{1}{3}\left(3p\frac{p-s_t}{p}\right)\mathbb{I}(s_t\le p,q\le b_t) + \frac{1}{3}\left(3(1-q)\frac{b_t-q}{1-q}\right)\mathbb{I}(s_t\le p,q\le b_t)\\
    & \hspace{3cm}+\frac{1}{3}\left(3(q-p)\right)\mathbb{I}(s_t\le p,q\le b_t)\bigg)\\
    & = \frac{1}{|\mathcal{B}_j|}\sum_{t\in \mathcal{B}_j}(b_t-s_t)\mathbb{I}(s_t\le p,q\le b_t)=\gft_j(p,q)
\end{align*}    
\end{proof}

To prove \Cref{lemma: n hat adv} it is useful to first prove the following lemma, which is the analogous of \Cref{lemma: grid stoch naive bound} for the adversarial case.
\begin{lemma}
 If \Cref{alg: adv} is initialized with $T,\beta, \delta$, then 
 \[\bigg|\bigcup_{j=1}^N \mathcal{F}_j\bigg|\le \frac{2}{\alpha}\]
\end{lemma}

\lemmaProbAdv*
\begin{proof}
    For each $j\in [N]$ and $(p,q)\in \mathcal{F}_j$, through Hoeffding inequality, it holds that with probability at least $1-\frac{\delta}{T}$ the following:
    \[\bigg|\hat n_j(p,q)-\sum_{j=0}^k \mathbb{I}((p,q)\in \mathcal{F}_j)\sum_{t\in \mathcal{B}_j}\mathbb{I}(q\le b_t < s_t \le p)\bigg|\le 4\sqrt{\frac{\ln\left(\frac{2T}{\delta}\right)}{2}N},\]
    thanks to \Cref{cor: adv ind est}.
\end{proof}
\LemmaGridBoundAdv*

\begin{proof}
    
    
    Define $\mathcal{A}_i$ as the set that contains all and only the $(p,q)\in \bigcup_{j=1}^{N}\mathcal{F}_j$ such that $(p-q)=\frac{1}{K2^i}$. 
    Then, by \Cref{lemma: n hat adv} with probability at least $1-\delta$
    \begin{align*}
        \sum_{(p,q)\in \mathcal{A}_i}\underline{n}_{N}(p,q)& \le  \sum_{(p,q)\in \mathcal{A}_i}\frac{N}{T}\sum_{t\in [T]}\mathbb{I}((p,q)\in \mathcal{F}_{j(t)})\mathbb{I}(q\le b_t \le p, q \le  s_t \le p) \le N,
    \end{align*}
    where the last inequality holds since $\{(q\le b_t \le p, q \le  s_t \le p)\}_{(p,q)\in \mathcal{A}_i}$ are a set of all disjoint events.
    Thus, if we call $r_i$ the number of prices $(p,q)$ in $\mathcal{A}_i$ such that there exists a block $j$ in which the condition $\underline{n}_j(p,q)> \alpha K 2^i$ is met then $|\mathcal{A}_{i+1}|= 2 r_i$ by construction, and with probability at least $1-\delta$
    \begin{align*}
        N \ge \sum_{(p,q)\in \mathcal{A}_i}\max_{j\in [N]}\underline{n}_{j}(p,q) \ge r_i \alpha K 2^i,
    \end{align*}
    and therefore
    \begin{equation*}
        r_i \le \frac{N}{\alpha K 2^i}.
    \end{equation*}
    Finally, considering that by construction $|\mathcal{A}_0|=K$ we have that with probability at least $1-\delta$
    \begin{align*}
        \mathcal{N}&= \bigg|\bigcup_{i=0}^{+\infty}\mathcal{F}_i\bigg|= |\mathcal{A}_0|+ \sum_{i=1}^{+\infty}|\mathcal{A}_i|\\
        & \le K + \sum_{i=1}^{+\infty}2\frac{N}{\alpha K 2^{i-1}}\\
        & \le K + 4\frac{N}{\alpha K} .
    \end{align*}
\end{proof}

\section{Omitted Proofs \Cref{sec:lowerbound}: Lower Bound}

\probwelldef*

\begin{proof}
    We start by analyzing the unperturbed $\mu_0$. 
    It is easy to check that $\sum_{w\in \cW}\mu_0(w)=1$.
    Hence, we are left to show that $\mu_0(w)\ge0$ for all $w\in \cW$.
    Thus, we observe :
    \begin{align*}
    &\mathbb{\mu}_0(w_1^i)=\mu_0(w_3^{N-i})=\gamma_1\left(1+\frac{2i}{3(N)}\right)\ge \gamma_1 \ge 0 && \forall i\in \{0,\ldots,N\}\\
    &\mathbb{\mu}_0(w_2^i)=\mu_0(w_4^{N-i})= \gamma_1\left(1-\frac{2i}{3(N)}\right)\ge \gamma_1\frac{1}{3}\ge 0 && \forall i\in \{0,\ldots,N\}\\
    &\mu_0(w_5)=\gamma_5\ge 0,
    \end{align*}
    and finally we have that 
    \begin{align*}
    \mu_0(\omega^i_6)&=\frac{1}{4}\left(1-\sum_{j\in\{1,2,3,4\}}\sum_{i\in \{0,\ldots,N\}}\mu_0(w_j^i)-\gamma_5\right)\\
    &=\frac{1}{4}\left(1-4(N+1)\gamma_1-\gamma_5\right)\\
    &=\frac{1}{4}\left(1-g-\gamma_5\right)\ge \frac{1}{16} &&\forall i\in \{0,\ldots,N\}
    \end{align*}

    Then we prove the statement for a perturbed instances $\mu_k$,  $k\in\{1,\ldots,N-1\}$. 
    In this case, we can observe that the probabilities $\gamma_5$ and $\gamma_6$ do not change, and all other probabilities are guaranteed to be greater or equal than $0$ if and only if $\epsilon\le \frac{1}{3}\gamma_1$.
\end{proof}

\lemmaGFTZero*

\begin{proof}
It holds
\begin{align*}
    \mathbb{E}_0[\gft(M_{i,j})]& = \gamma_6 + \ell \gamma_5 + \sum_{i'=0}^{i}\gamma_1\left(1+\frac{2i'}{3(N-1)}\right)\left(1-\left(\frac{1-\ell}{2}+i'\Delta\right)\right)\\
    & \quad + \sum_{i'=0}^{i}\gamma_1 \left(1-\frac{2i'}{3(N-1)}\right)\left(1- 3\ell-\left(\frac{1-\ell}{2}+i'\Delta\right)\right)\\
    & \quad+ \sum_{j'=j}^{N}\gamma_1\left(1+\frac{2(N-1-j')}{3(N-1)}\right)\left(\frac{1-\ell}{2}+j'\Delta\right) \\
    & \quad+ \sum_{j'=j}^{N}\gamma_1\left(1-\frac{2(N-1-j')}{3(N-1)}\right)\left(\left(\frac{1-\ell}{2}+j'\Delta\right)-3\ell\right)\\
    & = \gamma_6 + \ell \gamma_5 + \sum_{i'=0}^{i}\gamma_1\left(1-2\ell-2i'\Delta+ \frac{2i'}{(N-1)}\ell\right)\\
    & \quad + \sum_{j'=j}^{N}\gamma_1\left(1-4\ell + 2(N-1-j')\Delta + \frac{2j'}{(N-1)}\ell\right)   \\
    & = \gamma_6 + \ell \gamma_5 + \sum_{i'=0}^{i}\gamma_1\left(1-2\ell-2i'\frac{\ell}{N-1}+ \frac{2i'}{(N-1)}\ell\right)\\
    & \quad + \sum_{j'=j}^{N}\gamma_1\left(1-2\ell - 2j'\frac{\ell}{N-1} + \frac{2j'}{(N-1)}\ell\right)\\
    & = \gamma_6 + \ell \gamma_5 + \sum_{i'=0}^{i}\gamma_1\left(1-2\ell\right) + \sum_{j'=j}^{N}\gamma_1\left(1-2\ell \right)\\
    & = \gamma_6 + \ell \gamma_5 + \gamma_1 (i+1) (1-2\ell) + \gamma_1 (N+1-j)(1-2\ell)\\
    & = \gamma_6 + \ell \gamma_5 + \gamma_1(N+2) (1-2\ell) + \gamma_1(1-2\ell)(i-j).
\end{align*}
\end{proof}

\ratiogftrev*
\begin{proof}
We will prove the lemma by contradiction. Suppose it exist indeed a $M_{i,j}$ and a $(p,q)$ such that 
\[\frac{\mathbb{E}_0[\gft(M_{i,i})]-\mathbb{E}_0[\gft(p,q)]}{\mathbb{E}_0[\gft(M_{i,j})]-\mathbb{E}_0[\gft(M_{ii})]}< \frac{\mathbb{E}_0[\rev(p,q)]}{-\mathbb{E}_0[\rev(M_{i,j})]}.\]

Consider that we know that
    \begin{equation*}
        \mathbb{E}_0[\gft(M_{i,i})]-\mathbb{E}_0[\gft(p,q)]\ge \ell \gamma_5,
    \end{equation*}
    \begin{equation*}
        \mathbb{E}_0[\gft(M_{i,j})]-\mathbb{E}_0[\gft(M_{ii})]\le (i-j)(1-2\ell)\gamma_1 ,
    \end{equation*}
    \begin{equation*}
        \mathbb{E}_0[\rev(p,q)]\le 2(N+1)(\gamma_1)+\gamma_6,
    \end{equation*}
    \begin{equation*}
        -\mathbb{E}_0[\rev(M_{i,j})]\ge (2(N+2)\gamma_1+\gamma_5+\gamma_6)\Delta(i-j).
    \end{equation*}
Hence, it should hold
\[\frac{\ell \gamma_5}{(i-j)(1-2\ell)\gamma_1 }< \frac{2(N+1)(\gamma_1)+\gamma_6}{(2(N+2)\gamma_1+\gamma_5+\gamma_6)\Delta(i-j)},\]
however
\[\frac{\ell \gamma_5}{(1-2\ell)\gamma_1 }> \frac{1}{3}\frac{N+1}{g}\ge 8N,\]
and
\[\frac{2N(\gamma_1)+\gamma_6}{(2(N+1)\gamma_1+\gamma_5+\gamma_6)\Delta}\le4N\]
therefore
\[8N \le \frac{\mathbb{E}[\gft(M_{i,i})]-\mathbb{E}[\gft(p,q)]}{\mathbb{E}[\gft(M_{i,j})]-\mathbb{E}[\gft(M_{ii})]}< \frac{\mathbb{E}[\rev(p,q)]}{-\mathbb{E}[\rev(M_{i,j})]}\le 4N,\]
which is a contradiction.
\end{proof}

\ratiogftrevk*
\begin{proof}
    By contradiction, suppose 
    \[\frac{\mathbb{E}_k[\gft(M_{i,i})]-\mathbb{E}_k[\gft(p,q)]}{\mathbb{E}_k[\gft(M_{i,j})]-\mathbb{E}_k[\gft(M_{ii})]}< \frac{\mathbb{E}_k[\rev(p,q)]}{-\mathbb{E}_k[\rev(M_{i,j})]}.\]

    Then, similarly to \ref{lemma: ratio gft rev} we have that 
    \[\frac{\mathbb{E}_k[\rev(p,q)]}{-\mathbb{E}_k[\rev(M_{i,j})]}<\frac{4N}{(i-j)}\]
    and 
    \begin{align*}\frac{\mathbb{E}_k[\gft(M_{i,i})]-\mathbb{E}_k[\gft(p,q)]}{\mathbb{E}_k[\gft(M_{i,j})]-\mathbb{E}_k[\gft(M_{ii})]}&>\frac{\ell \gamma_5}{(i-j)(1-2\ell)\gamma_1 +3\ell \epsilon}\ge \frac{\ell \gamma_5}{(i-j)(1-2\ell)\gamma_1 +3(i-j)\ell \epsilon}\\
    & \ge \frac{8}{9}\frac{\ell\gamma_5}{(i-j)\gamma_1}= \frac{2N}{g}\ge 8N.
    \end{align*}
    The proof can be conclude analogously to \Cref{lemma: ratio gft rev}.
\end{proof}

\end{document}